\documentclass[prx,twocolumn,showpacs,amsmath,amssymb,superscriptaddress,floatfix,reprint, nofootinbib]{revtex4-2}
\usepackage{amsmath}
\usepackage{amssymb}
\usepackage{amsthm}
\usepackage{amsfonts}
\usepackage{dsfont}
\usepackage{listings}
\usepackage{macros}
\usepackage{enumerate}
\usepackage{latexsym}
\usepackage{pifont}
\usepackage{hyperref}

\usepackage{bbold}

\usepackage{psfrag}
\usepackage{xcolor}
\usepackage{float}

\usepackage{comment}

\usepackage{graphicx}
\usepackage{subfigure}

\begin{document}

\tolerance 10000
\title{Symmetries as Ground States of Local Superoperators: Hydrodynamic Implications}
\author{Sanjay Moudgalya}
\email{sanjay.moudgalya@gmail.com}
\affiliation{Department of Physics and Institute for Quantum Information and Matter,
California Institute of Technology, Pasadena, California 91125, USA}
\affiliation{Walter Burke Institute for Theoretical Physics, California Institute of Technology, Pasadena, California 91125, USA}
\affiliation{Department of Physics, Technische Universit\"{a}t M\"{u}nchen (TUM), James-Franck-Str. 1, 85748 Garching, Germany}
\affiliation{Munich Center for Quantum Science and Technology (MCQST), Schellingstr. 4, 80799 M\"{u}nchen, Germany}
\author{Olexei I. Motrunich}
\email{motrunch@caltech.edu}
\affiliation{Department of Physics and Institute for Quantum Information and Matter,
California Institute of Technology, Pasadena, California 91125, USA}
\begin{abstract}
Symmetry algebras of quantum many-body systems with locality can be understood using commutant algebras, which are defined as algebras of operators that commute with a given set of local operators.
In this work, we show that these symmetry algebras can be expressed as frustration-free ground states of a local superoperator, which we refer to as a ``super-Hamiltonian".
We demonstrate this for conventional symmetries such as $\mbZ_2$, $U(1)$, and $SU(2)$, where the symmetry algebras map to various kinds of ferromagnetic ground states, as well as for unconventional ones that lead to weak ergodicity breaking phenomena of Hilbert space fragmentation and quantum many-body scars.
In addition, we show that the low-energy excitations of this super-Hamiltonian can be understood as approximate symmetries, which in turn are related to slowly relaxing hydrodynamic modes in symmetric systems.
This connection is made precise by relating the super-Hamiltonian to the superoperator that governs the operator relaxation in noisy symmetric Brownian circuits, and this physical interpretation also provides a novel interpretation for Mazur bounds for autocorrelation functions.
We find examples of gapped/gapless super-Hamiltonians indicating the absence/presence of slow-modes, which happens in the presence of discrete/continuous symmetries.
In the gapless cases, we recover hydrodynamic modes such as diffusion, tracer diffusion, and asymptotic scars in the presence of $U(1)$ symmetry, Hilbert space fragmentation, and a tower of quantum scars respectively.
In all, this demonstrates the power of the commutant algebra framework in obtaining a comprehensive understanding of exact symmetries,  and associated approximate symmetries and hydrodynamic modes, and their dynamical consequences in systems with locality.
\end{abstract}
\date{\today}
\maketitle
%


\tableofcontents

\section{Introduction}
\label{sec:intro}
Developing an understanding of symmetries in their most general form has been a recent quest in many different parts of physics.
The definition of symmetries in most of the quantum many-body physics literature implicitly assumes some kinds of restriction imposed on the symmetry operators, e.g., they are usually on-site unitary symmetries with nice group structures, or lattice symmetries such as translation, rotation, reflection, etc.
However, several recent works have demonstrated that generalized symmetries beyond what is usually studied in much of the literature can naturally arise in several physical settings.
In the context of equilibrium physics, several new types of symmetries have recently been studied in the context of various physical lattice models or quantum field theories~\cite{mcgreevy2022generalized, cordova2022snowmass}. 
Examples include ``higher-form symmetries," where the symmetry operators live on manifolds of some non-zero co-dimension~\cite{batista2005generalized, nussinov2009sufficient, nussinov2009symmetry, gaiotto2015generalized}; ``non-invertible" or ``categorical" symmetries, where the symmetries are not representations of groups but of categories~\cite{ji2020categorical, moradi2022topological, cao2023subsystem}; ``MPO symmetries," where the symmetry operators are Matrix Product Operators (MPO)~\cite{aasen2016topological, aasen2020topological, lootens2021dualities}; or even more exotic symmetries that appear in the study of fractons~\cite{nandkishore2019fractons, pretko2020fracton, seiberg2020exotic3d, seiberg2021exoticZn}, where the symmetries depend on the system size.
In the context of non-equilibrium physics, a general framework for symmetries based on ``commutant algebras" has naturally appeared in the study of dynamical phenomena known as weak ergodicity breaking~\cite{serbyn2020review, papic2021review, moudgalya2021review, chandran2022review}.
For example, systems exhibiting Hilbert Space Fragmentation (HSF)~\cite{sala2020fragmentation, khemani2020localization, moudgalya2019thermalization, rakovszky2020statistical} have symmetry algebras that grow exponentially with system size~\cite{moudgalya2021hilbert}, and systems exhibiting Quantum Many-Body Scars (QMBS)~\cite{shiraishi2017systematic, moudgalya2018exact, turner2017quantum, moudgalya2018entanglement, turner2018quantum,  serbyn2020review, moudgalya2021review, chandran2022review} have non-local symmetries such as projectors onto certain pure states~\cite{moudgalya2022exhaustive}.
The discovery of such wide varieties of symmetries motivates the search for certain characterizing properties of symmetry operators that are allowed in physical quantum many-body systems.
A hint comes from the framework where symmetry algebras can be understood as \textit{commutant algebras}, i.e., as the associative algebra of operators that commute with a given set of \textit{local} operators.
In our previous works, we explored this framework in detail, and demonstrated that it can be used to understand regular symmetries and symmetry sectors in a wide variety of standard Hamiltonians~\cite{moudgalya2022from}, as well as to discover novel symmetries that explain the phenomena of Quantum Many-Body Scars~\cite{moudgalya2022exhaustive} and Hilbert Space Fragmentation~\cite{moudgalya2021hilbert}.
Further, in \cite{moudgalya2023numerical}, we introduced numerical methods to calculate commutant algebras.
One such method was based on the property of these operator algebras that when operators are viewed as states in a doubled Hilbert space, these algebras are the ground state spaces of certain frustration-free local superoperators.
Mapping the determination of symmetry algebras to a ground state problem led to efficient algorithms to determine symmetries, which used ideas from tensor network methods for determining ground states in general~\cite{schollwock2011density} as well as specialized methods for determining frustration-free ground states~\cite{debeaudrap2010solving, movassagh2010unfrustrated, yao2022bounding}. 
In this work, we explore the analytical aspects and consequences of the idea that symmetries are ground states of local superoperators, which we refer to as ``super-Hamiltonians.''
This allows us to analytically understand several properties of symmetric systems with locality.
We work out the explicit super-Hamiltonians, which have interpretations in terms of simple ladder or bilayer Hamiltonians, and we solve for their ground states, which precisely map onto commutant algebras.
This allows us to obtain a priori bizarre connections between $\mbZ_2$, $U(1)$, and $SU(2)$ symmetry algebras to ferromagnetic states of various kinds, which can all be expressed as ground states of frustration-free Hamiltonians. 
In addition, we illustrate the commutant algebras in certain unconventional symmetries, including some of the examples of HSF and QMBS discussed in \cite{moudgalya2021hilbert, moudgalya2022exhaustive}. 
In addition to a clear understanding of \textit{exact} symmetries, which are understood as ``white" or ``black" properties of the system, i.e., a given symmetry either exists or it does not, this mapping to ground states also introduces a grey-scale, and provides a precise language for discussing approximate symmetries.
These approximate symmetries are naturally defined as operators that are in the low-energy spectrum of the super-Hamiltonians, of which \textit{exact} symmetries are the ground states.
Since the exact symmetries in many of the examples map onto ferromagnetic states, the low-energy excitations are given by spin waves, which map back onto approximate symmetries.
These are approximately conserved quantities, which can be loosely viewed as long-wavelength textures in the densities of the exactly conserved quantities, and hence are conserved up to times that diverge with the system size (e.g., taking the longest wavelengths fitting into the system).
This feature of approximate symmetries also illustrates their connection to hydrodynamic modes, as we discuss below.
This connection to approximate symmetries is made precise by a remarkable physical relation between the super-Hamiltonian and noisy Brownian circuit models similar to those studied in the context of noisy spin chains~\cite{bauer2017stochastic, bauer2019equilibrium, bernard2021entanglement, bernard2022dynamics, swann2023spacetime}, or as toy models for quantum chaos~\cite{lashkari2013towards,shenker2015stringy,  xu2019locality, zhou2019operator, sunderhauf2019quantum, jian2021note, jian2022linear, agarwal2022emergent, agarwal2023charge, zhang2023information}.
In particular, the low-energy spectrum of the superoperator is related to the relaxation rates of noise-averaged autocorrelation functions towards their Mazur bound values dictated by symmetry~\cite{mazurbound1969, dhar2020revisiting}, which leads to two main insights.
First, it provides an alternate physical meaning to the Mazur bound value, usually interpreted as a lower bound on the autocorrelation function of an operator evolving under a single physical Hamiltonian with a given set of symmetries.
Second, it shows that the \textit{approximate symmetries} that appear as low-energy excitations above the ground state of the super-Hamiltonian correspond to slowly relaxing hydrodynamic modes that govern late-time transport properties in symmetric systems with locality.
For example, we are able to understand the $\sim L^2$ relaxation time in $U(1)$ symmetric systems, which occurs due to the presence of charge or spin diffusion, in terms of spin-wave excitations above ferromagnetic ground states of the superoperator Hamiltonian. 
In addition, we are also able to use this framework to understand hydrodynamic modes for unconventional symmetries such as QMBS, which coincide with slowly thermalizing initial states in such systems, recently referred to as \textit{asymptotic} QMBS~\cite{gotta2023asymptotic}.
With these insights, we connect exact symmetries understood in the commutant algebra framework to approximate symmetries that are related to hydrodynamic relaxation modes and late-time transport, which have been of significant interest lately in systems with various kinds of symmetries~\cite{guardado2020subdiffusion, gromov2020fracton, feldmeier2020anomalous, moudgalya2021spectral, iaconis2019anomalous, sala2022dynamics, ogunnaike2023unifying, morningstar2023hydrodynamics, gliozzi2023hierarchical}.
This paper is organized as follows.
In Sec.~\ref{sec:commutants}, we review key concepts in the study of commutant algebras and their connection to ground states of local super-Hamiltonians.
In Sec.~\ref{sec:examples}, we work out several examples in the context of conventional and unconventional symmetries.
Then in Sec.~\ref{sec:liouvspec}, we illustrate the connection between the low-energy spectrum of the super-Hamiltonians and operator relaxation to Mazur bounds, which can be made concrete in Brownian or noisy circuit models.
We also exhibit the approximate conserved quantities in the context of various kinds of symmetries.
Finally, we conclude with open questions in Sec.~\ref{sec:conclusions}.
\section{Commutant Algebras and Ground States}\label{sec:commutants}
We first review the connection between commutant algebras and frustration-free ground states of local superoperator Hamiltonians, which was first introduced in the context of numerical methods to detect symmetries in \cite{moudgalya2023numerical}.
Here we only review aspects necessary for this work, and more comprehensive discussions can be found in our previous papers~\cite{moudgalya2021hilbert, moudgalya2022from, moudgalya2022exhaustive, moudgalya2023numerical}.
\subsection{Definitions}
The essential idea of commutant algebras is to think of symmetries in terms of a pair of operator algebras $(\mA, \mC)$, referred to as the local algebra and commutant algebra respectively, which are centralizers of each other in the algebra of all operators on the full (finite-dimensional) Hilbert space.
As the name suggests, the \textit{local algebra} $\mA$ is generated by a set of Hermitian local operators $\{\hH_\alpha\}$, which can either be strictly local or extensive local, and we denote it as $\mA = \lgen \{\hH_\alpha\} \rgen$.
In the case when all the operators $\hH_\alpha$ are strictly local, the algebra $\mA$ is also commonly referred to as a \textit{bond algebra}~\cite{nussinov2009bond, cobanera2010unified, cobenera2011bond}. 
The \textit{commutant algebra} $\mC$, by definition, is the centralizer of $\mA$, i.e., the set of all operators that commute with the $\{\hH_\alpha\}$, which is the symmetry algebra for \textit{all} Hamiltonians in $\mA$, i.e., those that can be expressed in terms of linear combinations of products of $\{\hH_\alpha\}$.
For the Hamiltonians constructed out of the generators of $\mA$, symmetry sectors and dynamically disconnected Krylov subspaces due to the symmetry algebra $\mC$ can be understood in terms of their representation theory of von Neumann algebras. 
Thinking of symmetries in this commutant algebra framework leads to a comprehensive understanding of symmetry algebras, symmetric operators, and associated symmetry quantum number sectors, we refer to \cite{moudgalya2021hilbert, moudgalya2022from, moudgalya2022exhaustive} for concrete examples in a variety of systems. 
\subsection{Liouvillians and Super-Hamiltonians}\label{subsec:liouvsuper}
Given the local algebra $\mA$, determining the commutant $\mC$ is not always straightforward in practice.
Hence in \cite{moudgalya2023numerical}, we introduced two numerical methods to numerically construct the full commutant algebra $\mC$ given a set of local terms $\{\hH_\alpha\}$ that generate the local algebra $\mA$.
The method relevant for this work is the ``Liouvillian method," which starts by interpreting operators $\hO$ on the Hilbert space $\mH$ as vectors $\oket{\hO}$.
In particular, operators on the Hilbert space $\mH$, which themselves form a Hilbert space denoted as $\mL(\mH)$, can be mapped onto states on the doubled Hilbert space $\mH \otimes \mH$ via the mapping
\begin{equation}
    \hO = \sumal{\mu, \nu}{}{o_{\mu\nu}\ketbra{v_\mu}{v_\nu}} \iff \oket{\hO} = \sumal{\mu,\nu}{}{o_{\mu\nu}\ket{v_\mu}\otimes\ket{v_\nu}},
\label{eq:opmapping}
\end{equation}
where $\{\ket{v_\mu}\}$ is an orthonormal basis for $\mH$, which we take to be the computational basis of product states. 
For example, for a spin-1/2 system we have
\begin{equation}
    \oket{\mathds{1}}_j = \ket{\up}_j \otimes \ket{\up}_j + \ket{\dn}_j \otimes \ket{\dn}_j, 
\label{eq:idexample}
\end{equation}
where $j$ labels a site [identity operator in a many-body system is $\hat{\mathds{1}} \iff \otimes_j \oket{\mathds{1}}_j$]. 
In this work, we will sometimes interchangeably use $\ket{\bullet}$ and $\oket{\bullet}$ when referring to operators as states on a doubled Hilbert space, and the meaning should be obvious from the context.
Adapting the definition of Eq.~(\ref{eq:opmapping}) together with the conventional inner product in the doubled space implies that the inner product in the operator Hilbert space is defined as\footnote{Note that this differs from the commonly-used convention~\cite{moudgalya2021hilbert} of $\obraket{\hA}{\hB}_{\text{dyn.}} \defn \text{Tr}(\hA^\dagger \hB)/\text{Tr}(\mathds{1})$, which is useful in the study of time-dependent correlation functions.}
\begin{equation}
    \obraket{\hA}{\hB} \defn \text{Tr}(\hA^\dagger \hB).
\label{eq:overlapAB}
\end{equation}
The action of the commutator of an operator $\hO$ with an operator $\hH_\alpha$ can be represented as
\begin{equation}
[\hH_\alpha, \hO]\;\;\iff\;\;\overbrace{\left(\hH_\alpha \otimes \mathds{1} - \mathds{1} \otimes \hH_\alpha^T\right)}^{\hmL_\alpha \defn}\oket{\hO},
\label{eq:commliouv}
\end{equation}
where the transpose is taken in the computational basis. 
In Eq.~(\ref{eq:commliouv}), $\hmL_\alpha$ is referred to as the \textit{Liouvillian} corresponding to the term $\hH_\alpha$, which is the superoperator that represents the adjoint action of that term, i.e., $\hmL_\alpha\opket{\bullet} \defn [\hH_\alpha, \bullet]$.
Given an algebra $\mA = \lgen \{\hH_\alpha\} \rgen$, the operators in the commutant $\mC$ by definition commute with each of the $\hH_\alpha$.
Hence, as vectors in the doubled Hilbert space, they span the common kernel of the Liouvillian superoperators $\{\hmL_\alpha\}$, which is also the nullspace of the positive semi-definite (p.s.d.) superoperator defined as
\begin{equation}
\hmP \defn \sumal{\alpha}{}{\overbrace{\hmL_\alpha^\dagger \hmL_\alpha}^{\hmP_\alpha \defn}},\;\;\;\hmP\oket{\hO} = 0 \iff \hmL_\alpha\oket{\hO} = 0\;\;\forall \alpha, 
\label{eq:psdLiouv}
\end{equation}
where the second condition follows since all the $\hmP_{\alpha}$ are positive semi-definite (p.s.d.).
As discussed in \cite{moudgalya2023numerical}, this provides an efficient numerical method to compute the full commutant $\mC$, given the generators $\{\hH_\alpha\}$.
In addition, in the absence of exact symmetries, the low-energy spectrum of $\hmP$ can be treated as approximate symmetries, as we discuss later in this work.
In App.~\ref{app:extraneous_features} we comment on the dependence of the super-Hamiltonians on the choice of the generators of the bond algebra $\mA$, which does not affect the exact ground states and is not of any concern when using this method to find the commutant $\mC$.
\subsection{Ladder/Bilayer Interpretation}
\label{subsec:ladder_bilayer}
\begin{figure}[t!]
\centering
\includegraphics[scale=0.3]{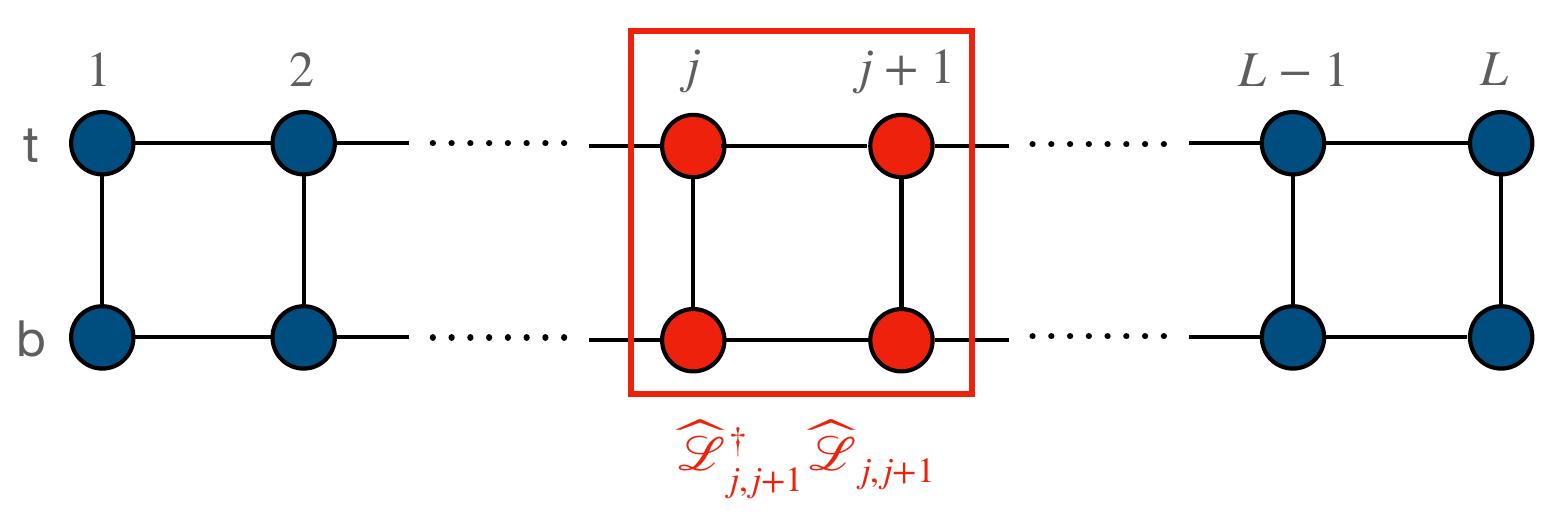}
\caption{The doubled Hilbert space $\mH \otimes \mH$ for one-dimensional systems depicted as a ladder with the two legs labelled by $\{t, b\}$.
The super-Hamiltonian term $\hmL^\dagger_{j,j+1} \hmL_{j,j+1}$ arising from a nearest-neighbor bond algebra term $\hh_{j,j+1}$ is a nearest-neighbor term along the ladder, i.e., it acts non-trivially only on rungs $j$ and $j+1$.}
\label{fig:ladder}
\end{figure}
In order to study the $\hmP$ in typical cases, where $\mH$ is a tensor product Hilbert space with qudits arranged on some lattice, it is convenient to view the Hilbert space $\mH \otimes \mH$ as two copies of the original lattice with the sites aligned, i.e., a site $j$ on the first copy is taken to be neighbor of the site $j$ on the second copy.
This geometry corresponds to a ladder (in one dimension, see Fig.~\ref{fig:ladder}) or a bilayer (in two dimensions), hence we will refer to the first copy of the Hilbert space as the ``top" leg/layer, the second as the ``bottom" leg/layer, and the link between the two legs/layers as ``rungs."
We denote operators on the doubled Hilbert space that act only on one leg/layer using shorthand notations
\begin{equation}
\hO_{\alpha;t} \defn \hO_\alpha \otimes \mathds{1}, \quad
\hO_{\alpha;b} \defn \mathds{1}\otimes \hO_\alpha ~,
\end{equation}
using leg/layer labels $\{t, b\}$ as ``local indices," as often done when writing local Hamiltonians in compact form (we label them using subscripts $\{t, b\}$).
In most examples we study, $\mA$ is a bond algebra, i.e., the operators $\{\hH_\alpha\}$ are strictly local operators, e.g., nearest-neighbor terms of the form $\{\hh_{j,j+1}\}$.
In such cases, the superoperators $\hmP_{\alpha} \defn \hmL_\alpha^\dagger \hmL_\alpha$ are also strictly local operators on the ladder/bilayer with the same range along the ladder/bilayer as the $\{\hH_\alpha\}$, e.g., nearest-neighbor terms $\{\hh_{j,j+1}\}$ give rise to superoperator terms $\{\hmL^\dagger_{j,j+1}\hmL_{j,j+1}\}$ as shown in Fig.~\ref{fig:ladder}.
As a consequence, $\hmP$ of Eq.~(\ref{eq:psdLiouv}) is an extensive local operator on the ladder/bilayer, and has a natural interpretation as superoperator Hamiltonian, which we refer to as a ``super-Hamiltonian".
The definition of the $\hmL_\alpha$ in Eq.~(\ref{eq:commliouv}) can be used to directly obtain an illuminative expression for $\hmP$ of Eq.~(\ref{eq:psdLiouv}) in terms of the $\{\hH_\alpha\}$, which reads
\begin{equation}
    \hmP = \sumal{\alpha}{}{\hmP_\alpha} = \sumal{\alpha}{}{\left(\hH_{\alpha;t}^2 + (\hH_{\alpha;b}^\ast)^2 - 2 \hH_{\alpha;t}\hH_{\alpha;b}^\ast\right)},
\label{eq:Pexp}
\end{equation}
where we have used the fact that the $\hat{H}_\alpha$'s are Hermitian.
This super-Hamiltonian object is the key focus of this work, and in the subsequent sections we will study several examples of this superoperator in various bond and commutant algebras.
Symmetries, which are the operators in the commutant $\mC$ and satisfy $\hmL_\alpha\oket{\cdot}=0$, hence are the {\bf frustration-free ground states} of the local superoperator $\hmP$, since they are ground states of each term $\hmP_\alpha$ of $\hmP$.
The dimension of the commutant, $\dim(\mC)$, is given by the number of independent ground states of $\hmP$.
Finally, we note that this super-Hamiltonian $\hmP$ also has an interpretation as the dissipator of a Lindbladian if we treat $\hH_\alpha$'s as jump operators of a Lindblad master equation, since the action of $\hmP$ on an operator $\oket{\hO}$ reads
\begin{equation}
    \hmP\oket{\hO} \iff -\sumal{\alpha}{}{\left(2\hH_\alpha \hO \hH_\alpha - \{\hH^2_\alpha, \hO\}\right)},
\label{eq:lindbladian}
\end{equation}
which corresponds to the dissipative part of the Lindbladian. 
Indeed, similar mappings are also commonly used in the literature on Lindblad systems~\cite{medvedyeva2016exact, ziolkowska2020yangbaxter}. 
In App.~\ref{app:superHsymms}, we further discuss formal symmetry properties of the super-Hamiltonians viewed as the ladder/bilayer systems; in particular, we encounter ones known as strong symmetries in the Lindblad context and consider how the ground states of the super-Hamiltonians relate to these.
\section{Examples of Super-Hamiltonians}\label{sec:examples}
We now study super-Hamiltonians $\hmP$ in the context of several conventional and unconventional symmetries studied in earlier literature~\cite{moudgalya2021hilbert, moudgalya2022from, moudgalya2022exhaustive}, and show that the respective commutant algebras can be understood as its ground states.
Note that the super-Hamiltonian corresponding to a given symmetry is not unique, and it depends on the choice of generators of the corresponding bond algebra.
The ground states of all such super-Hamiltonians are the same by definition, but the excited states can differ.
We will restrict to a simple choice of local generators of the bond algebras, which lead to simple local super-Hamiltonians, and we expect qualitative features of low-energy excited states to be the same for any other choice of local generators of the bond algebras.
We also restrict examples to one-dimensional systems, and higher dimensional examples proceed in similar ways.
\subsection{Global Symmetry}\label{subsec:globalsymmetry}
We start with the case of global symmetries, studied in \cite{moudgalya2022from}, and we separately show examples of discrete and continuous symmetries.
Note that we only illustrate the Hamiltonians for one-dimensional systems, but the results generalize to higher dimensions as well.
\subsubsection{\texorpdfstring{$\mbZ_2$}{} Symmetry}\label{subsubsec:Z2}
As an example of a discrete symmetry, we focus on $\mathbb{Z}_2$ symmetry in spin-1/2 systems, where we know that the pair of bond and commutant algebras are given by~\cite{moudgalya2022from}
\begin{align}
    \mA_{\mbZ_2} &= \lgen \{X_j X_{j+1}\}, \{Z_j\} \rgen, \nn \\
    \mC_{\mbZ_2} &= \lgen \prod_j{Z_j} \rgen = \text{span}\{\mathds{1}, \prodal{j=1}{L}{Z_j}\}, 
\label{eq:Z2bond}
\end{align}
where $X_j$'s and $Z_j$'s are Pauli matrices on site $j$.
We can use the generators of $\mA_{\mbZ_2}$ to construct the corresponding superoperator $\hmP_{\mbZ_2}$, which when interpreted as a Hamiltonian on a ladder reads, following Eq.~(\ref{eq:Pexp})
\begin{equation}
    \hmP_{\mbZ_2} = 2\sumal{j}{}{[1- Z_{j;t} Z_{j;b}]} + 2 \sumal{j}{}{[1 - X_{j,  t} X_{j;b} X_{j+1;t}  X_{j+1;b}]},
\label{eq:PZ2}
\end{equation}
where $X_{j; \ell}$ and $Z_{j,\ell}$ are now the Pauli matrices on the $\ell$-th leg of the $j$-th site of the ladder system representing the doubled Hilbert space.
Note that all the terms of $\hmP_{\mbZ_2}$ commute among themselves, hence its spectrum is completely solvable!
We can directly solve for the ground states of $\hmP_{\mbZ_2}$ by starting with configurations that ``satisfy," i.e., minimize, the energy under each of the terms individually.
First, we note that ``rung term" $\{1 - Z_{j;t} Z_{j;b}\}$ in $\hmP_{\mbZ_2}$ is satisfied when both the spins on the rung at site $j$ are aligned; hence we can work in the space of composite spins on the rungs, defined as
\begin{equation}
    \sket{\tup} \defn \ladket{\up}{\up},\;\;\;\sket{\tdn} \defn \ladket{\dn}{\dn},
\label{eq:compspins}
\end{equation}
where the top and bottom spins in the ket are states of the top and bottom legs respectively.
Within the $2^L$-dimensional space spanned by all product configurations of the ``composite spins" $\sket{\tup}$ and $\sket{\tdn}$, the action of $\hmP_{\mbZ_2}$ reads
\begin{equation}
    \hmP_{\mbZ_2|\text{comp}} = 2\sumal{j}{}{[1 - \tX_j \tX_{j+1}]}, 
\label{eq:PZ2comp}
\end{equation}
where $\tX_j$ is the composite spin Pauli matrix on site $j$; this is because the action of $X_{j;t} X_{j;b}$ flips the composite spins, and hence can be mapped to $\tilde{X}_j$.
Equation~(\ref{eq:PZ2comp}) is simply the classical Ising ferromagnet, and its two ground states $\ket{G_\rt}$ and $\ket{G_\lt}$ are given by 
\begin{gather}
\ket{G_\rt} = \ket{\trt\ \trt\ \cdots\ \trt\ \trt},\;\;
\ket{G_\lt} = \ket{\tlt\ \tlt\ \cdots\ \tlt\ \tlt},\nn \\
\ket{\trt} \defn \frac{\sket{\tup} + \sket{\tdn}}{\sqrt{2}},\;\;\;
\ket{\tlt} \defn \frac{\sket{\tup} - \sket{\tdn}}{\sqrt{2}}.
\label{eq:PZ2GS}
\end{gather}
In the operator language, the composite spins on rung $j$ map to physical spin projectors on site $j$ as $\sket{\tup}_j = \oket{\ketbra{\up}}_j$ and $\sket{\tdn}_j = \oket{\ketbra{\dn}}_j$.
Hence the composite spins of Eq.~(\ref{eq:PZ2GS}) map as
\begin{equation}
    \ket{\trt} = \frac{1}{\sqrt{2}}\oket{\mathds{1}}_j,\;\;\ket{\tlt} = \frac{1}{\sqrt{2}}\oket{Z}_j,
\label{eq:ladderopcorr}
\end{equation}
and the normalized ground states are
\begin{equation}
    \ket{G_{\rt}} = \frac{1}{2^{\frac{L}{2}}}\oket{\mathds{1}},\;\;\; \ket{G_{\lt}} = \frac{1}{2^{\frac{L}{2}}}\oket{\prod_{j=1}^L{Z_j}},
    \label{eq:GrtGlt}
\end{equation} 
which are precisely the two linearly independent operators that span the commutant algebra for the $\mbZ_2$ symmetry, shown in Eq.~(\ref{eq:Z2bond}). 
Hence the ground state space of the superoperator $\hmP_{\mbZ_2}$ precisely maps to the commutant algebra $\mC_{\mbZ_2}$.
In App.~\ref{subapp:superHsymms_Z2} we further discuss the fate of the formal inherited symmetries of the super-Hamiltonian---here independent $\mbZ_2$ symmetries associated with each leg---in the corresponding quantum ``phase" that contains these ground states, and show that they can be understood in terms of particular spontaneous symmetry breaking.
As an extension, it is easy to see that all Pauli string bond algebras, i.e., those generated by Pauli strings, have super-Hamiltonians that are composed of commuting projectors. 
This is because the Pauli strings themselves square to $1$, making the first two terms in Eq.~(\ref{eq:Pexp}) constants, while the Pauli strings appear in pairs in the last term, which commute with any other pair of Pauli strings.
Hence the ground states of these super-Hamiltonians can be solved to reproduce the respective commutants (which in such cases are also generated by physical Pauli strings~\cite{moudgalya2022from}).
Since the commutants in such cases are also generated by Pauli strings~\cite{moudgalya2022from}, we expect them to correspond to discrete symmetries.
\subsubsection{\texorpdfstring{$U(1)$}{} Symmetry}\label{subsec:U1symmetry}
We next illustrate a continuous symmetry, turning to the commutant of the spin-1/2 $U(1)$ bond algebra, given in one dimension by~\cite{moudgalya2021hilbert, moudgalya2022from}
\begin{align}
    &\mA_{U(1)} = \lgen \{X_j X_{j+1} + Y_j Y_{j+1}\}, \{Z_j\} \rgen, \nn \\
    &\mC_{U(1)} = \lgen \sumal{j = 1}{L}{Z_j} \rgen = \text{span}\{\sumal{j_1 < \cdots < j_m}{}{Z_{j_1} \cdots Z_{j_m}}\}.
\label{eq:U1bond}
\end{align}
We can then use the generators of $\mA_{U(1)}$ to construct the superoperator $\hmP_{U(1)}$ using Eq.~(\ref{eq:Pexp}).
When expressed on the two-leg ladder, after simplification it reads
\begin{gather}
\hmP_{U(1)} = 2\sum_j \left[1 - Z_{j;t} Z_{j;b}\right] + 2 \sum_j \Big[2 - \sum_{\ell \in \{t, b\}} Z_{j,\ell} Z_{j+1,\ell} \nn\\
- (X_{j;t} X_{j+1;t} + Y_{j;t} Y_{j+1;t}) (X_{j;b} X_{j+1;b} + Y_{j;b} Y_{j+1;b})\Big] .
\label{eq:PU1}
\end{gather}
Although $\hmP_{U(1)}$ is not composed of commuting terms, we can solve for their exact ground states. 
Note that similar to the $\mbZ_2$ case, each of the first rung terms proportional to $1 - Z_{j;t} Z_{j;b}$ commutes with all other terms in $\hmP_{U(1)}$ and is satisfied when spins on both legs at site $j$ are aligned, which then justifies working in terms of the composite spins of Eq.~(\ref{eq:compspins}).
It is also easy to see that $\hmP_{U(1)}$ leaves the subspace spanned by the $2^L$ product configurations of the composite spins invariant, and, within that subspace, $\hmP_{U(1)}$ acts as
\begin{gather}
    \hmP_{U(1)|\text{comp}} = 8\sumal{j}{}{(\sket{\tup\tdn} - \sket{\tdn\tup})(\sbra{\tup\tdn} - \sbra{\tdn\tup})}_{[j, j+1]}\nn \\
    =4\sumal{j}{}{[1 - (\tX_j \tX_{j+1} + \tY_j \tY_{j+1} + \tZ_j \tZ_{j+1})]},
\label{eq:PU1comp}
\end{gather}
where $\tX_j$, $\tY_j$, and $\tZ_j$ are the composite spin Pauli operators on site $j$, defined in the obvious way.
Up to an overall factor, this is simply the ferromagnetic Heisenberg model reviewed in App.~\ref{app:Heiscanon}, here in terms of the composite spins.
Its ground state space is hence the $(L+1)$-dimensional ferromagnetic multiplet of the composite spins; these are obtained by replacing the regular spins of the usual ferromagnet shown in Eq.~(\ref{eq:HeisenbergGS}) by the composite spins.
Using the correspondence between states on the ladder and operators of Eq.~(\ref{eq:ladderopcorr}), the states of the composite spin ferromagnetic multiplet translate into operators of the form
\begin{equation}
   \oket{Q^z_{m}} \sim \sumal{j_1 < \cdots < j_m}{}{\oket{Z_{j_1}Z_{j_2} \cdots Z_{j_m}}},
\label{eq:HeisenbergGSop}
\end{equation}
where we have ignored an overall constant.
These are precisely the $L+1$ linearly independent operators forming a basis in the commutant algebra $\mC_{U(1)}$ corresponding to the $U(1)$ symmetry~\cite{moudgalya2021hilbert, moudgalya2022from}, shown in Eq.~(\ref{eq:U1bond}).
While the above operator mapping is more evident in the $\hat{x}$-basis of the composite spins, the same multiplet can be described in terms of the $\hat{z}$-basis of composite spins, analogous to Eq.~(\ref{eq:HeisenbergGS}).
Since the composite spin states $\sket{\tup}_j$ and $\sket{\tdn}_j$ map onto physical spin projectors $\ketbra{\up}_j$ and $\ketbra{\dn}_j$ in the operator language, this $\hat{z}$ basis for the ground state space of $\hmP_{U(1)}$ corresponds to projectors onto the $L+1$ spin sectors labelled by different values of the physical spin $S^z_{\tot}$.
Indeed, for Abelian symmetries, the projectors onto symmetry sectors form an orthogonal basis for the commutant algebra~\cite{moudgalya2021hilbert}. 
Finally, in App.~\ref{subapp:superHsymms_U1} we consider the formal inherited symmetries of the super-Hamiltonian---here independent $U(1)$ symmetries associated with each leg---and show that the ground state manifold can be understood using particular spontaneous symmetry breaking.
\subsubsection{\texorpdfstring{$SU(2)$}{} Symmetry}\label{subsec:SU2symmetry}
As an example of a non-Abelian symmetry, we now illustrate the commutant of the spin-1/2 $SU(2)$ bond algebra, given by
\begin{gather}
    \mA_{SU(2)} = \lgen \{\vec{S}_j \cdot \vec{S}_{j+1} \} \rgen = \lgen \{P^{(2)}_{[j, j+1]}\} \rgen, \nn \\
    P^{(2)}_{[j, j+1]} \defn 2 \vec{S}_j\cdot\vec{S}_{j+1} + \frac{1}{2},\nn \\
    \mC_{SU(2)} = \lgen S^x_{\tot}, S^y_{\tot}, S^z_{\tot} \rgen,
\label{eq:SU2bondalgebra}
\end{gather}
where $P^{(2)}_{[j, j+1]}$ here is the spin-1/2 permutation operator between sites $j$ and $j+1$, i.e.,
\begin{equation}
    P^{(2)}_{[j, j+1]}\ket{\sigma\sigma'}_{[j, j+1]} = \ket{\sigma'\sigma}_{[j, j+1]} ~,
\label{eq:P2exch}
\end{equation}
and $\{S^\alpha_{\tot}\}$ are the total spin operators. 
(This bond algebra contains the Heisenberg Hamiltonian reviewed in App.~\ref{app:Heiscanon}.)
As we will see, the expression in terms of the permutation operators is more convenient for solving the corresponding super-Hamiltonian, and in this form the treatment immediately generalizes to bond algebras for $SU(q)$ symmetry for any $q \geq 2$, which are generated by permutation operators for $q$-level systems, which we denote by $P^{(q)}_{[j, j+1]}$.
In the following, we denote the permutation operators without superscripts to mean that they hold for any $q$, and the $q = 2$ corresponds to the $SU(2)$ case. 
Using $(P_{[j, j+1]})^2 = 1$, the super-Hamiltonian associated with $\{P_{[j, j+1]}\}$ in the ladder representation has the form
\begin{equation}
    \hmP_{SU(q)} = 2\sumal{j}{}{(1 - P_{[j, j+1];t} P_{[j, j+1];b})} = 2\sumal{j}{}{(1 - P^{\text{rung}}_{[j, j+1]}) ~},
\label{eq:suqsuperhamiltonian}
\end{equation}
where $P^{\text{rung}}_{[j, j+1]}$ is the permutation operator for the rungs $j$ and $j+1$, i.e., $P^{\text{rung}}_{[j, j+1]} \ladket{\sigma\ \sigma'}{\tau\ \tau'}_{[j, j+1]} = \ladket{\sigma'\ \sigma}{\tau'\ \tau}_{[j, j+1]}$.
Note that this can also be viewed as the permutation operator $P^{(q^2)}_{[j, j+1]}$ acting on the $q^2$-level systems associated with each of the rungs $j$ and $j+1$.
The permutation operator $P^{(q^2)}_{[j, j+1]}$  possesses an $SU(q^2)$ symmetry, and the super-Hamiltonian $\hmP_{SU(q)}$ is then equivalent to a chain of $q^2$-level systems with nearest-neighbor ferromagnetic $SU(q^2)$-invariant interactions.
The ground states of $\hmP_{SU(q)}$ are the ferromagnetic states of this chain of $q^2$-level systems which can be tabulated as follows.
Given a fixed number of on-site states of each type $N_1, N_2, \dots, N_{q^2}$ with constraint $N_1 + N_2 + \dots + N_{q^2} = L$, we define $\sket{\Psi_{N_1,N_2,\dots,N_{q^2}}}$ to be an equal-weight superposition of all possible configurations with precisely the given number of on-site states of each type.
This is analogous to the spin-1/2 ferromagnetic multiplet that appears as the ground state of the super-Hamiltonian in the $U(1)$ case.
These states are invariant under permutations $P^{(q^2)}_{[j, j+1]}$, and it is easy to prove that they completely span the ground state manifold of $\hmP_{SU(q)}$.
Their count is
\begin{align}
\text{dim}(\mC_{SU(q)}) &= \sum_{N_1, N_2, \dots, N_{q^2} = 0}^L  \delta_{N_1 + N_2 + \dots N_{q^2} = L} \nn \\
&= \binom{L+q^2-1}{q^2-1},
\end{align}
which, for the $q=2$ case, matches precisely the known description of the commutant for the $SU(2)$ symmetry~\cite{moudgalya2021hilbert, moudgalya2022from}.
\subsection{Hilbert Space Fragmentation}\label{subsec:HSF}
We now turn to examples of Hilbert space fragmentation, where the dimension of the commutant scales exponentially with the system size~\cite{moudgalya2021hilbert, andreadakis2022coherence, li2023hilbert}, which leads to exponentially many disconnected Krylov subspaces~\cite{sala2020fragmentation, khemani2020localization, moudgalya2019thermalization,  moudgalya2021review, papic2021review}, which are analogues of quantum number sectors for conventional symmetries.
We start with an example of classical fragmentation, the $t-J_z$ model~\cite{zhang1997tJz, batista2000tJz, batista2001generalizedJW}, which is a model of two species of spins $\up$ and $\dn$, which are allowed to hop but not allowed to cross.
Schematically, there are 3 possible states at any given site: spin $\up$, spin $\dn$, and the vacant site 0, and the allowed ``moves" can be denoted as $\up 0 \leftrightarrow 0 \up$ and $\dn 0 \leftrightarrow 0 \dn$. 
These moves satisfy a conservation of the full pattern of spins (i.e., $\up$'s and $\dn$'s, ignoring the $0$'s) in one dimension, which results in a fragmented Hilbert space with exponentially many Krylov subspaces corresponding to exponentially many allowed patterns~\cite{rakovszky2020statistical, moudgalya2021hilbert}.
In \cite{moudgalya2021hilbert}, we showed that these exponentially many subspaces were attributed to exponentially many conserved quantities in the commutant algebra. 
Specifically, the bond algebra corresponding to the $t-J_z$ model, when viewed as a hard-core bosonic model for simplicity, is given by~\cite{moudgalya2021hilbert}
\begin{gather}
    \mA_{t-J_z} = \lgen \{\hT^{\up}_{[j, j+1]}\}, \{\hT^{\dn}_{[j, j+1]}\}, \{S^z_j\} \rgen,\nn \\
    \hT^{\sigma}_{[j, j+1]} \defn (\ketbra{\sigma\ 0}{0\ \sigma} + h.c.)_{[j, j+1]} ,\nn \\
    S^z_j \defn (\ketbra{\up} - \ketbra{\dn})_j,
\label{eq:tJzbond}
\end{gather}
where $\sigma \in \{\up, \dn\}$, and $\{T^{\sigma}_{[j, j+1]}\}$ are the nearest-neighbor hopping terms for the spin of type $\sigma$.
The corresponding commutant, derived in \cite{moudgalya2021hilbert}, reads
\begin{gather}
    \mC_{t-J_z} = \text{span} \{N^{\sigma_1 \sigma_2 \cdots \sigma_k}, k=0,1,\dots,L \},\nn \\
    N^{\sigma_1 \sigma_2 \cdots \sigma_k} = \sumal{j_1 < j_2< \cdots< j_k}{}{N^{\sigma_1}_{j_1} N^{\sigma_2}_{j_2}\cdots N^{\sigma_k}_{j_k}},\;\;\
    \sigma_j \in \{\up,\dn\},
\label{eq:tJzcomm}
\end{gather}
Note that most of the conserved quantities in the commutant $\mC_{t-J_z}$ are functionally independent from the two obvious $U(1)$ symmetries, $N^\sigma = \sum_j{N^\sigma_j}$, $\sigma \in \{\up, \dn\}$, which are the separate particle number conservations of $\up$ spins and $\dn$ spins.
The commutant $\mC_{t-J_z}$ can be generated by a distinct set of non-local operators, referred to as Statistically Localized Integrals of Motion (SLIOMs)~\cite{rakovszky2020statistical, moudgalya2021hilbert}; however, for our purposes working with the linearly independent basis for $\mC_{t-J_z}$ is more convenient. 
We can construct the super-Hamiltonian using the generators of $\mA_{t-J_z}$ to understand the operators in $\mC_{t-J_z}$ as its ground states.
We first note that the super-Hamiltonian is of the form
\begin{equation}
\hmP_{t-J_z} = \sum_{j}{(S^z_{j;t} - S^z_{j;b})^2} + \sum_j{F(\{\hT^\sigma_{[j, j+1]; \ell}\})},
\label{eq:tJzliouv}
\end{equation}
where $F(\{\hT^\sigma_{[j, j+1]; \ell}\})$ is the positive semi-definite part of the super-Hamiltonian that comes from the generators $\{\hT^\sigma_{[j, j+1]}\}$ of $\mA_{t-J_z}$. 
We now observe that the first sum in Eq.~(\ref{eq:tJzliouv}) enforces that the ground states $\ket{\Psi}$ of $\hmP_{t-J_z}$ satisfy $S^z_{j;t}\ket{\Psi} = S^z_{j;b}\ket{\Psi}$; 
hence the ground states lie in the subspace spanned by product states of composite spins, defined here as 
\begin{equation}    
\sket{\tup} \defn \ladket{\up}{ \up},\;\;\;\sket{\tzero} \defn \ladket{0}{ 0},\;\;\;\sket{\tdn} \defn \ladket{\dn}{ \dn}.
\label{eq:tJzsuperspins}
\end{equation}
It is easy to check that these composite spins are left invariant under the action of $\hmP_{t-J_z}$, and the effective Hamiltonian in this subspace reads:
\begin{equation}
    \hmP_{t-J_z|\text{comp}} = 2 \sum_{j, \sigma \in \{\up, \dn\}}{\left(\sket{\tilde{\sigma}\ \tzero} - \sket{\tzero\ \tilde{\sigma}}\right)\left(\sbra{\tilde{\sigma}\ \tzero} - \sbra{\tzero\ \tilde{\sigma}}\right)_{[j, j+1]}}, 
\label{eq:tJzcomposite}
\end{equation}
where we have used the expression for $\hT^{\sigma}_{[j, j+1]}$ in Eq.~(\ref{eq:tJzbond}). 
This resembles the form of the ferromagnetic Heisenberg Hamiltonian in Eq.~(\ref{eq:PU1comp}), and in fact, as we discuss in Sec.~\ref{subsubsec:fragmentation}, there is a mapping between eigenstates of $\hmP_{t-J_z|\text{comp}}$ and those of the Heisenberg model discussed in App.~\ref{app:Heiscanon}.
For the purposes of the ground states, it is easy to show that they must have equal amplitudes on \textit{all} product states that are ``connected" by the nearest-neighbor ``hops" of the composite spin $\tilde{\sigma}$.
An orthogonal basis for the ground states is formed by states that are equal weight superpositions of all spin configurations with a fixed pattern of composite spins $\tup$ and $\tdn$.
For example, the ground states of $\hmP_{t-J_z}$ with OBC read
\begin{equation}
    \ket{G^{\sigma_1 \cdots \sigma_k}} = \sumal{j_1 < j_2 < \cdots < j_k}{}{\sket{\tilde{\sigma}_1(j_1)\ \tilde{\sigma}_2(j_2)\ \cdots\ \tilde{\sigma}_k(j_k)}},
\label{eq:PtJzGS}
\end{equation}
where the notation $\tilde{\sigma}_l(j_l)$ indicates that the composite spin $\tilde{\sigma}_l$ is at site $j_l$, and the remaining sites are implicitly assumed to be $\tzero$. 
We note that this form of the ground state is a consequence of the fact that $\hmP_{t-J_z}$ of Eq.~(\ref{eq:tJzcomposite}) is of the Rokhsar-Kivelson form [also referred to as Stoquastic or Stochastic Matrix Form (SMF) decomposable]~\cite{castelnovo2005rk, bravyi2015monte}. 
Such superoperators often appear in systems with ``classical" symmetries, i.e., where all the symmetry operators are diagonal in the product state basis, and we discuss these connections in App.~\ref{app:RKHamils}.
Note that the symmetry of the action of the super-Hamiltonian $\hmP_{t-J_z|\text{comp}}$ on $\tup$ and $\tdn$ in Eq.~(\ref{eq:tJzcomposite}) is in fact a composite spin $SU(2)$ symmetry.
As a consequence, $\hmP_{t-J_z|\text{comp}}$ has the same form when written in terms of composite spins $\trt$ and $\tlt$ in the $\hat{x}$-basis defined in Eq.~(\ref{eq:PZ2GS}), i.e., when $\tilde{\sigma} \in \{\trt, \tlt\}$ in Eq.~(\ref{eq:tJzcomposite}).
We can map the ground states of $\hmP_{t-J_z}$ to operators by noting that a composite spin configuration $\sket{\tilde{s}}_j$, $s \in \{\up, 0, \dn\}$, maps to on-site projectors $\oket{\ketbra{s}}_j$.
The ground states of $\hmP_{t-J_z}$ then correspond to the projectors onto the Krylov subspaces of the $t-J_z$ model, which are spanned by all configurations with a fixed pattern of spins. 
Thus, the ground state space of $\hmP_{t-J_z}$ is equivalent to the subspace spanned by these projectors, which one can argue is the same as the subspace spanned by the (non-orthogonal) operators $N^{\sigma_1 \sigma_2 \cdots \sigma_k}$, hence it reproduces precisely the commutant algebra $\mC_{t-J_z}$ of Eq.~(\ref{eq:tJzcomm}).
\subsection{Quantum Many-Body Scars}\label{subsec:isolatedQMBS}
We finally analyze commutants of quantum many-body scars (QMBS)~\cite{serbyn2020review, papic2021review, moudgalya2021review, chandran2022review}, using the super-Hamiltonian picture.
In \cite{moudgalya2022exhaustive} we built on several earlier works~\cite{shiraishi2017systematic, pakrouski2020many, pakrouski2021group} to show that QMBS can be defined as singlets of locally generated algebras.
By this we mean that there are local operators such that their common eigenstates are the QMBS. 
A simple example is when the QMBS eigenstates are simultaneous eigenstates of a set of strictly local operators, w.l.o.g., projectors $\{R_{[j]}\}$ acting on few sites neighboring $j$ (with range bounded by some fixed number) that annihilate the QMBS's, as proposed by Shiraishi and Mori in \cite{shiraishi2017systematic}.
We denote the common kernel of $\{R_{[j]}\}$ as
\begin{equation}
    \text{span}\{\ket{\Phi_n}\} \defn \{\ket{\psi}: R_{[j]}\ket{\psi} = 0\;\;\;\forall\;j\}, 
\label{eq:QMBSproperty}
\end{equation}
where $\{\ket{\Phi_n}\}$ is an orthonormal basis for the kernel, and these are the QMBS eigenstates.
Several known examples of QMBS, including embedded MPS states~\cite{shiraishi2017systematic}, towers of states in the spin-1 XY~\cite{schecter2019weak}, Hubbard~\cite{mark2020eta, moudgalya2020eta, pakrouski2020many, pakrouski2021group}, as well as those in the AKLT models~\cite{moudgalya2018exact, moudgalya2018entanglement, odea2020from, rozon2023broken}, can be understood to be of this form. 
For simplicity, we restrict ourselves to spin-1/2 systems.
As discussed in \cite{moudgalya2022exhaustive}, we expect the bond and commutant algebra
\begin{align}
\mA_{\scar} &= \lgen \{ R_{[j]} \sigma_k^\alpha, ~ k \in \Lambda_j, ~ \alpha = 0,x,y,z \}\rgen, \nn \\
\mC_{\scar} &= \lgen \{\ketbra{\Phi_n}{\Phi_m}\} \rgen,
\label{eq:QMBSalgebra}
\end{align}
where $\{\sigma_k^\alpha\}$ are the on-site Pauli matrices with $\sigma^0_k = \mathds{1}_k$, $R_{[j]}$ is a projector acting on few sites near $j$, and the index $k$ runs over a set of sites $\Lambda_j$ that does not have overlap with the support of $R_{[j]}$, but is in the vicinity.
\footnote{\label{ftn:scarbond} In general, the bond algebra might require generators of the form $\{R_{[j]} O_{\text{nb}[j]}\}$ where ``extensions" $O_{\text{nb}[j]}$ go over all operators over a region neighboring $\text{supp}(R_{[j]})$ but not overlapping with it, with the total range of the generators bounded by a fixed number that depends on the specific case.
However, for simplicity we will only study the simplest examples where it is sufficient to consider single-site extensions, where $O_{\text{nb}[j]} \to \{\sigma_k^\alpha, \alpha=0,x,y,z\}$.}
Note that there are many different choices of generators for $\mA_{\scar}$, which then determine the super-Hamiltonian, and we have chosen one set that is convenient for calculations.
The algebra $\mA_{\scar}$ is claimed to be the \textit{exhaustive} algebra of operators with $\{\ket{\Phi_n}\}$ as degenerate eigenstates with eigenvalue 0, i.e., any such operator can be expressed as sums of products of generators of $\mA_{\scar}$.
We will verify this claim using the super-Hamiltonian formalism for specific examples below.
In examples of towers of QMBS, a ``lifting" term~\cite{mark2020eta, odea2020from} can be added to the generators of these algebras to obtain the exhaustive algebra of operators $\{\ket{\Phi_n}\}$ as non-degenerate eigenstates, see \cite{moudgalya2022exhaustive} for more details.
Following Eq.~(\ref{eq:Pexp}), the expression for the super-Hamiltonian $\hmP_{\scar}$ reads
\begin{gather}
 \hmP_{\scar} = \sumal{j, k, \alpha}{}{[R_{[j];t} + R_{[j];b} - 2 R_{[j];t} R_{[j];b} \, \eta^\alpha \sigma_{k;t}^\alpha \sigma_{k;b}^\alpha]}, \nn \\
 = 4\sumal{j, k}{}{(R_{[j];t} - R_{[j];b})^2} + 8\sumal{j,k}{}{R_{[j];t} R_{[j];b} [1 - \ketbra{\iota}{\iota}]_{k}},
\label{eq:RSigm}
\end{gather}
where for simplicity we have assumed $R_{[j]}^T = R_{[j]}$ (equivalently, $R_{[j]}^* = R_{[j]}$) in the computational basis, we have defined $\eta^0 = \eta^x = \eta^z = 1$ and $\eta^y = -1$, and the state
\begin{equation}
    \ket{\iota}_k \defn \frac{1}{\sqrt{2}}\left(\ladket{\up}{ \up}_k + \ladket{\dn}{\dn}_k\right). 
\label{eq:idproj}
\end{equation}
Note that $\ket{\iota}_k$ is identical to the composite spin $\ket{\trt}_k$ of Eq.~(\ref{eq:PZ2GS}), and hence in the operator language it maps onto $\frac{1}{\sqrt{2}}\oket{\mathds{1}}_k$.
Since all the individual terms in Eq.~(\ref{eq:RSigm}) are positive semi-definite, any ground state $\ket{\Psi}$ of $\hmP_{\scar}$ should satisfy
\begin{align}
&(R_{[j];t} - R_{[j];b})^2\ket{\Psi} = 0 \implies  R_{[j];t}\ket{\Psi} = R_{[j];b}\oket{\Psi}\label{eq:Rcond1} \\
&R_{[j];t} R_{[j];b}[1 - \ketbra{\iota}{\iota}]_k\ket{\Psi} = 0\implies R_{[j],\ell} [1 - \ketbra{\iota}{\iota}]_k\ket{\Psi} = 0, 
\label{eq:Rcond2}
\end{align}
where $\ell \in \{t, b\}$, $k \in \Lambda_j$, and we have used Eq.~(\ref{eq:Rcond1}) in the second step of Eq.~(\ref{eq:Rcond2}).
Using these equations, it is easy to see the presence of the following two types of ground states
\begin{equation}
    \ket{G_{m,n}} \defn \ket{\Phi_m}_t\otimes \ket{\Phi_n}_b,\;\;\;\ket{G_{\mathds{1}}} = \bigotimes_k{\ket{\iota}_k}. 
\end{equation}
In the operator language, $\ket{G_{m,n}}$ maps onto $\oket{\ketbra{\Phi_m}{\Phi_n}}$ (using $\ket{\Phi_n}$ with real-valued amplitudes in the computational basis corresponding to the earlier assumption of real-valuedness of $R_{[j]}$'s in this basis), and $\ket{G_{\mathds{1}}}$ is proportional to the  global identity operator $\oket{\mathds{1}}$. 
These operators are all in the commutant $\mC_{\scar}$ shown in Eq.~(\ref{eq:QMBSalgebra}).
\subsubsection{Isolated QMBS}\label{subsubsec:isolatedQMBS}
Proving that these are the \textit{only} ground states of $\hmP_{\scar}$ is more challenging, and we do that in specific cases in App.~\ref{app:QMBSliouv}.
There we consider an example of an \textit{isolated} QMBS~\cite{moudgalya2021review}, where $R_{[j]} \defn R_j = (1 - \sigma_j^z)/2 = \ketbra{\dn}_j$, hence following Eq.~(\ref{eq:QMBSproperty}) the only scar state is 
\begin{equation}
    \ket{\Phi} \defn \ket{\up\ \up\ \dots\ \up}.
\label{eq:isoscar}
\end{equation}
We start with a set of bond generators of the form of Eq.~(\ref{eq:QMBSalgebra}) with $\Lambda_j = \{j-1, j+1\}$ and follow Eq.~(\ref{eq:RSigm}) to construct the super-Hamiltonian $\hmP_{\text{iso}}$ corresponding to this case [explicit expression in Eq.~(\ref{eq:Liouvfullexp})].
As we show in App.~\ref{app:QMBSliouviso}, in the same composite spin subsector of interest, we obtain the spin model
\begin{equation}
\hmP_{\text{iso}|\text{comp}} = 2 \sumal{j}{}{[(1 - \tZ_j) (1 - \tX_{j+1}) + (1 - \tX_j)(1 - \tZ_{j+1})]},
\label{eq:peschelemeryscar}
\end{equation}
where $\{\tX_j, \tZ_j\}$ are the composite spin Pauli operators also used in Eq.~(\ref{eq:PU1comp}).
As we show in App.~\ref{app:QMBSliouviso}, this Hamiltonian can be mapped to a frustration-free model that lies on the so-called Peschel-Emery line in the vicinity of a transverse field Ising model~\cite{peschel1981calculation,katsura2015exact, mahyaeh2018exact}.
From these studies, this is known to possess only two ground states both in PBC and OBC, which in the composite spin language read
\begin{align}
    &\ket{G_{\mathds{1}}} = \ket{\trt\ \trt\ \cdots\ \trt\ \trt} = \frac{1}{2^{\frac{L}{2}}}\oket{\mathds{1}},\nn \\
    &\ket{G_{\text{QMBS}}} = \sket{\tup\  \tup\ \cdots\ \tup\ \tup} = \oket{\ketbra{\Phi}}.
\label{eq:isoscarGS}
\end{align}
This proves the above claim and hence show the existence of the following pairs of algebras of the form of Eq.~(\ref{eq:QMBSalgebra}) that are centralizers of each other:
\begin{equation}
    \mA_{\text{iso}} \defn \lgen \{ R_j \sigma^\alpha_{j+1}, \sigma^\alpha_j R_{j+1}\} \rgen,\;\;\mC_{\text{iso}} \defn \lgen \ketbra{\Phi} \rgen. 
\label{eq:ACiso}
\end{equation}
Any operator constructed out of the generators of $\mA_{\text{iso}}$ contains the state $\ket{\Phi}$ as an eigenstate, which can be a QMBS if it is in the bulk of the spectrum~\cite{moudgalya2022exhaustive}.
Finally, in App.~\ref{subapp:superHsymms_isoQMBS} we consider relation between the ground state manifold and the inherited formal symmetries of the super-Hamiltonian in this case.
\subsubsection{Tower of QMBS}\label{subsubsec:towerofQMBS}
Similarly, as an example of a tower of QMBS, we consider $R_{[j]} \defn R_{j, j+1} = \frac{1}{4} - \vec{S}_j \cdot \vec{S}_{j+1} = \frac{1}{2}(\ket{\up\dn} - \ket{\dn\up})(\bra{\up\dn} - \bra{\dn\up})_{j,j+1}$, where the common kernel contains the entire ferromagnetic tower of states as scars:
\footnote{Note that these states are by definition the ground states of the ferromagnetic Heisenberg model discussed in App.~\ref{app:Heiscanon}, hence we have $\ket{\Phi_{n,0}} = \ket{F^z_n}$ of Eq.~(\ref{eq:HeisenbergGS}).}
\begin{equation}
	\ket{\Phi_{n,0}} \defn \frac{1}{n!} \binom{L}{n}^{-\frac{1}{2}}(S^{z-}_{\tot})^n\ket{\up\ \up\ \dots\ \up},\;\;\;0 \leq n \leq L, 
\label{eq:FMtower}
\end{equation}
where $S^{z-}_{\tot}$ is the total spin lowering operator in the $\hat{z}$ direction.
We again start with bond generators of the form of Eq.~(\ref{eq:QMBSalgebra}) with $\Lambda_j \in \{j-1, j+2\}$ and follow Eq.~(\ref{eq:RSigm}) to construct the super-Hamiltonian $\hmP_{\text{tower}}$ corresponding to this case [explicit expression in Eq.~(\ref{eq:Liouvfulltower})], and then solve for all of its ground states to show the existence of the following pairs of algebras that are centralizers of each other
\begin{align}
    \mA_{\text{tower}} &\defn \lgen \{ R_{j, j+1} \sigma^\alpha_{j+2}, \sigma^\alpha_{j-1} R_{j,j+1}\} \rgen \nn \\
    \mC_{\text{tower}} &= \lgen \{\ketbra{\Phi_{n,0}}{\Phi_{m,0}}\} \rgen. 
\label{eq:ACtower}
\end{align}
Any operator constructed out of the generators of $\mA_{\text{tower}}$ contains the ferromagnetic multiplet as degenerate QMBS~\cite{moudgalya2022exhaustive}.
These results also generalize to algebras corresponding to the case where the QMBS are non-degenerate towers. 
However, we discuss this case in App.~\ref{app:QMBStower} due to more care required in the Brownian circuit setup and analysis. 
\section{Super-Hamiltonian Spectrum: Approximate Symmetries and Slow Modes}\label{sec:liouvspec}
While the ground states of the superoperator $\hmP$ correspond to the symmetry operators in the commutant algebra, we now show that the low-energy excitations correspond to operators that are ``approximate" symmetries until long times under local dynamics.  
To make this idea precise, we show that the superoperator $\hmP$ acts as a dissipator for operators in ensemble-averaged noisy Brownian circuits.  
Several types of Brownian circuits have been studied in the literature, e.g., in the context of information scrambling~\cite{lashkari2013towards, shenker2015stringy, xu2019locality}, the SYK model~\cite{sunderhauf2019quantum, jian2021note,jian2022linear,  agarwal2022emergent, zhang2023information}, in quantum generalizations of certain classical processes such as SSEP~\cite{bauer2017stochastic, bauer2019equilibrium, bernard2020entanglement, bernard2022dynamics, swann2023spacetime}, in the context of transport with  symmetries~\cite{agarwal2023charge, ogunnaike2023unifying}; and similar connections between superoperators and ensemble-averaged dissipative dynamics of operators have been noted in some of them. 
\subsection{Algebra-based Brownian Circuits}\label{subsec:algebrabrownian}
Given a bond algebra $\mA = \lgen \{\hH_\alpha\} \rgen$, we consider an associated \textit{Brownian circuit} that consists of time evolution with the Hamiltonian $H = \sum_\alpha{J^{(t)}_\alpha \hH_\alpha}$ for a short time of $\dt$.
At each time step, $\{J^{(t)}_\alpha\}$ are chosen to be uncorrelated random variables from a fixed distribution, and we are eventually interested in ensemble-averaged quantities such as correlation functions.  
Operators under this circuit evolve in the adjoint language as
\begin{align}
\oket{\hO(t + \dt)} &= e^{i \sum_{\alpha}{J_\alpha^{(t)} \hmL_\alpha \dt}} \oket{\hO(t)} \nn \\
&= \oket{\hO(t)} + i \dt \sumal{\alpha}{}{J_\alpha^{(t)} \hmL_\alpha \oket{\hO(t)}} \nn \\
&- \frac{(\dt)^2}{2} \sumal{\alpha,\beta}{}{J_\alpha^{(t)} J_\beta^{(t)} \hmL_\alpha \hmL_\beta \oket{\hO(t)}} + \mO((\dt)^3),
\label{eq:opevolutionliouv}
\end{align}
where $\hmL_\alpha$ is the Liouvillian corresponding to $\hH_\alpha$, defined in Eq.~(\ref{eq:commliouv}).
If $\{J_\alpha^{(t)}\}$ at different times $t$ are chosen to be uncorrelated random variables, and we are only interested in ensemble-averaged quantities that are linear in the operator $\hO(t)$, e.g.,  correlation functions such as $\textrm{Tr}(\hA^\dagger \hO(t) \rho_0)$ for some fixed operator $\hA$ and density matrix $\rho_0$, we can average the operator over the probability distribution of the random variables directly.
Using Gaussian distributions for $J^{(t)}_\alpha$'s: 
\begin{equation}
P(\{J^{(t)}_\alpha\}) \sim \exp(-\sumal{t, \alpha}{}{\frac{(J^{(t)}_\alpha)^2}{2 \sigma^2_\alpha}}),\;\;\sigma^2_\alpha = \frac{2\kappa_\alpha}{\dt},
\label{eq:Jalphapdf}
\end{equation}
where $\sigma^2_\alpha$ is the variance of $J^{(t)}_\alpha$, we obtain
\begin{equation}
	 \oket{\overline{\hO(t + \dt)}} =  \oket{\overline{\hO(t)}} - \dt \sumal{\alpha}{}{\kappa_\alpha \hmL^\dagger_\alpha \hmL_\alpha \oket{\overline{\hO(t)}}} + \mO((\dt)^2),
\label{eq:Oavgfinal}
\end{equation}
where $\overline{\cdots}$ denotes the average over all the $\{J^{(t)}_\alpha\}$ performed independently at all times $t$, and we have used the properties that $\overline{J^{(t)}_\alpha} = 0$, $\overline{J^{(t)}_\alpha J^{(t')}_\beta} = \sigma^2_\alpha \delta_{\alpha,\beta} \delta_{t, t'}$ and also the Hermiticity of the super-Hamiltonian with respect to the Frobenius scalar product of Eq.~(\ref{eq:overlapAB}).
Note that in the continuous time limit when $\dt \rt 0$, the distribution of Eq.~(\ref{eq:Jalphapdf}) is referred to as shot noise~\cite{majumdar2005brownian}, and the Brownian circuits can be understood using the language of stochastic processes and It\^{o} calculus~\cite{grigoriu2002stochastic, schilling2014brownian}, which we will not discuss here. 
For our purposes, it is sufficient to note that in the continuous time limit we obtain
\begin{align}
	\frac{d}{dt} \oket{\overline{\hO(t)}} &= -\sumal{\alpha}{}{\kappa_\alpha \hmL^\dagger_\alpha\hmL_\alpha}\oket{\overline{\hO(t)}},\nn \\
 \implies\;\oket{\overline{\hO(t)}} &= e^{-\kappa \hmP t}\oket{\hO(0)}, 
\label{eq:Oavgcontinuum}
\end{align}
where for simplicity we assume $\kappa_\alpha = \kappa$, we use that $\overline{\hO(0)} = \hO(0)$ since it is independent of all the $J^{(t)}_\alpha$'s, and $\hmP$ is defined in Eqs.~(\ref{eq:psdLiouv}) and (\ref{eq:Pexp}).
Ensemble-averaged quantities linear in the time-evolved operator $\hO(t)$, e.g., two-point correlation functions with fixed second operator, can then be expressed in terms of the ensemble-averaged operator $\oket{\overline{\hO(t)}}$.
We note in passing that ensemble-averages of higher-order functionals of the operator $\hO(t)$, e.g., higher-point correlation functions or R\'{e}nyi entropies, can also be studied using similar methods; these usually involve studying effective Hamiltonians on more copies/replicas of the original Hilbert space, see e.g., \cite{sunderhauf2019quantum, agarwal2022emergent} for discussions of such techniques.
\subsection{Correlation Functions}
With this understanding, ensemble-averaged correlation functions can also be studied using the eigenstates and spectrum of $\hmP$. 
The two-point dynamical correlation functions of operators $\hA$ and $\hB$ at infinite temperature, are defined as\footnote{We remind readers again that the definition of $\obraket{\hB}{\hA}$ differs by a factor of $D$ from the commonly-used definition.}
\begin{equation}
    C_{\hB, \hA}(t) \defn \frac{1}{D}\text{Tr}(\hB(0)^\dagger \hA(t)) = \frac{1}{D}\obraket{\hB(0)}{\hA(t)},
\label{eq:corrdefn}
\end{equation}
where $D \defn \text{Tr}(\mathds{1}) = \dim(\mH)$. 
After ensemble-averaging this can be written in terms of eigenstates of $\hmP$ as
\begin{align}
\overline{C_{\hB, \hA}(t)} &= \frac{1}{D}\obraket{\hB(0)}{\overline{\hA(t) }} = \frac{1}{D}\obra{\hB(0)}e^{-\kappa\hmP t}\oket{\hA(0)}\nn \\
&= \frac{1}{D}\sumal{\mu}{}{\obraket{\hB}{\lambda_\mu}\obraket{\lambda_\mu}{\hA} e^{-\kappa p_\mu t} },\nn \\
&=\frac{1}{D}\sumal{E}{}{e^{- \kappa E t} \sumal{\nu_E = 1}{N_E}{\obraket{\hB}{\lambda_{\nu_E}(E)}\obraket{\lambda_{\nu_E}(E)}{\hA}}}, 
\label{eq:2ptcorrelationavg}
\end{align}
where $\{\oket{\lambda_\mu}\}$ are the orthonormal eigenstates of $\hmP$ with eigenvalues $\{p_\mu\}$, with real $p_\mu \geq 0$ since $\hmP$ is positive semi-definite. 
In the last step we have reorganized the sum in terms of energy eigenvalues of $\hmP$ and their degeneracies, where $\{\oket{\lambda_{\nu_E}(E)}\}$ are eigenstates with eigenvalue $E$, and $N_E$ is the degeneracy at that energy; this form is convenient to work with in examples we study in Sec.~\ref{sec:slowmodes}.
As $t \rt \infty$ for a finite system size $L$, we obtain the equilibrium value of the ensemble-averaged two point correlation function,
\begin{equation}
\overline{C_{\hB, \hA}(\infty)} \defn \lim_{t \rt \infty}{\overline{C_{\hB, \hA}(t)}} = \frac{1}{D}\sumal{\mu}{}{\delta_{p_\mu, 0}\obraket{\hB}{\lambda_\mu}\obraket{\lambda_\mu}{\hA}}.
\label{eq:2ptcorrlimit}
\end{equation}
The information of the late-time transport associated with the symmetry is stored in the nature of the approach to the infinite-time quantity, which for a finite system is given by $\overline{C_{\hB, \hA}(t)} - \overline{C_{\hB, \hA}(\infty)}$.
However, for such purposes we are usually interested in the $L \rt \infty$, in which case we usually have $\overline{C_{\hB, \hA}(\infty) \rt 0}$, and it is sufficient to focus on $\overline{C}_{\hB, \hA}(t)$.
\subsection{Autocorrelation Functions and Mazur Bounds}
With the understanding of correlation functions, we illustrate a novel interpretation for the Mazur bounds of autocorrelation functions, studied extensively in the literature~\cite{mazurbound1969, suzukiequality1971, ilievski2013thermodynamic, prosen2014quasilocal, zadnik2016quasilocal, dhar2020revisiting, rakovszky2020statistical, moudgalya2021hilbert}.
The autocorrelation function of an operator $\hA$ is defined as
\begin{equation}
    C_{\hA}(t) \defn \frac{1}{D}\obraket{\hA(0)}{\hA(t)},
\end{equation}
and using Eq.~(\ref{eq:2ptcorrelationavg}) its ensemble-averaged value can be written as
\begin{align}
	\overline{C_{\hA}(t)} &= \frac{1}{D}\obra{\hA(0)}e^{-\kappa\hmP t}\oket{\hA(0)} = \frac{1}{D}\sumal{\mu}{}{|\obraket{\lambda_\mu}{\hA}|^2 e^{-\kappa p_\mu t} }\nn \\
    &=\sumal{E}{}{e^{- \kappa E t} \underbrace{\frac{1}{D}\sumal{\nu_E = 1}{N_E}{|\obraket{\lambda_{\nu_E}(E)}{\hA}|^2}}_{W_{\hA}(E) \defn}},
\label{eq:correlationavg}
\end{align}
where in the second line we have expressed the sum in terms of energies (possibly degenerate) and corresponding eigenstates of $\hmP$, similar to Eq.~(\ref{eq:2ptcorrelationavg}).
Note that $W_{\hA}(E)$ can also be viewed as the weight of the state $\oket{\hA}$ in the subspace spanned by energy eigenstates $\{\oket{\lambda_{\nu_E}(E)}\}$, i.e., the degenerate subspace of eigenstates with eigenvalue $E$.
This shows that only eigenstates $\oket{\lambda_\mu}$ that have non-zero overlap with the operator $\oket{\hA}$ contribute to $\overline{C_{\hA}(t)}$.
At $t = 0$, this is simply the total weight of the initial operator $\hA(0)$, given by
\begin{equation}
    \overline{C_{\hA}(0)} = \frac{1}{D}\obraket{\hA}{\hA} = \sum_E{W_{\hA}(E)}.
\label{eq:opnorm}
\end{equation}
For operators of interest like a local operator, this starts at a value of $O(1)$, e.g., for a Pauli matrix in a spin-$\frac{1}{2}$ system, we have $\obraket{Z_j}{Z_j}/D = 1$ in our definition of the operator scalar product.
As $t \rt \infty$ for a finite system size $L$, its equilibrium value is
\begin{equation}
\overline{C_{\hA}(\infty)} \defn \lim_{t \rt \infty}{\overline{C_{\hA}(t)}} = \frac{1}{D}\sumal{\mu}{}{\delta_{p_\mu, 0}  
|\obraket{\lambda_\mu}{\hA}|^2}.
\label{eq:autocorrlimit}
\end{equation}
Noting that $\{\oket{\lambda_\mu}\}_{p_\mu = 0}$,  i.e., the ground states of $\hmP$ among the above eigenstates, is an orthonormal basis for the commutant algebra $\mC$, the R.H.S. of Eq.~(\ref{eq:autocorrlimit}) is precisely the Mazur bound~\cite{mazurbound1969, dhar2020revisiting, moudgalya2021hilbert}.
Hence the Mazur bound can also be interpreted as the saturation value of the ensemble-averaged autocorrelation function in Brownian circuits, in addition to the conventional interpretation as a lower bound for the time-averaged autocorrelation function for a static Hamiltonian. 
Note that there is an extra factor of $\frac{1}{D}$ in Eq.~(\ref{eq:autocorrlimit}) due to the different definition of operator overlap in Eq.~(\ref{eq:overlapAB}) from that commonly used in the related literature.
The nature of decay of the autocorrelation to the Mazur bound of Eq.~(\ref{eq:autocorrlimit}) also reveals information about the slowest operators or hydrodynamic modes in the system.
The deviation of the autocorrelation function from the Mazur bound for a finite system is $\overline{C_{\hA}(t)} - \overline{C_{\hA}(\infty)}$.
Unlike the Mazur bound, which is usually computed for finite $L$, in the decay of autocorrelations we are usually interested in the limit $L \rt \infty$ and finite but long times $t$.
Since the Mazur bound $\overline{C_{\hA}(\infty)}$ vanishes in the $L \rt \infty$ limit for all the examples we study, we can still restrict our study to $\overline{C_{\hA}(t)}$.
In this limit, it is clear that the behavior is dominated by the nature of the low-energy excitations of $\hmP$.
However, the precise behavior also depends on their degeneracies as well as the weights of the operator of interest on these eigenstates.
If $\hmP$ is gapped in the thermodynamic limit, with gap $E_{\min} \defn \min_{\mu}\ E_\mu > 0$, using Eqs.~(\ref{eq:opnorm}) and (\ref{eq:correlationavg}), and assuming $\overline{C_{\hA}(\infty) = 0}$ in the thermodynamic limit, we obtain that $\overline{C_{\hA}(t)} \leq \overline{C_{\hA}(0)} \exp(-\kappa E_{\min} t)$, hence it decays exponentially fast, with a rate proportional to the inverse gap. 
When $\hmP$ is gapless, i.e., when $E_{\min} \rt 0$ as $L \rt \infty$, we need to be more careful to derive the form of decay.
Since $\hmP$ is a local superoperator, the low-energy excited states are usually quasiparticles such as spin-waves with dispersion relations of the form $E(k) \sim k^\gamma$.
If the full weight of the operator $\hA$, or at least a majority of it lies within this quasiparticle band of states, we can heuristically write Eq.~(\ref{eq:correlationavg}) as $\overline{C_{\hA}(t)} \sim \int \mathrm{d}k\ e^{-\kappa E(k) t}  \sim t^{-\frac{1}{\gamma}}$.
As we will see with concrete examples in the next section, for many conventional symmetries such as $U(1)$ or $SU(2)$, we obtain $E(k) \sim k^2$, and hence we obtain a power-law decay of the autocorrelation function, i.e., $\overline{C_{\hA}(t)} \sim \frac{1}{\sqrt{t}}$.  
However, note that this argument is not rigorous, and if the weight of the operator does not lie fully in the lowest quasiparticle band, this argument may not lead to the correct form of the late-time $\overline{C_{\hA}(t)}$, and we demonstrate this with an example of Hilbert space fragmentation in the next section. 
Hence, when $\hmP$ is gapless, it is important to carefully study the nature of the operator weight distribution $W_{\hA}(E)$ in Eq.~(\ref{eq:correlationavg}) across the spectrum.
\subsection{Approximate Block-Diagonalization}\label{subsec:approxblock}
While exact symmetries lead to exact block diagonalizations of operators with those symmetries~\cite{moudgalya2021hilbert, moudgalya2022from}, it is natural to ask if approximate symmetries lead to approximate block diagonalizations of operators.
While we have not been able to establish this in complete generality, here we nevertheless make some simple observations in this direction.
Given a super-Hamiltonian $\hmP$ of the form of Eq.~(\ref{eq:psdLiouv}),  and a subspace $\mK$ of the Hilbert space with dimension $D_{\mK}$,  the ``energy" of its projector $\Pi_{\mK}$ under $\hmP$ is a measure of how connected this subspace is to the rest of the Hilbert space under the action of the terms $\{\hH_\alpha\}$. 
To see this,  note that the energy of $\Pi_{\mK}$ is given by
\begin{align}
\varepsilon_{\mK} &\defn \frac{\obra{\Pi_{\mK}} \hmP \oket{\Pi_{\mK}}}{\obraket{\Pi_{\mK}}{\Pi_{\mK}}} \nn \\
&= \frac{2}{D_{\mK}}\sum_{\alpha}{(\text{Tr}[\hH_\alpha^2 \Pi_{\mK} ] - \text{Tr}[\Pi_{\mK} \hH_\alpha \Pi_{\mK} \hH_\alpha])} \nn \\
&= \frac{2}{D_{\mK}}\sum_{\alpha}{\text{Tr}[\Pi_{\mK} \hH_\alpha \Pi_{\mK^\perp} \hH_\alpha]},
\label{eq:subspaceenergy}
\end{align}
where $\Pi_{\mK^\perp}$ is the projector onto the subspace orthogonal to $\mK$.
The last line in Eq.~(\ref{eq:subspaceenergy}) is precisely the sum of the norms of the ``block off-diagonal" parts of the $\hH_\alpha$'s, i.e., the sums of squares of matrix elements between states in $\mK$ and $\mK^\perp$. 
This is consistent with the fact that if $\mK$ is a symmetry sector or a closed Krylov subspace, $\Pi_{\mK}$ has zero energy under $\hmP$, which implies that $\mK$ is completely disconnected from the rest of the Hilbert space. 
Hence the existence of a basis in which all the $\hH_\alpha$'s have an ``approximate block diagonal structure" (in the sense that $\varepsilon_{\mK}$ of Eq.~(\ref{eq:subspaceenergy}) is small) implies the existence of low-energy excitations in the corresponding super-Hamiltonian. 
Likewise, the existence of a \textit{projector} that has a small energy under $\hmP$ implies the existence of a basis in which $\hH_\alpha$'s (and hence the Hamiltonians formed by their linear combinations) have approximate block-diagonal forms. 
However, the existence of general low-energy excitations of the super-Hamiltonian, which is what we show for a number of cases in Sec.~\ref{subsec:gapless}, does not itself guarantee the existence of low-energy projectors in general. 
We defer a careful exploration of this issue for future works. 
\section{Examples of Low-Energy Excitations}\label{sec:slowmodes}
In this section, we construct the low-energy excited states of the super-Hamiltonian $\hmP$ for several of the examples discussed in Sec.~\ref{sec:examples} and discuss corollaries for dynamical properties.
We illustrate examples of gapped and gapless super-Hamiltonians separately, since they lead to qualitatively different physics.
Again, we restrict explicit illustrations to one-dimensional systems, but the results carry over to higher dimensional systems \textit{mutatis mutandis}.
Note that while the low-energy spectra depend on the precise choice of generators of the bond algebras $\mA$ that leads to the super-Hamiltonians $\hmP$, qualitative aspects should not depend on these details despite some choices leading to extraneous features that allow more tractability, as we argue in App.~\ref{app:extraneous_features}.
\subsection{Gapped Super-Hamiltonians}\label{subsec:gapped}
\subsubsection{\texorpdfstring{$\mbZ_2$}{} Symmetry}
We begin by considering example from Sec.~\ref{subsubsec:Z2}.
Since $\hmP_{\mbZ_2}$ of Eq.~(\ref{eq:PZ2}) is a commuting projector Hamiltonian, it is easy to see that it is gapped, since the lowest energy excitations can be constructed only by ``unsatisfying" one of the terms $Z_{j;t} Z_{j;b}$ or $X_{j;t} X_{j;b} X_{j+1;t} X_{j+1;b}$.
For example, one of the lowest excited state can be constructed by ``destroying" a single composite spin, say at the rung $j_0$, by acting the operator $X_{j_0;b}$ on either of the ferromagnetic ground states $\ket{G_{\rt}}$ or $\ket{G_{\lt}}$ of Eq.~(\ref{eq:PZ2GS}).
The resulting excitation excitation ``violates" the $Z_{j_0;t} Z_{j_0;b}$ term, and hence the energy of the excitation from Eq.~(\ref{eq:PZ2}) is $4$.
In the original operator language, these excitations can be written as $\oket{X_{j_0}}$ and $\oket{i Y_{j_0}\prod_{j \neq j_0}{Z_j}}$ respectively (ommitting numerical factors), which commute with all $\{X_j X_{j+1}\}$ terms and all $\{Z_j\}$ in the generators of the bond algebra $\mA_{\mbZ_2}$, except with $Z_{j_0}$.
Hence they can be also viewed as ``local charge insertion" operators for the original $\mbZ_2$ symmetry generated by $\prod_j{Z_j}$.
Another type of excitation can be constructed by violating one of the $X_{j;t} X_{j;b} X_{j+1;t} X_{j+1;b}$ terms, which corresponds to a domain wall between the two ferromagnetic configurations $\ket{G_{\lt}}$ and $\ket{G_{\rt}}$.
For a system with OBC, such excitations have the same energy as the ones in the previous paragraph, whereas for PBC, they have twice the energy since domain walls necessarily appear in pairs in this case.
In the operator language, considering OBC for simplicity, we get operators of the form  $\oket{\prod_{k \leq j_0} Z_{k}}$ and $\oket{\prod_{k > j_0} Z_{k}}$, which can be viewed as operators that create a charge of the ``dual $\mbZ_2$ symmetry" obtained by applying a Kramers-Wannier duality transformation on the original system~\cite{chatterjee2022algebra}.
These are operators that commute with all the generators of $\mA_{\mbZ_2}$ except $X_{j_0} X_{j_0+1}$.
Both types of the low-energy excitations can be viewed as operators that create either a $\mbZ_2$ charge or a dual $\mbZ_2$ charge discussed in \cite{chatterjee2022algebra} (note that their convention has $Z$ and $X$ interchanged compared to our Sec.~\ref{subsubsec:Z2}).
In the ``holographic" view of symmetry in 1d~\cite{chatterjee2022algebra, moradi2022topological}, these charge creation operators in turn correspond to $e$ or $m$ particles of the 2d toric code, and it would be interesting to make further connections between their view of symmetries and our super-Hamiltonian perspective.
Finally, we note that the gapped nature following from the commuting-projector property of $\hmP_{\mbZ_2}$ extends to super-Hamiltonians for all Pauli string algebras, which usually correspond to discrete symmetries.
Hence this feature also carries over to higher dimensional super-Hamiltonians corresponding to such symmetries, which can be interpreted as commuting projector Hamiltonians on a bilayer geometry.
\subsubsection{Isolated QMBS}
\label{subsubsec:gappedisolatedQMBS}
The case of an isolated QMBS, discussed in Sec.~\ref{subsec:isolatedQMBS}, also gives rise to a gapped super-Hamiltonian $\hmP_{\text{iso}}$ corresponding to the algebra $\mA_{\text{iso}}$ of Eq.~(\ref{eq:ACiso}).
This super-Hamiltonian, restricted to the composite spin sector, is shown in Eq.~(\ref{eq:peschelemeryscar}).
An intuition for the gap is that $\hmP_{\text{iso}}$ has exactly two linearly independent ground states shown in Eq.~(\ref{eq:isoscarGS}), and there is no natural ``smooth" low-energy excitation on top of the two ground states. 
While this model is not solvable, it is frustration-free and has been studied in the earlier literature, and has been proven to be gapped with OBC~\cite{katsura2015exact}.
We also numerically find evidence that it is gapped with PBC, see App.~\ref{app:QMBSliouv} for more details.
We expect that similar phenomenology holds for other examples of isolated QMBS, and it would be interesting to prove this in general, perhaps using some of the mathematical physics methods developed for such purposes~\cite{knabe1988energy, nachtergaele1996spectral, gosset2016local, kastoryano2018divide}.
\subsection{Gapless Super-Hamiltonians}\label{subsec:gapless}
We now move on to demonstrate interesting examples of gapless super-Hamiltonians, which lead to slowly relaxing hydrodynamic modes associated with the symmetry. 
\subsubsection{\texorpdfstring{$U(1)$}{} Symmetry}\label{subsec:U1gapless}
We start with the case of $U(1)$ symmetry.
As discussed in Sec.~\ref{subsec:globalsymmetry}, the ground states of the super-Hamiltonian $\hmP_{U(1)}$ of Eq.~(\ref{eq:PU1}) are in the composite-spin sector, obtained by minimizing the energy under the rung terms $\{1 - Z_{j;t} Z_{j;b}\}$, defined on rungs of the ladder in Eq.~(\ref{eq:compspins}).
Since each of the rung terms $\{1 - Z_{j;t} Z_{j;b}\}$ commutes with all the terms in Eq.~(\ref{eq:PU1}), any state not in the composite spin sector has an energy of at least $4$, hence we expect the lowest excited states of $\hmP_{U(1)}$ to be within the composite spin sector.
Within this sector, the effective Hamiltonian maps onto the ferromagnetic Heisenberg model of Eq.~(\ref{eq:PU1comp}), and as discussed in App.~\ref{app:Heiscanon} the lowest energy eigenstates are spin-waves on top of the ferromagnetic multiplet, shown in Eq.~(\ref{eq:Heisenbergspinwavestates}).
The energies of these states are given by $32\sin^2\left(\frac{k}{2}\right)$, where $k$ is quantized as shown in Eqs.~(\ref{eq:HeisenbergPBC}) or (\ref{eq:HeisenbergOBC}), and hence the gap of $\hmP_{U(1)}$ scales as $\sim 1/L^2$, showing that it is gapless.
When mapped to the operator language, using Eq.~(\ref{eq:ladderopcorr}), the spin-wave states of Eq.~(\ref{eq:Heisenbergspinwavestates}) translate to
\begin{equation}
    \oket{\lambda_{m, k}} = \frac{1}{\sqrt{2^L \mM_{m,k}}}\sumal{j_1 < \cdots < j_m}{}{\left(\sumal{\ell = 1}{m}{c_{j_{\ell},k}}\right)\oket{Z_{j_1} \cdots Z_{j_m}}},
\label{eq:Heisenbergspinwavesoperator}
\end{equation}
where $c_{j,k}$ is the form of the orbitals given in Eqs.~(\ref{eq:HeisenbergPBC}) and (\ref{eq:HeisenbergOBC}) for PBC and OBC, and $\mM_{m,k}$ is a normalization factor shown in Eq.~(\ref{eq:spinwavenorms}).
With the exact form of the excited states, we can compute the ensemble-averaged correlation functions of local operators $\overline{C_{\hB, \hA}(t)}$ and $\overline{C_{\hA}(t)}$, shown in Eqs.~(\ref{eq:2ptcorrelationavg}) and (\ref{eq:autocorrlimit}). 
We study correlators of the on-site operators $\{Z_j\}$, which in the composite spin language read
\begin{equation}
\oket{Z_j} = 2^{\frac{L}{2}}\ket{\trt\ \cdots\ \trt\ \tlt_j\ \trt\ \cdots\ \trt}
\label{eq:Zjcomp}
\end{equation}
Using this expression it is clear that all the weight of this operator belongs to the Hilbert space spanned by the one-spin-flip spin-waves $\{\oket{\lambda_{1,k}}\}$ of Eq.~(\ref{eq:Heisenbergspinwavesoperator}), including the $k = 0$ case which belongs to the ground state manifold (see App.~\ref{app:Heiscanon} for more details), and $\oket{Z_j}$ is orthogonal to every other eigenstate of $\hmP_{U(1)}$, both in and outside the composite spin sector.
Hence the full time-evolution of ensemble-averaged correlation functions $\overline{C_{Z_j, Z_{j'}}(t)}$ and $\overline{C_{Z_j}(t)}$ can be computed just using these one-spin-flip spin-waves! 
The overlap of these operators on the single spin-wave eigenstates reads
\begin{equation}
    \obraket{\lambda_{1, k}}{Z_j} = \frac{c_{j,k}^* 2^{L/2}}{\sqrt{\mM_{1, k}}},\;\;\;  \frac{1}{D}|\obraket{\lambda_{1, k}}{Z_j}|^2 = \frac{|c_{j,k}|^2}{\mM_{1,k}} = |c_{j,k}|^2.
\label{eq:MZjU(1)}
\end{equation}
When $k = 0$, the last expression is precisely the Mazur bound, $\overline{C_{Z_j}(\infty)} = 1/L$, which was also computed in \cite{moudgalya2021hilbert}.
From now, we specialize to PBC for simplicity.\footnote{Results in OBC can be obtained similarly and agree with PBC when $j$ is very far from the boundary, but one can in principle study also behavior at any distance from the boundary, which we do not pursue here.}
Using Eqs.~(\ref{eq:correlationavg}) and (\ref{eq:MZjU(1)}), and the PBC parameters discussed in App.~\ref{app:Heiscanon}, the time-dependence of the autocorrelation reads 
\begin{align}
\overline{C_{Z_j}(t)} &= \frac{1}{L}\sumal{k}{}{e^{-32 \kappa \sin^2\left(\frac{k}{2}\right) t}} =  \int_{0}^{2\pi}{\frac{\mathrm{d}k}{2\pi}\ e^{-16 \kappa [1 - \cos(k)] t}} \nn \\
&= e^{-16 \kappa t} \mathcal{I}_0(16 \kappa t) \approx \frac{1}{\sqrt{2\pi \cdot 16\kappa t}}\;\;\text{at large}\;t,
\label{eq:CZjt}
\end{align}
where $k$ in the first sum is quantized as $2\pi n/L$ for $0 \leq n \leq L-1$ and we have taken the $L \rt \infty$ limit to go from the sum to the integral, and $\mathcal{I}_0$ is the modified Bessel function of the first kind.
Note that the late-time dependence can also be easily recovered by simply substituting the ``slow mode" dispersion relation that $\sim k^2$, which directly leads to $\overline{C_{Z_j}(t)} \sim \frac{1}{\sqrt{t}}$, consistent with diffusive systems.
It is also possible to recover the Gaussian spatial spreading nature of the two-point correlation of Eq.~(\ref{eq:2ptcorrelationavg}) which in the $L \rt \infty$ limit leads to the integral
\begin{align}
\overline{C_{Z_j, Z_{j'}}(t)} &= \int_0^{2\pi}{\frac{\mathrm{d}k}{2\pi}\ e^{-16 \kappa [1 - \cos(k)] t} e^{i k (j - j')}} \nn \\
& = e^{16 \kappa t} \mathcal{I}_{j-j'}(16\kappa t) \\ 
& \overset{\kappa t \gg 1}{\approx} \int_{-\infty}^{\infty} \frac{{\mathrm{d}k}}{2\pi}\ e^{-8\kappa k^2 t} e^{i k (j - j')} = \frac{e^{-\frac{(j-j')^2}{32\kappa t}}}{\sqrt{32 \pi \kappa t}},
\label{eq:corrU1}
\end{align}
where in the second line ${\mathcal{I}}_\nu$ is the modified Bessel function of the first kind of index $\nu=j-j'$, while the last line shows behavior for $\kappa t \gg 1$.
The Gaussian nature of the correlation function in Eq.~(\ref{eq:corrU1}), with a variance growing linearly in $t$ is consistent with the prediction of diffusion.
It is easy to see that similar results hold in higher dimensions, and the complete weight of the local spin operator $\oket{Z_j}$ is within the one-spin-flip spin wave band. This allows us to compute the ensemble-averaged correlation functions, recovering the standard results expected from diffusion, e.g., the decay of autocorrelations as $\sim t^{-\frac{d}{2}}$ in $d$ dimensions.
Moreover, as discussed in App.~\ref{app:extraneous_features}, the super-Hamiltonian can be different if one starts with a different set of generators of the bond algebra $\mA_{U(1)}$, and it need not be solvable or $SU(2)$-symmetric.
However, the ground states are always the same by construction,  and on physical grounds we expect the $\sim k^2$ dispersion of the low-energy excitations to be the same as long as the generator set is chosen to be local,  since it still represents a $U(1)$ symmetric Brownian circuit.
Indeed, the gap in longer-range Hamiltonians with this set of ground states was studied numerically in \cite{moudgalya2021spectral} and was shown to be consistent with $\sim 1/L^2$ scaling for a system of size $L$, and the $\sim k^2$ form of low-energy excitations was argued based on mappings to a field theory.
\footnote{Additional evidence for this phenomenon comes from the construction of explicit trial states for low-energy excitations above the ferromagnetic ground states of Eq.~(\ref{eq:HeisenbergGSop}), which exactly map to the states of Eq.~(\ref{eq:HeisenbergGS}). 
Since the ferromagnetic states form towers of QMBS as discussed in Sec.~\ref{subsubsec:towerofQMBS}, the trial states are the so-called asymptotic scars~\cite{gotta2023asymptotic} reviewed in Sec.~\ref{subsec:asymptoticQMBS}, which are states orthogonal to the ferromagnetic multiplet with energy variance $\sim 1/L^2$.
Asymptotic QMBS were explicitly constructed in Ref.~\cite{gotta2023asymptotic} in a particular class of spin-1 XY-like models hosting an exact QMBS tower qualitatively similar to the ferromagnetic multiplet, and the same arguments and trial states hold when the QMBS states are actually the frustration-free ground states rather than in the middle of the spectrum. 
These trial states then provide upper bounds for the nearby excitation energies dispersing as $\sim k^2$.
Alternately, the trial states can also be constructed using standard tools such as the Feynman-Bijl ansatz on top of the ferromagnetic multiplet~\cite{ogunnaike2023unifying}.
}
\subsubsection{Hilbert Space Fragmentation}\label{subsubsec:fragmentation}
We now discuss the low-energy excitations of the super-Hamiltonian $\hmP_{t-J_z}$ of Eq.~(\ref{eq:tJzliouv}) in the case with $t-J_z$ fragmentation, and we show that its low-energy excitations can be used to understand slow modes and late-time behavior of the $t-J_z$ model. 
Due to the conservation of the pattern of spins, the $t-J_z$ model at late-times is expected to exhibit \textit{tracer diffusion} for typical initial states~\cite{barma1994slow, feldmeier2022emergent}, which is the phenomenology exhibited in one-dimension by a single ``tracer" particle that is not allowed to cross its neighbors~\cite{harris1965diffusion, levitt1973dynamics, alexander1978diffusion}. 
This leads to a $\sim t^{-\frac{1}{4}}$ prediction for the nature of decay of spin autocorrelation functions at late-times.\footnote{However, note that in systems exhibiting pattern conservation (or ``irreducible strings"~\cite{dhar1993conservation, menon1997conservation}), the autocorrelation function for different kinds of initial states can show drastically different behaviors~\cite{barma1994slow, menon1995irreducible, barma1997deposition}.} 
To determine the low-energy excitations, it is sufficient to work in the composite spin sector defined by spins of the form of Eq.~(\ref{eq:tJzsuperspins}), since there is necessarily a gap to other sectors due to the first term in Eq.~(\ref{eq:tJzliouv}).
The super-Hamiltonian restricted to composite spins, $\mP_{t-J_z|\text{comp}}$ of Eq.~(\ref{eq:tJzcomposite}), in the ket-bra notation then closely resembles the Heisenberg Hamiltonian $\mP_{U(1)|\text{comp}}$ of Eq.~(\ref{eq:PU1comp}).
In fact, apart from the degeneracies, the spectra of these Hamiltonians are identical. 
To see that, we first note that $\hmP_{t-J_z|\text{comp}}$ acting on states with $n$ ``spinful particles" in OBC preserves the pattern of spins $\tilde{\sigma}_1, \tilde{\sigma}_2, \dots, \tilde{\sigma}_n$ on these particles ordered from left to right.
Working in a sector with such a fixed pattern, the action of $\hmP_{t-J_z|\text{comp}}$ does not differentiate in any way between the spins $\tup$ and $\tdn$ on these particles, and we can simply label the states in this sector by marking each particle location as ``state'' $\tilde{1}$, obtaining a Hilbert space of $L$ qubits $\tilde{1}/\tzero$ with precisely $n$ qubits in state $\tilde{1}$.
The $\tzero$ and $\tone$ can further be mapped onto the Hilbert space of spins $\up$ and $\dn$, and the Hamiltonian $\hmP_{t-J_z|\text{comp}}$ of Eq.~(\ref{eq:tJzcomposite}) after these identifications precisely maps to the ferromagnetic Heisenberg model $H_{\text{Heis}}$ of Eq.~(\ref{eq:heisenbergcanon}), up to an overall factor. 
Once the eigenstates of the Heisenberg model are written in terms of the spins $\up$'s and $\dn$'s, they can first be mapped to $\tzero$'s and $\tone$'s respectively, and then the $\tone$'s can be replaced by the specific pattern of spins in the given sector in the $t-J_z$ model, i.e., $\tup$'s and $\tdn$'s to obtain an eigenstate of $\hmP_{t-J_z|\text{comp}}$, and hence that of $\hmP_{t-J_z}$. 
Hence any eigenstate of the Heisenberg model with $n$ $\up$'s and $(L-n)$ $\dn$'s corresponds to $2^n$ degenerate eigenstates of $\hmP_{t-J_z}$. 
This also maps the $L+1$ ground states of the ferromagnetic Heisenberg model in Eq.~(\ref{eq:HeisenbergGS}) to the total of $\sum_{n=0}^L 2^n = 2^{L+1} - 1$ ground states of $\hmP_{t-J_z|\text{comp}}$ in Eq.~(\ref{eq:PtJzGS}).
Morever, due to the composite spin $SU(2)$ symmetry of $\hmP_{t-J_z|\text{comp}}$ discussed in Sec.~\ref{subsec:HSF}, this mapping also holds in terms of composite spin states $\{\trt, \tlt\}$ defined in Eq.~(\ref{eq:PZ2GS}) instead of $\{\tup, \tdn\}$; this is useful in the discussion in App.~\ref{app:tracerdiffusion}.
The entire list of mappings can be summarized as follows:
\begin{equation}
    (\tup/\tdn, \tzero)/(\trt/\tlt, \tzero) \longleftrightarrow (\tone, \tzero) \longleftrightarrow (\dn, \up)/(\leftarrow, \rightarrow),
\label{eq:tJzHeismapsummary}
\end{equation}
where the leftmost states are in the $t-J_z$ composite spin Hilbert space, and the rightmost ones are in the spin-1/2 Hilbert space.
To understand the behavior of autocorrelation functions, we restrict to a local operator $\hA = S^z_j$, defined in Eq.~(\ref{eq:tJzbond}).
We then need to compute the behavior of $\overline{C_{S^z_j}(t)}$, which according to Eq.~(\ref{eq:correlationavg}) requires the computation of the overlaps between $S^z_j$ and the eigenstates of $\hmP_{t-J_z}$. 
First, since the ground states of $\hmP_{t-J_z}$ are precisely the operators of the commutant $\mC_{t-J_z}$, the total weight of $\oket{S^z_j}$ on all these ground states is the Mazur bound.
This bound was computed exactly for the spin operator $S^z_j$ in \cite{moudgalya2021hilbert}, and it was shown to decay with the system size as $\sim \frac{1}{\sqrt{L}}$ for OBC in the bulk of the chain, and remain $O(1)$ at the boundaries even as $L \rt \infty$. 
Since we are interested in the bulk transport properties, we focus on the behavior of $\overline{C_{S^z_j}(t)}$ at large $t$ for $j$ in the middle of the chain as $L \to \infty$. 

\begin{figure}[t!]
\includegraphics[scale=0.95]{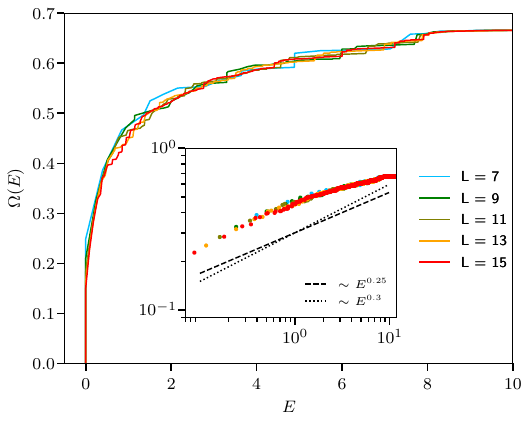}
\caption{(Color online) Cumulative weight $\Omega(E)$ of the operator $\oket{S^z_{j}}$, with $j = \frac{L+1}{2}$ in the middle of an OBC chain, on the eigenstates of the super-Hamiltonian $\hmP_{t-J_z|\text{comp}}$ of Eq.~(\ref{eq:tJzcomposite}) as a function of energy for various system sizes $L$.
Note that $\Omega(E)$ appears to be converging to an $L$-independent function.
At large $E$, it approaches the total weight of the operator, which is $\frac{2}{3}$.
With our choices of overall factors in $\hmP_{t-J_z|\text{comp}}$, the Heisenberg model it maps to has a one-magnon bandwidth of $8$ and a total bandwidth of $\mathcal{O}(2.6L)$, which are the natural energy scales to judge the horizontal axis.
Inset: Same plot in a log-log scale to extract its scaling as $E \rt 0$.
The form appears to be $\Omega(E) \sim E^\gamma$, with $\gamma \in [0.25,0.3]$.}
\label{fig:tJzweight}
\end{figure}

We discuss the computation of the overlap with other low-energy eigenstates of $\hmP_{t-J_z}$ in App.~\ref{app:tracerdiffusion}.
Unlike the case for $U(1)$ symmetry discussed in Sec.~\ref{subsec:U1gapless}, where the weight of the local spin operator was completely within the spin-wave band of excitations of $\hmP_{U(1)}$, the weight distribution of the $S^z_j$ operator seems to be significantly more complicated,  in particular a significant portion of the weight appears to lie in states of higher energy.
Since the operator $\oket{S^z_j}$ corresponds to the composite spin $\ket{\tlt}_j$ on the ladder,  it is easy to see that it has non-zero overlap only on the eigenstates of $\hmP_{t-J_z|\text{comp}}$.
Although the Heisenberg model, hence $\hmP_{t-J_z|\text{comp}}$ is completely integrable,  its eigenstates do not have a simple form, which hinders a fully analytical computation of these overlaps.
Nevertheless, with a combination of analytical and numerical results, we are able to deduce the existence of tracer diffusion from  the spectrum of $\hmP_{t-J_z|\text{comp}}$.
We first express $\overline{C_{\hA}(t)}$ from Eq.~(\ref{eq:correlationavg}) as
\begin{gather}
\overline{C_{\hA}(t)} = \int{\mathrm{d}E\ w_{\hA}(E)\ e^{-\kappa E t}} = \int{\mathrm{d}E\ \frac{\mathrm{d} \Omega_{\hA}}{\mathrm{d} E}\ e^{-\kappa E t}}, \nn \\
w_{\hA}(E) \defn \lim_{\Delta E \to 0} \frac{1}{\Delta E} \sum_{E' \in (E,E+\Delta E)} W_{\hA}(E'), \nn \\
\Omega_{\hA}(E) \defn \int_0^E {\mathrm{d}E'}\ w_{\hA}(E') = \sum_{E' < E} W_{\hA}(E'),
\label{eq:CAweight}    
\end{gather}
where $w_{\hA}(E)$ can be interpreted as the ``density" of the weight at a given energy $E$, and $\Omega_{\hA}(E)$ is the cumulative weight on eigenstates at energies below $E$.
This expression is valid for finite sizes with discrete levels (possibly degenerate but all are included in the formal sum, but is written in anticipation of the thermodynamic limit $L \to \infty$ where for a local observable $\hA$ we expect $\Omega_{\hA}(E)$ to converge to an $L$-independent function that at $E \to \infty$ gives the total weight $\frac{1}{D}\obraket{\hA}{\hA}$ which is a fixed $O(1)$ number.
Note that in fact we expect most of the weight to be spread over a finite range of $E$, since $\oket{\hA}$ for any strictly local operator $\hA$ can be viewed as a result of an action of a super-operator $\oketbra{\hA}{\mathds{1}}$ that is strictly local in the ladder formulation on one of the ground states $\oket{\mathds{1}}$ of the super-Hamiltonian, which can deposit only finite energy of the latter.
As an example of immediate interest to us, a simple calculation gives $\obra{S^z_j} \hmP_{t-J_z} \oket{S^z_j} / \obraket{S^z_j}{S^z_j} = 4/3$.
We can then use the behavior of $\Omega_{\hA}(E)$ at small $E$ to deduce the behavior of $\overline{C_{\hA}(t)}$ at large $t$. 
For example, if we have $\Omega_{\hA}(E) \sim E^\gamma$ as $E \rt 0$, according to Eq.~(\ref{eq:CAweight}) we have 
\begin{equation}
    \overline{C_{\hA}(t)} \sim \gamma \int{\mathrm{d}E\ E^{\gamma-1} e^{- \kappa E t}} \sim t^{-\gamma}\;\;\text{for large }t. 
\label{eq:CAscaling}
\end{equation}
For the operator $\oket{S^z_j}$, the weights on the eigenstates of $\hmP_{t-J_z|\text{comp}}$ can formally be written down in terms of eigenstates of the Heisenberg model; we present the details in App.~\ref{app:tracerdiffusion}.  
The cumulative weight distribution $\Omega_{S^z_j}(E)$ can then be computed numerically, and its form for $j = \frac{L+1}{2}$ and odd system sizes with OBC in shown in Fig.~\ref{fig:tJzweight}.
The nature of this distribution as $E \rt 0$ is consistent with $\gamma \in [0.25, 0.3]$ in Eq.~(\ref{eq:CAscaling}), which is consistent with the scaling expected from the tracer diffusion~\cite{feldmeier2022emergent}.
\subsubsection{Asymptotic QMBS}\label{subsec:asymptoticQMBS}
We now demonstrate that asymptotic QMBS, introduced in \cite{gotta2023asymptotic}, can be understood in terms of low-energy excitations of the super-Hamiltonians corresponding to towers of QMBS, e.g.,  $\hmP_{\scar}$ of Eq.~(\ref{eq:RSigm}) for the bond algebra of Eq.~(\ref{eq:ACtower}).
In models with a tower of exact QMBS, asymptotic QMBS are low-entanglement states orthogonal to exact QMBS that have a vanishing energy variance in the $L \rt \infty$.
As a consequence of their low variance, their relaxation time diverges with system size, a property that does not happen for generic low-entanglement states under local Hamiltonian dynamics~\cite{banuls2020entanglement}.
Simple examples of asymptotic QMBS~\cite{gotta2023asymptotic} in the context of models with the ferromagnetic tower of QMBS, which correspond to the algebras of Eq.~(\ref{eq:ACtower}), are
\begin{equation}
	\ket{\Phi_{n,k}} \defn \frac{1}{\sqrt{\mathcal{N}_{n,k}}} S^-_k \ket{\Phi_{n+1,0}},\;\;\;1 \leq n \leq L-1,
\label{eq:aQMBS}
\end{equation}
where $S^-_k \defn \sum_{j}{}{c_{j, k} S^-_j}$, with $c_{j, k}$ chosen such that $\braket{\Phi_{n,k}}{\Phi_{n',0}} = \delta_{n,n'} \delta_{k,0}$ and the variance of $\ket{\Phi_{n,k}}$ decreases with increasing system size $L$, and $\mathcal{N}_{n,k}$ is a normalization constant that can be tedious to compute.
A simple way to satisfy these conditions is to choose $\ket{\Phi_{n,k}}$ to be a spin-wave on top of the ferromagnet with $k \ll 2\pi$ such that $k \to 0$ as $L \to \infty$, hence the similarity between $\ket{\Phi_{m,k}}$ and $\sket{\lambda^z_{m,k}}$ of Eq.~(\ref{eq:Heisenbergspinwavestates}).
For example, using $(c_{j,k}, k)$ of the form of Eq.~(\ref{eq:HeisenbergPBC}) with $k = \frac{2\pi}{L}$ leads to an energy variance of $\sim 1/L^2$ and fidelity relaxation timescale $\sim L$, similar to the example discussed in \cite{gotta2023asymptotic}.
We now show that the behavior of asymptotic QMBS of the form of Eq.~(\ref{eq:aQMBS}) can be understood from the low-energy excitations of the super-Hamiltonian $\hmP_{\text{tower}}$ for the tower of QMBS.
We start with a subspace of states of the form $\ket{\psi}_t \otimes \ket{\Phi_{m,0}}_b$, where the state on the bottom leg of the ladder is $\ket{\Phi_{m,0}}$, an exact QMBS shown in Eq.~(\ref{eq:FMtower}).
Using the fact that $R_{[j, j+1];b}\ket{\Phi_{m,0}}_b = 0$ and $R_{[j,j+1]; \ell}^2 = R_{[j,j+1]; \ell}$, $\hmP_{\text{tower}}$ of the form of Eq.~(\ref{eq:RSigm}) keeps this subspace invariant, and acts within it as
\begin{equation}
    \hmP_{\text{tower}} = 8\sumal{j}{}{R_{[j,j+1];t}} = 2L - 8\sumal{j}{}{(\vec{S}_{j;t}\cdot\vec{S}_{j+1;t})},
\label{eq:Ptowerrest}
\end{equation}
which, up to an overall factor, is simply the ferromagnetic Heisenberg Hamiltonian of App.~\ref{app:Heiscanon} acting on the top leg.
Hence, excitations within this subspace are spin-waves on the top leg of the ladder of the form
\begin{equation}
\ket{\Phi_{n,k}}_t \otimes \ket{\Phi_{m,0}}_b =
\oket{\ketbra{\Phi_{n,k}}{\Phi_{m,0}}},
\label{eq:aQMBSeigenstate}
\end{equation}
where $\ket{\Phi_{n,k}}$ is a spin-wave on top of the Heisenberg ferromagnet, e.g., as defined in Eq.~(\ref{eq:aQMBS}).
These states on the ladder have energy as shown in Eq.~(\ref{eq:U1energy}), in particular the dispersion scales as $\sim k^2$.
Hence the relaxation of the autocorrelation function of any operator to its Mazur bound is expected to occur on timescales of $\sim L^2$, provided that operator has a non-zero overlap with the slowly relaxing mode.
To apply the general theory of autocorrelation functions to study the asymptotic QMBS, we first note some general properties that hold for an initial state $\ket{\psi}$ evolving under a Brownian circuit corresponding to an algebra $\mA = \lgen \{\hH_\alpha\} \rgen$.
The autocorrelation function of an operator $\hA = \ketbra{\psi}{\Upsilon}$, for any normalized state $\ket{\Upsilon}$, is given by 
\begin{equation}
    \text{Tr}[\hA(0)^\dagger \hA(-t)] = \braket{\psi(0)}{\psi(t)}\braket{\Upsilon(t)}{\Upsilon(0)},
\label{eq:overlapauto}
\end{equation}
where $\hA(-t)$ can be viewed as the time-evolved operator under the Brownian circuit with bond generators of opposite sign; in this context $\hA(-t) = \ketbra{\psi(t)}{\Upsilon(t)}$, with $\ket{\psi(t)}$ and $\ket{\Upsilon(t)}$ the time-evolved states under a single realization of the Brownian circuit couplings $\{J_\alpha^{(t')}\}$, see App.~\ref{app:QMBSbrownian} for precise details.
Given that the algebra $\mA$ admits a \textit{singlet} $\ket{\Upsilon}$, i.e., that is an eigenstate of each of the $\{\hH_\alpha\}$ (so that $\ket{\Upsilon(t)} = e^{-i \varphi t} \ket{\Upsilon(0)}$ for some $\varphi$ that in general depends on the random couplings), we can use Eq.~(\ref{eq:overlapauto}) to lower-bound the ensemble-averaged fidelity as follows:
\begin{align}
    \overline{\mF(t)} &\defn \overline{\left|\braket{\psi(0)}{\psi(t)} \right|^2} \geq \overline{\left|\braket{\psi(0)}{\psi(t)} \right|}^2 \nn \\ &\geq \left|\overline{\braket{\psi(0)}{\psi(t)} e^{i \varphi t}}\right|^2 =  \left|\overline{\text{Tr}[\hA(0)^\dagger \hA(-t)]}\right|^2 ~.
    \label{eq:fidelityauto}
\end{align} 
In the case of asymptotic QMBS, we can choose $\hA = \ketbra{\Phi_{n,k}}{\Phi_{m,0}}$ with normalized $\ket{\Phi_{m,0}}$ and $\ket{\Phi_{n,k}}$, which satisfy the required conditions.
This operator $\oket{\hA}$ is precisely that in Eq.~(\ref{eq:aQMBSeigenstate}), and is hence an eigenstate of the super-Hamiltonian $\hmP_{\text{tower}}$ with eigenvalue $p_k = 8[1-\cos(k)]$.
Using Eq.~(\ref{eq:fidelityauto}) we have
\begin{equation}
\overline{|\braket{\Phi_{n,k}(0)}{\Phi_{n,k}(t)}|^2} \geq  \left|\overline{\obraket{\hA(0)}{\hA(-t)}}\right|^2 = e^{-2 \kappa p_k t} ~.
\label{eq:overlapavg}
\end{equation}
Thus, the average fidelity decays on timescales $\sim L^2$.
This qualitatively recovers that the fidelity of asymptotic QMBS decays on timescales that grow with system size.
However, note that this scaling differs quantitatively from the $\sim L$ scaling of the fidelity decay timescale seen in Hamiltonian systems with asymptotic QMBS~\cite{gotta2023asymptotic}.
We hypothesize on reasons for this difference between the Brownian circuit and Hamiltonian systems in App.~\ref{app:QMBSbrownian}, and it appears to be related to the quantum Zeno effect due to stochasticity in the Brownian circuit. 
The behavior of the overlap in Eq.~(\ref{eq:overlapavg}) can also be understood from a direct analysis of the evolution of states under Brownian circuit dynamics, which we discuss in App.~\ref{app:QMBSbrownian}.
Considerations discussed there also lead us to the following conjecture on the existence of asymptotic QMBS in Hamiltonians with exact QMBS.
\begin{restatable}{con}{aqmbs}\label{con:aqmbs}
Consider a space $\mS = {\text{span}}\{\ket{\Phi_n}\}$ that can be expressed as the exhaustive common kernel of a set of strictly local projectors. 
Any local Hamiltonian that realizes this subspace as the exact QMBS subspace also has asymptotic QMBS if $\mS$ cannot be expressed as the ground state space of a gapped frustration-free Hamiltonian.
Furthermore, the gapless excitations of any such Hamiltonian are the asymptotic QMBS.
\end{restatable}
The same phenomenology generalizes to cases where the QMBS are non-degenerate, and we discuss this in more detail in App.~\ref{app:QMBStower}.
Finally, we remark that even though one can construct low-energy excitations of $\hmP_{\text{tower}}$ with dispersion $\sim k^2$ that is similar to the dispersion of the low-energy excitations of $\hmP_{U(1)}$ in the $U(1)$ symmetry case, there is generally no diffusion of local operators in systems with only QMBS.
This is due to the exponentially small overlap of local operators on these low-energy modes, similar to the result that QMBS have an exponentially small contribution to the Mazur bound of general local operators demonstrated in \cite{moudgalya2022exhaustive}.
\subsubsection{Other Continuous Symmetries}
The strategy of studying the low-energy excitations of the super-Hamiltonians can be applied to more general symmetries, and we briefly discuss two cases here.
First, this can be applied to non-Abelian symmetries such as $SU(q)$ for $q \geq 2$.
As discussed in Sec.~\ref{subsec:SU2symmetry}, the simplest super-Hamiltonians in such cases are Heisenberg-like models with $SU(q^2)$ symmetry, Eq.~(\ref{eq:suqsuperhamiltonian}), and the ground states are $SU(q^2)$ ferromagnets.
We can then straightforwardly also obtain exact low-energy excitations of such Hamiltonians by creating spin-waves on top of these generalized ferromagnetic states, e.g.,
\begin{equation}
    \sum_{j=1}^L e^{i k j} S^{\tilde{m},\tilde{m}'}_j\ket{\tilde{m},\dots,\tilde{m}, \tilde{m}, \tilde{m}, \dots, \tilde{m}}, 
\end{equation}
where $S^{\tilde{m},\tilde{m}'}_j$ is the operator that changes the state on the (rung) site $j$ from $\tilde{m}$ to $\tilde{m}'$. 
This state can be shown to have energy $4(1 - \cos(k)) \sim k^2$ at small $k$, similar to the spin-waves of the spin-1/2 Heisenberg model.
This enables computations of various correlation functions, including autocorrelation functions of local operators similar to the $U(1)$ case discussed in Sec.~\ref{subsec:U1gapless}, and we get similar answers, e.g., diffusion due to the similar nature of the super-Hamiltonians in both cases.
Second, the analysis simplifies in the case of ``classical" symmetries, where the super-Hamiltonians map onto RK-type Hamiltonians, as discussed in App.~\ref{app:RKHamils}. 
Similar RK-type Hamiltonians appear in the study of spectral form factors in Floquet random circuits with symmetries or constraints~\cite{friedman2019spectral, moudgalya2021spectral, singh2021subdiffusion}, and the \textit{Thouless time} is determined by the scaling of the inverse of the gap of the corresponding Hamiltonian with system size.  
For example, RK-type Hamiltonians that appear in the context of dipole and multipole symmetries were studied in \cite{moudgalya2021spectral}.
Their low-energy physics can be understood using Lifshitz-like field theories, which leads to a $\sim k^4$ dispersion of the low-energy mode for dipole moment conserving systems,  and $\sim k^{2(m+1)}$ for systems conserving the $m$-th moment.
With the appropriate choice of bond generators, the same set of RK-type Hamiltonians would appear as super-Hamiltonians in our analysis,\footnote{A caveat is that one cannot write down a bond algebra generated by strictly local terms that has the commutant generated \textit{only} by the charge and dipole symmetries -- there are necessarily exponentially many other independent symmetries due to fragmentation~\cite{moudgalya2021hilbert}.
However, with terms of large enough support, the fragmentation is ``weak"~\cite{sala2020fragmentation} and the extra symmetries are not expected to have a dramatic effect on the correlations of generic operators.} and using heuristic arguments based on the dispersion relation of the low-energy modes, we obtain that autocorrelations should decay as $\sim 1/t^{\frac{d}{2(m+1)}}$ in $d$-dimensional systems with $m$-th multipole moment conservation, indicating subdiffusion.
This Brownian circuit approach to determine transport phenomena was also recently applied to short-range and long-range dipole conserving Hamiltonians~\cite{ogunnaike2023unifying}, where the low-energy excitations of the effective super-Hamiltonians (referred there as Lindbladians) yielded results consistent with those obtained from other methods~\cite{feldmeier2020anomalous, moudgalya2021spectral, morningstar2023hydrodynamics, gliozzi2023hierarchical}. 
We close this discussion with a general remark on the low-energy excited states of general super-Hamiltonians.
Note that the identity operator $\oket{\mathds{1}}$ is always a ground state of any super-Hamiltonian, since it always belongs to the commutant algebra. 
In the ladder language,  this corresponds to a ``homogeneous" product state, e.g., $\ket{\trt \trt \cdots \trt}$ for spin-1/2 systems.  
Given that the super-Hamiltonian is a local superoperator, it is natural to expect that its low-energy spectrum should be well approximated by a ``quasiparticle" trial state that aids in determining the late-time transport.
This happens in all of the cases we have studied, however exploring the validity of this statement or coming up with counterexamples would be an interesting avenue for future work.
\section{Conclusions and Outlook}\label{sec:conclusions}
In this work, we showed that many examples of both conventional and unconventional symmetries can be understood as ground states of local superoperators interpreted as Hamiltonians acting on a doubled ladder Hilbert space, hence we referred to them as ``super-Hamiltonians.''
This originates from the understanding of symmetries as commutants of bond algebras generated by local operators, as illustrated in \cite{moudgalya2021hilbert, moudgalya2022from, moudgalya2022exhaustive, moudgalya2023numerical}.
For conventional symmetries such as $\mbZ_2$, $U(1)$, $SU(2)$, the symmetry algebras can be interpreted as various kinds of ferromagnetic states of appropriate super-Hamiltonians.
Unconventional symmetries such as fragmentation and QMBS also led to frustration-free Hamiltonians with solvable ground states.
We then showed that the low-energy spectra of the super-Hamiltonians can be interpreted as approximate symmetries associated with the exact symmetries.
We did this by showing that super-Hamiltonians obtained this way are effective Hamiltonians that describe noise-averaged dynamics in noisy symmetric Brownian circuits constructed using the bond algebra generators.
This gives a physical interpretation for the super-Hamiltonians, and connects their low-energy excited states to slowly relaxing hydrodynamic modes of the symmetric Brownian circuits.
This also gives a novel interpretation for the Mazur bound~\cite{mazurbound1969, dhar2020revisiting}, usually interpreted as a lower bound for the time-averaged autocorrelation function, as the saturation value of the ensemble-averaged autocorrelation function of Brownian circuits.
The approach to this saturation value is governed by the low-energy spectra of the super-Hamiltonians, hence their low-energy eigenstates beyond the ground states have interpretations as approximate symmetries.
We then explicitly solved for the low-energy spectra of the super-Hamiltonians and discussed the dynamical consequences of the associated slowly-relaxing modes.
Using this framework, we first recovered well-known facts that while conventional discrete symmetries such as $\mbZ_2$ have gapped super-Hamiltonians, and hence no associated slow-modes, conventional continuous symmetries such as $U(1)$ and $SU(2)$ have gapless super-Hamiltonians, and the corresponding slow-modes lead to diffusion.
However, we showed that this framework works much more generally, for understanding slow modes associated with unconventional symmetries such as fragmentation and QMBS as well.
While isolated QMBS have gapped super-Hamiltonians, and hence no associated slow modes, towers of QMBS have \textit{asymptotic scars} of the type discussed in \cite{gotta2023asymptotic} as slow-modes. 
Hilbert space fragmentation in the $t-J_z$ model has slow-modes, which can be used to understand tracer diffusion in such systems, pointed out in earlier works~\cite{barma1994slow, feldmeier2022emergent}.
On a technical note, the quantitative understanding of the slow relaxation of certain observables in some cases such as the $t-J_z$ fragmentation required a careful analysis of the full low-energy spectrum (including appropriate weights for observables), rather than the simple scaling of the gap that has been sufficient for such purposes in earlier works~\cite{moudgalya2021spectral, ogunnaike2023unifying}.
In all, our work connects studies of the commutant algebra focusing on exact conserved quantities (ground states of the corresponding super-Hamiltonians) to studies of hydrodynamic and transport properties controlled by approximately conserved quantities (low-lying excitations of the super-Hamiltonians) in symmetric systems.
While we restricted illustrations to one-dimensional systems, the results and phenomenology generalize straightforwardly to higher dimensional systems.
It would be interesting to explore the applicability of this method to other generalized symmetries being studied in the literature, e.g., subsystem symmetries~\cite{xu2004strong, you2018subsystem, iaconis2019anomalous,  zhou2021fractal, wildeboer2021symmetryprotected}, spatially modulated symmetries~\cite{sala2022dynamics}, categorical or MPO symmetries~\cite{aasen2020topological, lootens2021dualities, borsi2023matrix}, and understand if they can be viewed as ground states of local super-Hamiltonians.
The low-energy spectrum of the corresponding super-Hamiltonians should reveal the late-time dynamical properties of such systems and of the associated hydrodynamic modes, which would also be interesting to explore in other models of Hilbert space fragmentation~\cite{sala2020fragmentation, moudgalya2021hilbert, read2007enlarged, li2023hilbert, stephen2022ergodicity, stahl2023topologically, hart2023exact}, and lattice gauge theories with strictly local symmetries~\cite{smith2017disorder, sala2024disorder}.
It would also be interesting to try and reproduce in this language the sector-dependent hydrodynamic behavior observed in many models of Hilbert space fragmentation based on pattern conservation or ``irreducible strings."~\cite{barma1994slow, menon1995irreducible}
Of course, some symmetries such as dynamical symmetries~\cite{buca2019nonstationary, moudgalya2022from}, strictly speaking, are not ground states of local Hermitian superoperators since they correspond to algebras generated by including also extensive local terms (see~\cite{moudgalya2022from} for more details), but we hope that a generalization of this story might capture many more examples. 
In addition, lattice symmetries or those that appear in the context of integrability~\cite{yuzbashyan2013integrability} have so far not been explored in the commutant framework, which would be an interesting direction to pursue.
The fact that the super-Hamiltonians can be understood as frustration-free Hamiltonians, moreover of the Rokhsar-Kivelson (RK) form in many cases, also opens up many directions of exploration.
First, such Hamiltonians are easy to analyze, and this method might provide a better  systematic approach to prove the exhaustion of commutant algebras, which has turned out to be tedious using brute-force methods~\cite{moudgalya2022from, moudgalya2022exhaustive}.
Second, they are also amenable to standard techniques for proving gaps or their absence~\cite{knabe1988energy, nachtergaele1996spectral, gosset2016local, kastoryano2018divide, lsm1961soluble, oshikawa2000commensurability, hastings2004lieb}, and the understanding of which symmetries have a gap is important for understanding the nature of late-time transport in symmetric systems. 
Third, RK Hamiltonians have connections to several standard concepts in classical Master equations and also to spectral graph theory~\cite{chung1997spectral, castelnovo2005rk, bravyi2015monte, moudgalya2021spectral}, and it would be interesting to exploit this property to study the low-energy excited states using existing methods such as classical stochastic circuits similar to those used in the literature~\cite{ritort2003glassy, morningstar2020kineticallyconstrained, feldmeier2020anomalous, li2023hilbert}, and potentially also Quantum Monte Carlo techniques~\cite{yan2022height}. 
Finally, many of these super-Hamiltonians also have interesting continuum limits,
and their low-energy physics can be understood in terms of field theories.
For example, several types of RK Hamiltonians map onto Lifshitz field theories that are easy to analyze~\cite{henley2004classical, moessner2011dimer, fradkin2013field, moudgalya2021spectral}.
Given that many generalized symmetries are studied in the context of quantum field theories in the continuum~\cite{mcgreevy2022generalized, cordova2022snowmass}, it is natural to wonder if the novel symmetries there, e.g., non-invertible symmetries understood via category theory, can be understood as ``ground states'' in any sense.
Some aspects of the super-Hamiltonian constructions, e.g., working in a doubled Hilbert space and studying the low-energy physics, resemble the Schwinger-Keldysh formalism~\cite{chou1985equilibrium, landsman1987real, kamenev2009keldysh, haehl2017schwinger}, and it would be useful to elucidate these connections further.
More speculatively, connecting symmetry algebras to ground states should also help impose some general constraints on symmetry operators, e.g., perhaps they necessarily have MPO forms, or some restrictions on their operator entanglement.
Moreover, the fact that symmetry, which is a property of the Hilbert space, is connected to ground state properties of a local operator, is consistent with the conjecture that symmetries are related to topological orders --- a ground state property --- in one higher dimension~\cite{kong2020algebraic, chatterjee2022algebra, moradi2022topological}.
The commutant framework along with this ground state mapping might be a framework to make such a correspondence more precise in lattice systems.
The language of super-Hamiltonians also introduces a precise language to discuss approximate symmetries.
While we illustrated this only for approximate symmetries that accompany exact symmetries, it would be very interesting to identify bond algebras without exact symmetries, but with approximate symmetries that appear as low-energy excited states of the super-Hamiltonians, which could lead to slow dynamics and phenomena such as prethermalization.
Furthermore, as we have shown in Sec.~\ref{subsec:approxblock}, approximate symmetries are also potentially related to approximate block-diagonal structures, and hence the language of super-Hamiltonians might help explain the origin of approximate symmetries in certain systems in the literature, e.g., the PXP model~\cite{serbyn2020review} is known to exhibit approximate QMBS and approximate block-diagonal structures~\cite{bull2020quantum}.
On a different note, since algebra-based Brownian circuits played a crucial role in understanding/interpreting the super-Hamiltonian spectrum, it seems like a useful setting to explore more.
For example, it is likely that several results on Haar random circuits can be reproduced using the seemingly more tractable Brownian circuits, and indeed, there have been many interesting works studying the properties of ``generic" Brownian circuits using ``effective Hamiltonians", which are super-Hamiltonians of the kind we study in this work~\cite{lashkari2013towards, sunderhauf2019quantum, xu2019locality}.
On the other hand, Brownian circuits with symmetries have been much less studied, and the large class of ``algebra-based" circuits we introduced in this work, which are defined using the bond algebra corresponding to the symmetry, might prove to be useful toy models that are easier to study than symmetric Haar random circuits for a number of reasons.
First, defining the latter requires the knowledge of the irreducible representations~\cite{rakovszky2018diffusive, khemani2018operator, friedman2019spectral, hearth2023unitary}, i.e., the block diagonal structure of each gate, whereas Brownian circuits only require the generators of the corresponding bond algebra.
Second, in contrast to rigid Haar random circuits, the class of Brownian circuits we study possesses a lot of tunable parameters in the choice of their generators, which might lead to more analytically tractable super-Hamiltonians that provide better physical insights. 
Finally, while computations in Haar random circuits map onto questions in classical statistical mechanics, computations in Brownian circuits map onto the low-energy physics of effective super-Hamiltonians, which, although are equivalent to questions in classical statistical mechanics, are nevertheless more directly accessible using analytical and numerical treatments developed in the context of quantum many-body systems.
For example, the hydrodynamic modes associated with the symmetries, arise more ``naturally" as ``low-energy excitations" on top of simple ground states, which can be studied using a variety of variational methods.
Hence this should be a nice analytical tool to explore the physics of symmetric systems, including those with unconventional symmetries, and this can be contrasted from the physics of systems without any symmetry, by studying bond algebra generators that have a trivial commutant of only the identity operator. 
Finally, it is important to better understand the precise connections between the dynamics of Brownian circuits, and more general Hamiltonian or Floquet systems. 
While the microscopic physics is expected to be different, \textit{universal} properties such as hydrodynamic modes, that arise solely due to the symmetry and locality of the systems, should appear in both kinds of systems, even though they are analytically tractable only in Brownian circuits.
It would be interesting to check if these modes survive under ``relaxing" the structure of Brownian circuits and making it closer to non-Markovian Hamiltonian systems in various ways, e.g., by incorporating temporally correlated noise.
\section*{Acknowledgements}
We thank Deepak Dhar, David Huse, Aditi Mitra, Bruno Nachtergaele, Adam Nahum, Tibor Rakovszky, and Nat Tantivasadakarn for useful discussions.
S.M.\ thanks Lorenzo Gotta and Leonardo Mazza for a previous collaboration~\cite{gotta2023asymptotic}.
This work was supported by the Walter Burke Institute for Theoretical Physics at Caltech; the Institute for Quantum Information and Matter, an NSF Physics Frontiers Center (NSF Grant PHY-1733907); the National Science Foundation through grant DMR-2001186; and the Munich Center for Quantum Science and Technology (S.M.).
A part of this work was done at the Aspen Center for Physics, which is supported by National Science Foundation grants PHY-1607611 and PHY-2210452.
S.M.\ also acknowledges the hospitality of the Physik-Institut at the University of Zurich, where parts of this manuscript were written.
\bibliography{refs, newrefs}
\appendix 
\onecolumngrid
\section{The Ferromagnetic Heisenberg Model}\label{app:Heiscanon}
In this appendix, we define a canonical form for the ferromagnetic Heisenberg model and set the conventions we use to describe it and its eigenstates.
This appears repeatedly in the analysis of various super-Hamiltonians we study in the main text.
This is a spin-$\frac{1}{2}$ Hamiltonian acting on a system of size $L$ with the local degrees of freedom $\ket{\up}_j$ and $\ket{\dn}_j$ in the $\hat{z}$-basis or $\ket{\rt}_j$ and $\ket{\lt}_j$ in the $\hat{x}$-basis.
We use the convention that these are related as
\begin{equation}
    \ket{\rt}_j \defn \frac{\ket{\up}_j + \ket{\dn}_j}{\sqrt{2}},\;\;\;\ket{\lt}_j \defn \frac{\ket{\up}_j - \ket{\dn}_j}{\sqrt{2}}.
\label{eq:xzconvention}
\end{equation}
The standard forms of the Heisenberg Hamiltonian we use in this work are given by
\begin{align}
    H_{\text{Heis}} &= \sumal{j = 1}{L_{\text{max}}}{(\ket{\up \dn} - \ket{\dn \up})(\bra{\up \dn} - \bra{\dn \up})}_{[j, j+1]} = \sumal{j = 1}{L_{\text{max}}}{(\ket{\rt \lt} - \ket{\lt \rt})(\bra{\rt \lt} - \bra{\lt \rt})}_{[j, j+1]}
    = \sumal{j = 1}{L_{\text{max}}} (1 - P_{j,j+1}^{\text{(2)}})
    \nn \\
    &=\frac{1}{2}\sumal{j = 1}{L_{\text{max}}}{[1 - (X_j X_{j+1} + Y_j Y_{j+1} + Z_j Z_{j+1})]} = 2\sumal{j = 1}{L_{\text{max}}}{[\frac{1}{4} - \vec{S}_j\cdot\vec{S}_{j+1}]}
\label{eq:heisenbergcanon}
\end{align}
where $\{X_j, Y_j, Z_j\}$ are the Pauli operators on site $j$, $\{S^x_j, S^y_j, S^z_j\}$ are the spin operators on site $j$, which are the Pauli matrices multiplied by a factor of $\frac{1}{2}$, and $P^{(2)}_{j,j+1}$ is the operator that permutes the states on site $j$ and $j+1$ defined in Eq.~(\ref{eq:P2exch}).
Moreover, $L_{\max} = L - 1$ with OBC, and $L_{\max} = L$ with PBC with the subscripts taken to be modulo $L$.
We also define raising and lowering operators
\begin{equation}
    S^{z-}_{\tot} \defn \sumal{j = 1}{L}{S^{z-}_j} = \sumal{j = 1}{L}{\ketbra{\dn}{\up}_j},\;\;S^{x-}_{\tot} \defn \sumal{j = 1}{L}{S^{x-}_j} = \sumal{j = 1}{L}{\ketbra{\lt}{\rt}_j},\;\;S^{z+}_{\tot} \defn  (S^{z-}_{\tot})^\dagger,\;\;S^{x+}_{\tot} \defn (S^{x-}_{\tot})^\dagger, 
\label{eq:raisinglowering}
\end{equation}
which all commute with $H_{\text{Heis}}$.
This Hamiltonian has $(L+1)$-fold degenerate frustration-free ground states, which we often refer to as the ``ferromagnetic multiplet."
To obtain an orthonormal basis for this multiplet, we can start with the fully polarized state with all spins in either the $+\hat{z}$ or $+\hat{x}$ direction, which read
\begin{equation}
    \ket{F^{z}_0} \defn \ket{\up\up\cdots\up},\;\;\;\ket{F^x_0} \defn \ket{\rt\rt\cdots\rt}, 
\label{eq:polarizedstates}
\end{equation}
and repeatedly act with the corresponding lowering operators $S^{z-}_{\tot}$ or $S^{x-}_{\tot}$ of Eq.~(\ref{eq:raisinglowering}) respectively to obtain $L+1$ linearly independent states of the form
\begin{equation}
   \ket{F^\alpha_{m}} = \frac{1}{\sqrt{\binom{L}{m}}}\sumal{j_1 < \cdots < j_m}{}{S^{\alpha-}_{j_1} S^{\alpha-}_{j_2} \cdots S^{\alpha-}_{j_m} \ket{F^\alpha_0}},\;\; 0 \leq m \leq L,\;\;\alpha \in \{z, x\}.
\label{eq:HeisenbergGS}
\end{equation}
Beyond the ground state space, the low-energy excitations of $H_{\text{Heis}}$ are well-known to be spin-waves.
These spin-waves are $(L-1)$-fold degenerate (corresponding to degeneracy of a multiplet with total spin of $L/2-1$), and a complete orthonormal basis for these degenerate eigenstates can be chosen as
\begin{equation}
   \ket{\lambda^\alpha_{m, k}} = \frac{1}{\sqrt{\mM_{m,k}}}\sumal{j_1 < \cdots < j_m}{}{\left[\left(\sumal{\ell = 1}{m}{c_{j_{\ell},k}}\right)S^{\alpha-}_{j_1} \cdots S^{\alpha-}_{j_{m}}\ket{F^\alpha_{0}}\right]},\;\;\;1 \leq m \leq L-1,\;\;\;\alpha \in \{z, x\},
\label{eq:Heisenbergspinwavestates}
\end{equation}
where $k$ labels orthonormal ``orbitals" $c_{j,k}$ in the single-magnon problem, e.g., $k$ is the plane wave momentum in the PBC case:
\begin{equation}
c_{j, k}^{\text{PBC}} \defn \frac{1}{\sqrt{L}} e^{i k j}, \quad k = \frac{2\pi n}{L}, \;\; 1 \leq n \leq L-1 ~;
\label{eq:HeisenbergPBC}
\end{equation}
or $k$ is the appropriate standing wave ``momentum" in the OBC case:
\begin{equation}
c_{j, k}^{\text{OBC}} \defn \sqrt{\frac{2}{L}} \cos[k(j-1/2)], \quad k = \frac{\pi n}{L}, \;\; 1 \leq n \leq L-1 ~.
\label{eq:HeisenbergOBC}
\end{equation}
Further unpacking Eq.~(\ref{eq:Heisenbergspinwavestates}), integer $1 \leq m \leq L-1$ labels states in the given $SU(2)$ multiplet with fixed $k$, and $\mM_{m,k}$ is a normalization factor chosen so that $\ket{\lambda_{m,k}}$ is normalized, with precise form shown in Eq.~(\ref{eq:spinwavenorms}) below.
However, for much of the description we can keep the spin-wave orbitals general only requiring orthonormality among themselves as well as orthogonality to a completely uniform ``$k = 0$ orbital" obtained for convenience by setting $k = 0$ in Eqs.~(\ref{eq:HeisenbergPBC}) or (\ref{eq:HeisenbergOBC}) for PBC or OBC respectively.
The states in Eq.~(\ref{eq:Heisenbergspinwavestates}) formally corresponding to the latter in fact belong to the ferromagnetic multiplet, i.e., the ground states of the Heisenberg model, and we have $\ket{\lambda^\alpha_{m,0}} = \ket{F^\alpha_m}$ of Eq.~(\ref{eq:HeisenbergGS}).
In all, the normalization factor reads
\begin{equation}
    \mM_{m, k} = \twopartdef{ \binom{L-2}{m-1}}{k \neq 0}{\frac{2 m^2}{L} \binom{L}{m}}{k = 0}.
\label{eq:spinwavenorms}
\end{equation}
The above spin wave excitation solutions are directly obtained by solving the Heisenberg Hamiltonian in the one spin flip Hilbert space spanned by the states of the form $\ket{\up\ \cdots\  \up\ \dn\ \up\ \cdots\ \up}$ or $\ket{\rt\ \cdots\  \rt\ \lt\ \rt\ \cdots\ \rt}$, which gives $\sket{\lambda^z_{1,k}}$ or $\sket{\lambda^x_{1,k}}$ respectively with total spin $L/2-1$ for $k \neq 0$, and then by repeatedly acting with the lowering operator $S^{\alpha-}_{\tot}$ on the state $\sket{\lambda^\alpha_{1, k}}$ for $\alpha \in \{z, x\}$.
The energies of these states are given by
\begin{equation}
    H_{\text{Heis}}\sket{\lambda^\alpha_{m,k}} = 4\sin^2\left(\frac{k}{2}\right)\sket{\lambda^\alpha_{m,k}},\;\;\;1 \leq m \leq L-1,\;\;\alpha \in \{z, x\} 
\label{eq:U1energy}
\end{equation}
hence the gap of $H_{\text{Heis}}$ is given by $4\sin^{2}(\frac{\pi}{2L})$ or $4\sin^{2}(\frac{\pi}{L})$ for OBC or PBC.
This gap scales as $\sim 1/L^2$, showing that $H_{\text{Heis}}$ is gapless in the thermodynamic limit.
Moreover, the $H_{\text{Heis}}$ is known to be completely integrable, and the expressions for the eigenstates can in principle be derived using the Bethe Ansatz.
However, they are in general not simple to use for our purposes, and we refer interested readers to one of the numerous review articles on the subject for more details~\cite{faddeev1996bethe, karabach1997bethe, levkovich2016bethe}. 
\section{Rokhsar-Kivelson-type Super-Hamiltonians from Classical Symmetries}\label{app:RKHamils}
In this appendix, we discuss cases where the superoperator $\hmP$ is a Rokhsar-Kivelson (RK) type Hamiltonian, also known as Stochastic Matrix Form (SMF) decomposible or Stoquastic Hamiltonians. 
Abstractly, these are Hamiltonians that are defined over a Hilbert space spanned by a set of classical configurations, which could be a set of product states or something more complex such as non-intersecting dimer coverings of a lattice. 
Given this configuration space, one can define a set of local transitions that relate two different configurations; this defines a local Hamiltonian.  
With these sets of local configurations and transitions defined, a simple RK-type Hamiltonian is defined as
\begin{equation}
    H_{\text{RK}} = \sumal{\langle C, C'\rangle}{}{\hQ_{C, C'}},\;\;\;\;
    \hQ_{C, C'} \defn \left(\sket{C} - \sket{C'}\right)\left(\sbra{C} - \sbra{C'}\right),
\label{eq:hrkstandardform}
\end{equation}
where ``$\langle C, C'\rangle$" in the sum indicates that the configurations $C$ and $C'$ are connected by some local moves. 
Noting that each term $\hQ_{C, C'}$ in $H_{\text{RK}}$ is positive-semidefinite, it is easy to solve for its ground states $\{\ket{G^{(\mK)}}\}$, which are given by
\begin{equation}
    \hQ_{C, C'}(\sket{C} + \sket{C'}) = 0,~\forall C,C' \;\;\implies\;\;\sket{G^{(\mK)}} = \frac{1}{\sqrt{N_{\mK}}}\sumal{C \in \mK}{}{\ket{C}} \, ,
\label{eq:RKGS}
\end{equation}
where $\mK$ defines a Krylov subspace of $N_{\mK}$ configurations connected by the local moves, and there is one ground state corresponding to each Krylov subspace. 
As a simple example, the ferromagnetic Heisenberg Hamiltonian of Eq.~(\ref{eq:heisenbergcanon}) is a Hamiltonian of the RK form of Eq.~(\ref{eq:hrkstandardform}), where the configuration space is the space of all product states, and the local moves that connect different configurations are given by nearest-neighbor swap $\up \dn \; \leftrightarrow \; \dn \up$ in the $\hat{z}$-basis or $\rt \lt \; \leftrightarrow \; \lt \rt$ in the $\hat{x}$-basis.
The $(L+1)$ ferromagnetic ground states of the Heisenberg model are also ground states of the form of Eq.~(\ref{eq:RKGS}), where each Krylov subspace $\mK$ consists of product states with the same total spin (since they can all be connected via the aforementioned local moves assuming a connected lattice of sites), and there are $(L+1)$ such Krylov subspaces. 
Such RK-type Hamiltonians have been extensively studied in the literature, and they naturally appear in several different physically relevant contexts.
Examples include dimer models~\cite{rk1988superconductivity, moessner2001resonating, moessner2011dimer}, Markov processes satisfying detailed balance~\cite{henley2004classical, castelnovo2005rk}, and in the study of various kinds of random circuits~\cite{friedman2019spectral, roy2020random, moudgalya2021spectral, singh2021subdiffusion}.
There are also generalizations of Hamiltonians of this type, and we refer readers to \cite{castelnovo2005rk} for further discussions.
Turning to the problem of finding commutant algebras, here we consider families of Hamiltonians defined on a $q$-level Hilbert space, which are comprised of terms that relate some set of classical product state configurations $\{\sket{C}\}$ that form a basis of the Hilbert space. 
In particular, we work with strictly local terms defined as
\begin{equation}
    T^{(\vec{\tau}, \vec{\tau}')}_{[j, j+r]} \defn (\sketbra{\vec{\tau}}{\vec{\tau}'} + \sketbra{\vec{\tau}'}{\vec{\tau}})_{[j,j+r]}, \;\;\;
    S^z_j \defn \sumal{\sigma}{}{s_\sigma \ketbra{\sigma}{\sigma}}_{j}, \;\;1 \leq \sigma \leq q, \;\;
    N^{\vec{\tau}}_{[j,j+r]} \defn  \ketbra{\vec{\tau}}_{[j,j+r]}, 
\label{eq:opdefns}
\end{equation}
where $S^z_j$ is spin operator and $s_\sigma$ is the spin of level $\sigma$, $(\vec{\tau}, \vec{\tau}')$ denotes a pair of strictly local $(r+1)$-site orthogonal product configurations that are ``connected" according to some local rules that we leave general in this appendix.
The generators of several standard examples of bond algebras can be cast in this form, e.g., $X_j X_{j+1} + Y_j Y_{j+1} = 2\; T^{(\up\dn, \dn\up)}_{[j, j+1]}$. 
With these definitions, we can write down \textit{classical} bond and commutant algebras as
\begin{equation}
    \mA_{\text{cl}} = \lgen \{S^z_j\}, \{T^{(\vec{\tau}, \vec{\tau}')}_{[j, j+r]}\} \rgen, \;\;\;\;
    \mC_{\text{cl}} = \lgen \{F_\alpha(\{S^z_j\})\} \rgen,\;\;  
\label{eq:classicalalgebras}
\end{equation}
where $\{F_\alpha(\cdots)\}$ denotes some set of polynomials [that depends on the specific set of connections $(\vec{\tau},\vec{\tau}^\prime$) in $\mA_{\text{cl}}$], essentially stating that the operators in the commutant $\mC_{\text{cl}}$ are diagonal in the computational basis. 
Note that the diagonal form of operators in $\mC_{\text{cl}}$ directly follows from the inclusion of $\{S^z_j\}$ in the generators of the bond algebra, as we showed in App.~A of \cite{moudgalya2021hilbert} (here implicitly assuming that powers of $S^z$ on a qubit generate the space of all $q \times q$ diagonal matrices).
We refer to commutants of the form of $\mC_{\text{cl}}$ as classical symmetries, since they lead to block-decompositions of Hamiltonians in $\mA_{\text{cl}}$ that are completely understood in the product state basis. 
The super-Hamiltonian of Eq.~(\ref{eq:Pexp}) corresponding to $\mA_{\text{cl}}$ is then of the form 
\begin{align}
    \hmP_{\text{cl}} &= \sumal{j}{}{(S^z_{j;t} - S^z_{j;b})^2} + \sumal{j, (\vec{\tau}, \vec{\tau}')}{}{\Big[(T^{(\vec{\tau}, \vec{\tau}')}_{[j, j+r];t})^2 + (T^{(\vec{\tau}, \vec{\tau}')}_{[j, j+r];b})^2 - 2 T^{(\vec{\tau}, \vec{\tau}')}_{[j, j+r];t} T^{(\vec{\tau}, \vec{\tau}')}_{[j, j+r];b}\Big]}\nn\\
    &= \sumal{j}{}{(S^z_{j;t} - S^z_{j;b})^2} + \sumal{j, (\vec{\tau}, \vec{\tau}')}{}{\Big[N^{\vec{\tau}}_{[j, j+r];t} + N^{\vec{\tau}'}_{[j, j+r];t} + N^{\vec{\tau}}_{[j, j+r];b} + N^{\vec{\tau}'}_{[j, j+r];b} - 2 T^{(\vec{\tau}, \vec{\tau}')}_{[j, j+r];t} T^{(\vec{\tau}, \vec{\tau}')}_{[j, j+r];b}\Big]}.
\label{eq:Liouvclass}
\end{align}
The minimization of energy under the first term in Eq.~(\ref{eq:Liouvclass}) ensures that the ground state is in the subspace spanned by composite spins defined on the rungs of the ladder (or bilayer) as $\sket{\tilde{\sigma}} \defn \ket{\begin{array}{c} \sigma\\ \sigma\end{array}}$. 
$\mP_{\text{cl}}$ restricted to this composite spin subspace then reads,
\begin{align}
    \hmP_{\text{cl}|\text{comp}} &= 2\sumal{j, (\vec{\tau}, \vec{\tau}')}{}{\left[\sketbra{\widetilde{\vec{\tau}}}{\widetilde{\vec{\tau}}}_{[j,j+r]} + \sketbra{\widetilde{\vec{\tau}'}}{\widetilde{\vec{\tau}'}}_{[j,j+r]} - \sketbra{\widetilde{\vec{\tau}}}{\widetilde{\vec{\tau}'}}_{[j,j+r]} - \sketbra{\widetilde{\vec{\tau}'}}{\widetilde{\vec{\tau}}}_{[j,j+r]}\right]}\nn \\
    &=2\sumal{j, (\vec{\tau}, \vec{\tau}')}{}{(\sket{\widetilde{\vec{\tau}}} - \sket{\widetilde{\vec{\tau}'}})(\sbra{\widetilde{\vec{\tau}}} - \sbra{\widetilde{\vec{\tau}'}})_{[j,j+r]}}, 
\label{eq:Pclassketbra}
\end{align}
where $\sket{\widetilde{\vec{\tau}}}$ is the $(r+1)$-site composite spin configuration that contains  identical $\vec{\tau}$ placed on both the top and bottom legs of the ladder (similar to $\ket{\tilde{\sigma}}$). 
Note that several superoperators studied in the main text can be brought to this form, e.g., Eqs.~(\ref{eq:PU1comp}) and (\ref{eq:tJzcomposite}).
Eq.~(\ref{eq:Pclassketbra}) can alternately be expressed in terms of overall classical product configurations $\{\ket{C}\}$ as
\begin{equation}
    \hmP_{\text{cl}|\text{comp}} = \sumal{\langle C, C' \rangle}{}{(\sket{\tC} - \sket{\tC'})(\sbra{\tC} - \sbra{\tC'})},\;\;\; \sket{\tC} \defn \ket{\begin{array}{c} C \\ C\end{array}},
\label{eq:Pclass}
\end{equation} 
which is precisely a Hamiltonian of the Rokhsar-Kivelson form of Eq.~(\ref{eq:hrkstandardform}), and the corresponding analysis of the ground states can be immediately reused with tildes playing a dummy role (since $\widetilde{C}$'s are in one-to-one correspondence with $C$'s). 
Hence the ground states of $\hmP_{\text{cl}}$ are of the form of Eq.~(\ref{eq:RKGS}), and the number of ground states is the number of Krylov subspaces of classical configurations connected by the moves $\vec{\tau} \leftrightarrow \vec{\tau}'$.
In the operator language, noting that the composite spins $\ket{\tilde{\sigma}}$ map onto projectors $\oket{\ketbra{\sigma}}$, the ground state $\ket{G^{(\mK)}} \sim \sum_{C \in \mK} \sket{\widetilde{C}}$ of $\hmP_{\text{cl}|\text{comp}}$ maps onto $\oket{P_{\mK}}$, where $P_{\mK}$ is the projector onto the Krylov subspace $\mK$. 
In the case of classical conventional symmetries such as $U(1)$, these Krylov subspaces are equivalent to conventional symmetry quantum number sectors.

\section{Formal Symmetries of Super-Hamiltonians}
\label{app:superHsymms}
In this Appendix, we discuss some additional formal properties of the constructed super-Hamiltonians viewed as  ladder/bilayer systems as described in Sec.~\ref{subsec:ladder_bilayer} and Eq.~(\ref{eq:Pexp}).
By construction, such super-Hamiltonians have symmetries that descend directly from the symmetries of the bond algebra terms $\{ \hH_\alpha \}$: Using Eq.~(\ref{eq:Pexp}) it is easy to see that for any $\hC \in \mC$ (i.e., the commutant of the algebra $\mA = \lgen \{ \hH_\alpha \} \rgen$), we have that $\hC_t \otimes \mathbb{1}_b$ and $\mathbb{1}_t \otimes \hC^T_b$ commute with the super-Hamiltonian.
In studies of Lindbladians~\cite{buca2012note,albert2014symmetries} these symmetries associated independently with each leg/layer are often referred to as \textit{strong symmetries}, which commute with all the jump operators.
There are additional symmetries of the super-Hamiltonian that do not play significant roles in our analysis,  e.g., the super-Hamiltonians also commute with an {\it anti-unitary} operator composed of the exchange (swap) operation between the two legs and complex conjugation in the computational basis.
While there seem to be no other obvious symmetries, for each $\hC \in \mC$ the corresponding $\oket{\hC}$ is an exact zero-eigenvalue eigenstate of each $\hmL_\alpha$; hence $\oketbra{\hC}{\hC'}$ with $\hC,\hC' \in \mC$ are additional conserved quantities of the super-Hamiltonians, which can be understood by thinking about symmetries in terms of $\{\hmL_\alpha\}$.
This character of the physical symmetries $\oket{\mC}$ with respect to families of super-Hamiltonians resembles exact scar states~\cite{moudgalya2022exhaustive} but here occurring in the superoperator space.
Furthermore, since the super-Hamiltonians contain only positive-semidefinite terms $\{\hmL^\dagger_\alpha \hmL_\alpha\}$,  they have an additional ``quantitative feature" that all $\oket{\hC}$'s are exact ground states by the very construction.
As we discuss below, in examples with conventional Abelian symmetries this quantitative feature leads to the symmetries being broken in the ground states in a particular way that preserves some combinations of the symmetries; that is, the super-Hamiltonians can be loosely viewed as being in a particular partial symmetry breaking phase.
Below we will illustrate the formal symmetries and their fate in the ground states of the super-Hamiltonians for several conventional symmetries from Sec.~\ref{sec:examples}, and we also discuss the extension of the concepts to the unconventional case of isolated QMBS.
\subsection{Global $\mbZ_2$ Symmetry}
\label{subapp:superHsymms_Z2}
We start with the case of the global $\mbZ_2$ symmetry considered in Sec.~\ref{subsubsec:Z2}.
Here the inherited ``strong" symmetries can be expressed in terms of unitaries $U_t^{\mbZ_2} \defn \prod_j Z_{j;t}$ and $U_b^{\mbZ_2} \defn \prod_j Z_{j;b}$, which are $\mbZ_2$ symmetries associated with each individual leg.
The two exact ground states $\ket{G_\rt}$ and $\ket{G_\lt}$ of Eq.~(\ref{eq:GrtGlt}) break the individual $\mbZ_2$ symmetries while preserving the combined symmetry $U_t^{\mbZ_2} U_b^{\mbZ_2}$:
\begin{equation}
 U_\ell^{\mbZ_2} \ket{G_\rt} = \ket{G_\lt}, \quad
U_\ell^{\mbZ_2} \ket{G_\lt} = \ket{G_\rt}, \quad \ell \in \{t,b\}; \;\;\;\quad\;\;
 U_t^{\mbZ_2} U_b^{\mbZ_2} \ket{G_a} = \ket{G_a}, \quad a \in \{\rt,\lt\}.
\end{equation}
Alternatively, this $\mbZ_2$ symmetry breaking can be detected by checking if the ground state space contains states from both $+1$ and $-1$ quantum numbers of the symmetry,  which is satisfied by $U_t^{\mbZ_2}$ and $U_b^{\mbZ_2}$.
This symmetry breaking is also detected by a local order parameter $X_{j;t} X_{j;b}$ (charged under both $U_t^{\mbZ_2}$ and $U_b^{\mbZ_2}$), which clearly has perfect long range order in $\ket{G_\rt}$ and $\ket{G_\lt}$.
On the other hand, the combined symmetry $U_t^{\mbZ_2} U_b^{\mbZ_2}$ in fact acts as an identity in the entire composite-spin sector, so it is not broken in this sector and hence in the ground states.
While we can loosely say that this pattern of the $(\mbZ_2)_t \times (\mbZ_2)_b$ symmetry breaking down to a single remaining $\mbZ_2$ is responsible for the appearance of the two-dimensional ground state manifold, further structures/energetics in the super-Hamiltonian by construction are responsible for the specific ground state wavefunctions and their exact degeneracy at finite system sizes, landing the system at a ``fine-tuned point" inside of this particular symmetry breaking phase.
\subsection{Global $U(1)$ Symmetry}
\label{subapp:superHsymms_U1}
We next consider the case of the global $U(1)$ symmetry from Sec.~\ref{subsec:U1symmetry}.
Here the inherited strong symmetries of the super-Hamiltonians are $U(1)$ symmetries associated with each individual leg, which can be implemented with unitaries $U^{U(1)}_t(\theta_t) \defn \exp(i\theta_t \sum_j Z_{j;t})$ and $U_b^{U(1)}(\theta_b) \defn \exp(i\theta_b \sum_j Z_{j;b})$ with $\theta_t, \theta_b \in [0,2\pi)$.
The ground states of the super-Hamiltonian form the ferromagnetic manifold of the composite spins,  and these $U(1)$ symmetries act non-trivially in this manifold, e.g., in the case of the basis $\oket{Q_m^z} \sim \sket{\widetilde{F}_m^x}$ of Eq.~(\ref{eq:HeisenbergGSop}), and are hence broken.
Alternatively, consider a different ground state basis $\{\sket{\widetilde{F}_m^z} \}$, i.e., composite-spin version of the ferromagnetic states in Eq.~(\ref{eq:HeisenbergGS}) with polarization axis $\alpha = z$.
The inherited $U(1)$ symmetries act on these as
\begin{equation}
U_\ell^{U(1)}(\theta_\ell) \sket{\widetilde{F}_m^z} = e^{i \theta_\ell (L-2m)} \sket{\widetilde{F}_m^z} ~, \quad \ell\in \{t,b\},
\end{equation}
and the presence of such non-trivial eigenvalues in the ground state manifold signifies $U(1)$ symmetry breaking.
On the other hand, the combined symmetry $U_t^{U(1)}(\theta) U_b^{U(1)}(-\theta)$ acts trivially in the composite spin subspace and hence in the ground state manifold, hence it is not broken, i.e., we have only partial breaking of the formal $U(1)_t \times U(1)_b$ symmetry.
In physical terms, this symmetry breaking represents a quantum phase where the charges from the top and bottom legs are bound and the resulting composite particle is condensed, while they are individually gapped.
To give a more precise description of the character of the condensate, we note
that $\sket{\widetilde{F}_m^z}$ is an equal-weight superposition of all configurations with $m$ bosonic composite particles represented by the $\tdn$ composite spins in the ``vacuum" of $\tup$ composite spins, and such a wavefunction represents a ``perfect superfluid"--- i.e.,  Bose-Einstein Condensate (BEC)---of such bosons.
Indeed, correlations of a local order parameter $\widetilde{S}_j^+ \defn S_{j;t}^+ S_{j;b}^+ = \sket{\tup}\sbra{\tdn}$ [charged with respect to both $U(1)$'s] in terms of the boson density $\rho$ reads:
\begin{equation}
\sbra{\widetilde{F}^z_m} \widetilde{S}_j^+ \widetilde{S}_{j'}^- \sket{\widetilde{F}^z_m} = \frac{\binom{L-2}{m-1}}{\binom{L}{m}} = \frac{m(L-m)}{L(L-1)} \xrightarrow{L \rightarrow \infty} \rho(1-\rho) ~,\;\;\; \rho \defn \frac{m}{L}.
\end{equation}
The key observation is that this is non-zero for any $0 < \rho < 1$ and is independent of the separation between the points $j$ and $j'$, and this is true in any dimension.
This is unlike generic superfluid wavefunctions where the correlations would approach the non-zero limit in a power-law fashion in dimension $d>1$.
Related BEC vs generic superfluid difference shows up also in the excitation spectrum: quadratically dispersing excitations for our super-Hamiltonians as discussed in Sec.~\ref{subsec:U1gapless} vs linearly dispersing Goldstone modes in generic superfluids.
The perfect superfluid order revealed by all these perspectives as well as the exact degeneracy among the ground states are due to the further structures/energetics present in the super-Hamiltonian by construction, as discussed in the introductory part of this Appendix.
\subsection{Isolated QMBS}
\label{subapp:superHsymms_isoQMBS}
It is also curious to examine the fate of the inherited symmetries in the case of the isolated scar of Sec.~\ref{subsubsec:isolatedQMBS}.
Here the inherited symmetries can be viewed as highly non-local $\mbZ_2$ symmetries associated with each leg and specified by unitaries $U_t^{\text{iso}} \defn (\mathbb{1} - 2\ketbra{\Phi})_t \otimes \mathbb{1}_b$ and $U_b^{\text{iso}} \defn \mathbb{1}_t \otimes (\mathbb{1} - 2\ketbra{\Phi})_b$.
Simple analysis shows that in the ground state manifold spanned by the (non-orthonormal) basis of Eq.~(\ref{eq:isoscarGS}), both $U_t^{\text{iso}}$ and $U_b^{\text{iso}}$ take eigenvalues $\pm 1$, while $U_t^{\text{iso}} U_b^{\text{iso}}$ acts trivially (true on the full composite spin sector).
This structure of the eigenvalues of the symmetries is similar to the super-Hamiltonians constructed in the case of the global $\mbZ_2$ symmetry considered in Sec.~\ref{subapp:superHsymms_Z2}, and one may loosely say that $U_t^{\text{iso}}$ and $U_b^{\text{iso}}$ are broken while $U_t^{\text{iso}} U_b^{\text{iso}}$ is preserved.
However, in this case there is no local order parameter that would have non-trivial charge under these symmetries and that could detect this ``symmetry breaking", which is then not a very useful concept here.

\section{Extraneous Features of Specific Super-Hamiltonians}
\label{app:extraneous_features}
Many of the super-Hamiltonians in the main text have some additional extraneous features, sometimes allowing full or partial solvability beyond the expected exact ground states.
These features in fact depend on the specific choice of the bond algebra generators $\{ \hH_\alpha \}$ used to define the super-Hamiltonian $\hmP$ in Eq.~(\ref{eq:psdLiouv}), and in this Appendix we comment on this dependence.
Most importantly, for a fixed bond algebra $\mA$, the formal superoperator-space symmetries of the set of the Liouvillian super-operators $\{ \hmL_\alpha \}$ can depend on the choice of the generators $\{ \hH_\alpha \}$.
That is, for different sets of generators of  $\mA = \lgen \{ \hH_\alpha \} \rgen = \lgen \{ \hH_\beta^\prime \} \rgen$, the corresponding sets of the Liouvillian super-operators $\{ \hmL_\alpha \}$ and $\{ \hmL_\beta^\prime \}$ can generate different super-operator algebras, and their commutants can be different and larger than the set of formal symmetries discussed in App.~\ref{app:superHsymms} that are always present.
Examples of this include instances where the composite-spin subspace is invariant under the action of the super-Hamiltonians (many cases in the main text), the appearance of the $SU(2)$ symmetry in the case of $\hmP_{U(1)|\text{comp}}$ in Eq.~(\ref{eq:PU1comp}), or the appearance of the $SU(q^2)$ symmetry in the case of $\hmP_{SU(q)}$ in Eq.~(\ref{eq:suqsuperhamiltonian}).
Furthermore, the couplings with which $\{ \hmL_\alpha^2 \}$ enter in $\hmP$ also matter for some extraneous features as well as for lattice symmetries of the super-Hamiltonians.
Nevertheless, for the problem of finding the commutant of $\mA$, we are guaranteed that the commutant is the exact ground state manifold of any super-Hamiltonian constructed from any set of generators of $\mA$, so there is no issue here.
On the other hand, one may worry whether the low-energy spectra of such specific ``more symmetric'' super-Hamiltonians correspond to slow dynamical modes in more generic systems.
We expect this is true, namely, that possible additional features in the super-Hamiltonians do not change the qualitative character of the low-energy excitations, which we think is tied to the structure of the exact ground state space, and we demonstrate this in some cases in the main text, while here we provide more general comments.
First, in all cases, we can progressively suppress the additional features by adding more terms from the bond algebra (e.g., combining different sets of generators), and this would add more positive-semidefinite terms to the super-Hamiltonian.
By construction, the exact ground state manifold would remain unchanged, while all excitation energies would only increase.
In particular, the gapped cases would remain gapped (which we can then consider as a proof of generic gapfullness), while in the gapless cases the presented low excitation energies of specific super-Hamiltonians would at least provide exact lower bounds on the excitation energies of the modified super-Hamiltonians.
We further expect that the corresponding specific eigenstates could be used as trial states and would also provide variational upper bounds on the excitation energies of the modified super-Hamiltonians that would retain the same qualitative character as before, e.g., would give similar dispersion laws.
In the main text, we show evidence for this in the case of $\hmP_{U(1)}$.
In some cases the additional features allowing solvability beyond the exact ground states are like the integrability of the Heisenberg chain.
It is well established that low energy excitations of such integrable models capture qualitatively the physics of more generic models in the same phase.
Finally, our main confidence that the additional features are not qualitatively important comes from the fact that the specific super-Hamiltonians arise naturally as descriptions of properties of concrete Brownian circuits that by themselves do not look fine-tuned, e.g., each random instance is not solvable or special in any way.
The additional features in the super-Hamiltonians can be loosely thought of as coming from some choices of taking the simplest generators as well as convenient distributions of the random couplings, and such choices should not affect qualitative long-time hydrodynamic properties in the Brownian circuits.
\section{Details on Quantum Many-Body Scar Super-Hamiltonians}\label{app:QMBSliouv}
In this appendix, we provide some details on the ground states of the super-Hamiltonians $\hmP_{\text{scar}}$ that appear in the study of QMBS.
\subsection{Isolated QMBS}\label{app:QMBSliouviso}
In the case of a single isolated QMBS given by $\ket{\Phi} = \ket{\up\ \up\ \dots\ \up}$, we can choose $R_{[j]} \defn R_j = (1 - \sigma_j^z)/2 = \ketbra{\dn}_j$.
It is easy to see that the common kernel of the $R_j$'s contains a single state $\ket{\Phi}$.
We start with the bond algebra
\begin{equation}
    \mA_{\text{iso}} \defn \lgen \{ R_j \sigma^\alpha_{j+1}, \sigma^\alpha_j R_{j+1}\} \rgen,
\label{eq:Aiso}
\end{equation}
where for simplicity we choose PBC for the bond algebra generators, although similar results carry forward for the OBC case.
The full super-Hamiltonian of Eq.~(\ref{eq:RSigm}) reads
\begin{align}
    \hmP_{\text{iso}} &= 8\sumal{j}{}{(R_{j;t} - R_{j;b})^2} + 8\sumal{j}{}{\left(R_{j;t} R_{j;b} [1 - \ketbra{\iota}{\iota}]_{j+1} + [1 - \ketbra{\iota}{\iota}]_{j} R_{j+1;t} R_{j+1;b}\right)}\nn \\
    &=2\sumal{j}{}{(\sigma^z_{j;t} - \sigma^z_{j;b})^2} 
+ 8 \sumal{j}{}{\left(\ketbra{\dn}_{j;t} \ketbra{\dn}_{j;b} [1 - \ketbra{\iota}{\iota}]_{j+1} + [1 - \ketbra{\iota}{\iota}]_{j} \ketbra{\dn}_{j+1;t} \ketbra{\dn}_{j+1;b}\right)}.  
\label{eq:Liouvfullexp}
\end{align}
Note that the $(\sigma^z_{j;t} - \sigma^z_{j;b})^2$ terms enforce that the ground state is in the sector of composite spins defined in Eq.~(\ref{eq:compspins}), similar to the $\mbZ_2$ and $U(1)$ cases.
\subsubsection{Ground States}
We do not need to use the full structure of the super-Hamiltonian to obtain the ground states. 
According to Eq.~(\ref{eq:Rcond2}), the ground state space satisfies 
\begin{gather}
\ketbra{\dn}_{j; \ell}[1 - \ketbra{\iota}{\iota}]_{j+1} \ket{\Psi} = 0,\;\;\; [1 - \ketbra{\iota}{\iota}]_{j} \ketbra{\dn}_{j+1; \ell}\ket{\Psi} = 0,\;\;\;\ell \in \{t, b\}\nn\\
\implies \sket{\tdn,\tlt} \! \sbra{\tdn,\tlt}_{j,j+1} \ket{\Psi} = 0~, \quad
\sket{\tlt,\tdn} \! \sbra{\tlt,\tdn}_{j,j+1} \ket{\Psi} = 0~,
\label{eq:psiconst}
\end{gather}
where in the second line we have expressed the conditions in terms of composite spins defined in Eqs.~(\ref{eq:compspins}) and (\ref{eq:ladderopcorr}), and in replacing $[1 - \ketbra{\iota}]$ by $\ketbra{\tlt}$ we have used that the ground states $\ket{\Psi}$ are in the composite spin sector.
These conditions highly constrain the structure of $\ket{\Psi}$. 
In particular, suppose $\ket{\Psi}$ is decomposed as
\begin{equation}
    \ket{\Psi} = \sumal{\alpha}{}{\ket{u^{\alpha}}_{[j,j+1]} \otimes \ket{v^\alpha}_{\text{rest}}}, 
\label{eq:Psidecomp}
\end{equation}
where the supports of each part of the wavefunction along the ladder are indicated in the subscript, with ``rest" denoting the complement to $[j,j+1]$, and $\{\ket{v^\alpha}_{\text{rest}} \}$ form a linearly independent set.
Such a decomposition always exists, Schmidt decomposition being one example, but we will only require the linear independence of $\{\ket{v^\alpha}_{\text{rest}} \}$ and not orthonormality.
The conditions of Eq.~(\ref{eq:psiconst}) and the linear independence of $\{\sket{v^\alpha_{\text{rest}}}\}$ imply that both of the following should hold
\begin{equation}
\sket{\tdn,\tlt} \! \sbra{\tdn,\tlt}_{j,j+1} \ket{u^\alpha}_{j,j+1} = 0~, \quad
\sket{\tlt,\tdn} \! \sbra{\tlt,\tdn}_{j,j+1} \ket{u^\alpha}_{j,j+1} = 0 \quad
\implies 
\ket{u^\alpha}_{j,j+1} \in \text{span}\{\sket{\tup,\tup}_{j,j+1}, \sket{\trt,\trt}_{j,j+1} \} ~.
\label{eq:schmidtconditions}
\end{equation}
Hence, $\ket{\Psi}$ can be written as
\begin{equation}
\ket{\Psi} = \sket{\tup,\tup}_{j,j+1} \otimes \sket{\Upsilon}_{\text{rest}} + \sket{\trt,\trt}_{j,j+1} \otimes \sket{\Theta}_{\text{rest}} 
\end{equation}
with some states $\sket{\Upsilon}_{\text{rest}}$ and $\sket{\Theta}_{\text{rest}}$ on the complement to $[j,j+1]$.
Moving on to requiring Eq.~(\ref{eq:schmidtconditions}) on the next pair of sites $[j+1,j+2]$, since $\sket{\tup}_j$ and $\sket{\trt}_j$ are linearly independent, we can apply similar argument to independently $\sket{\tup}_{j+1} \otimes \sket{\Upsilon}_{\text{rest}}$ and $\sket{\trt}_{j+1} \otimes \sket{\Theta}_{\text{rest}}$.
For example, we get that
\begin{equation}
\sket{\tup}_{j+1} \otimes \sket{\Upsilon}_{\text{rest}} = \sket{\tup,\tup}_{j+1,j+2} \otimes \sket{\Upsilon'}_{\text{rest}'} + \sket{\trt,\trt}_{j+1,j+2} \otimes \sket{\Theta'}_{\text{rest}'} 
\implies \sket{\tup}_{j+1} \otimes \sket{\Upsilon}_{\text{rest}} = \sket{\tup,\tup}_{j+1,j+2} \otimes \sket{\Upsilon'}_{\text{rest}'} ~,
\end{equation}
where we have used linear independence of $\sket{\tup}_{j+1}$ and $\sket{\trt}_{j+1}$ and ``$\text{rest}'$" denotes complement to $[j,j+1,j+2]$.
In all, requiring Eq.~(\ref{eq:schmidtconditions}) on all pairs of neighboring sites, we can conclude that $\ket{\Psi}$ is spanned by $\sket{\tup,\tup, \dots, \tup}$ and $\sket{\trt,\trt, \dots, \trt}$, which correspond to operators $\oket{\mathds{1}}$ and $\oket{\ketbra{\Phi}}$.
In the original language, this means that the commutant of Eq.~(\ref{eq:Aiso}) is given by
\begin{equation}
    \mC_{\text{iso}} = \lgen \ketbra{\Phi} \rgen,
\label{eq:Ciso}
\end{equation}
where the $\mathds{1}$ is implicit in the notation $\lgen \cdots \rgen$.
Note that while we have included two types of bond algebra generators in Eq.~(\ref{eq:Aiso}), it is usually possible to choose a subset of them and still recover the same commutant of Eq.~(\ref{eq:Ciso}), although the analytical analysis might not be so straightforward.
\subsubsection{Gap and Low-Energy Excitations}
We then study the low-energy spectrum of $\hmP_{\text{iso}}$ of Eq.~(\ref{eq:Liouvfullexp}).
Since $(\sigma^z_{j;t} - \sigma^z_{j;b})^2$ commutes with all other terms in the Hamiltonian, configurations of composite spins form a closed subspace for $\hmP_{\text{iso}}$.
The violation of such a composite spin costs a constant amount of energy, hence we can work in the space of composite spins to determine if $\hmP_{\text{iso}}$ may have a smaller gap. 
Restricted to the space of composite spins, the action of $\hmP_{\text{iso}}$ reads
\begin{equation}
\hmP_{\text{iso}|\text{comp}} = 2 \sumal{j}{}{[(1 - \tZ_j) (1 - \tX_{j+1}) + (1 - \tX_j)(1 - \tZ_{j+1})]},\;\;\;
\tZ \defn \sketbra{\tup}{\tup} - \sketbra{\tdn}{\tdn},\;\;\;
\tX \defn \sketbra{\tup}{\tdn} + \sketbra{\tdn}{\tup}.
\end{equation}
This can be simplified to
\begin{align}
\hmP_{\text{iso}|\text{comp}} &= 2 \sumal{j}{}{[2 - (\tX_j + \tZ_j + \tX_{j+1} + \tZ_{j+1}) + \tX_j \tZ_{j+1} + \tZ_j \tX_{j+1}]} \nn \\
&= 2 \sumal{j}{}{[2 - \sqrt{2}(\tX'_j + \tX'_{j+1}) + \tX'_j \tX'_{j+1} - \tY'_j \tY'_{j+1}]},\;\;\;
\tX'_j \defn \frac{\tX_j + \tZ_j}{\sqrt{2}},\;\;\;
\tY'_j \defn \frac{\tX_j - \tZ_j}{\sqrt{2}},
\label{eq:Pisocomp}
\end{align}
where in the second step we have performed a basis transformation for the composite spins. 
For even system size and PBC, this further maps to 
\begin{equation}
    \hmP_{\text{iso}|\text{comp}} = 4\left[L  - \sqrt{2}\sumal{j}{}{\tX'_j} + \frac{1}{2}\sumal{j}{}{(\tX'_j \tX'_{j+1} + \tY'_j \tY'_{j+1})}\right],
\label{eq:isoeffective}
\end{equation}
where we have used the bipartiteness of the lattice to transform spins on even sites such that $\tY'_j \rt (-1)^j \tY'_j$ by rotating around the $\tX'$ axis. 
Note that this is simply the antiferromagnetic XX model with a longitudinal field with the specific value,  or equivalently the transverse field Ising model with nearest neighbor interactions with the specific field and interaction values, such that the product states $\ket{\trt\ \trt\ \cdots\ \trt}$ and $\sket{\tup\ \tup\ \cdots\ \tup}$ (in the original $\tilde{X},\tilde{Z}$ axes) are exact ground states. 
This model has in fact been studied in the earlier literature.
For example, Eq.~(\ref{eq:isoeffective}) is known to be dual to the well-studied ANNNI models, and in the phase diagram obtained by \cite{dutta2003gapless}, this appears to be in the gapped ferromagnetic phase.
Our own exact diagonalization study of the specific model Eq.~(\ref{eq:isoeffective}) in PBC indeed finds that there are two exactly degenerate ground states that spontaneously break the $\mbZ_2$ symmetry generated by $\prod_{j}{\tX'_j}$, separated by a gap between $0.5$ and $1.$ from the rest of the spectrum (thus, the lowest-energy excitation of $\hmP_{\text{iso}}$ indeed lies in the composite spin sector).
Moreover, the Hamiltonian of Eq.~(\ref{eq:Pisocomp}) exactly maps onto a frustration-free model that lies on the so-called Peschel-Emery line~\cite{peschel1981calculation, mahyaeh2018exact}. 
A relatively recent work \cite{katsura2015exact} proves that the OBC version of the Hamiltonian Eq.~(\ref{eq:Pisocomp}) is gapped, by performing a Jordan-Wigner mapping to an interacting Majorana chain in OBC and exhibiting a deformation path to a gapped free-fermion Hamiltonian without closing a gap.
The same argument can be carried out also directly in the spin model in the PBC system as well.
\subsection{Tower of QMBS}\label{app:QMBSliouvtower}
We now illustrate the example of the ferromagnetic tower of QMBS.
Consider $R_{[j]} \defn R_{j, j+1} = \frac{1}{4} - \vec{S}_j \cdot \vec{S}_{j+1} = \frac{1}{2} (\ket{\up\dn} - \ket{\dn\up})(\bra{\up\dn} - \bra{\dn\up})_{j,j+1} = \frac{1}{2}(1 - P_{j,j+1}^{(2)})$, where the target scar manifold contains the entire ferromagnetic tower of QMBS $\ket{\Phi_n} \sim (S^-_{\tot})^n\ket{\up\ \up\ \dots\ \up}$, $0 \leq n \leq L$.
We start with the bond algebra similar to Eq.~(\ref{eq:Aiso}), 
\begin{equation}
    \mA_{\text{tower}} \defn \lgen \{ R_{j, j+1} \sigma^\alpha_{j+2}, \sigma^\alpha_{j-1} R_{j,j+1}\} \rgen,
\label{eq:Atower}
\end{equation}
and for simplicity choose PBC for the bond algebra generators.
Any operator constructed out of the generators of $\mA_{\text{tower}}$ contains the ferromagnetic multiplet as degenerate QMBS~\cite{moudgalya2022exhaustive}.
The full super-Hamiltonian in this case reads
\begin{equation}
    \hmP_{\text{tower}} = 8\sumal{j}{}{(R_{[j, j+1]; t} - R_{[j,j+1]; b})^2} + 8\sumal{j}{}{\left(R_{[j,j+1]; t} R_{[j,j+1]; b} [1 - \oketbra{\iota}{\iota}]_{j+1} + [1 - \oketbra{\iota}{\iota}]_{j} R_{[j+1,j+2]; t} R_{[j+1,j+2]; b}\right)}.  
\label{eq:Liouvfulltower}
\end{equation}
Since $(R_{[j, j+1]; t} - R_{[j,j+1]; b})^2 = \frac{1}{2}(1 - P_{[j,j+1],t}^{(2)}P_{[j,j+1],b}^{(2)})$, the ground states must be symmetric under exchange of the states on the nearby rungs,
\begin{equation}
P_{[j,j+1]}^{\text{rung}} \ket{\Psi} = \ket{\Psi} \iff [1 - P^{\text{rung}}_{[j,j+1]}]\ket{\Psi} = 0 ~, \quad P_{[j,j+1]}^{\text{rung}} \defn P_{[j,j+1];t}^{(2)} P_{[j,j+1];b}^{(2)}
\label{eq:exchcondition}
\end{equation}
Furthermore, according to Eq.~(\ref{eq:Rcond2}), the ground state space satisfies
\begin{equation}
R_{[j, j+1]; \ell}[1 - \ketbra{\iota}{\iota}]_{j+2} \ket{\Psi} = 0,\;\;\; [1 - \ketbra{\iota}{\iota}]_{j} R_{[j+1, j+2]; \ell}\ket{\Psi} = 0,\;\;\;\ell \in \{t, b\}.
\label{eq:psiconsttower}
\end{equation}
As in the isolated QMBS, these conditions highly constrain the structure of the wavefunction $\ket{\Psi}$. 
Generalizing Eq.~(\ref{eq:Psidecomp}), given a region $A$ and its complement $A^c$, suppose we decompose $\ket{\Psi}$ as
\begin{equation}
    \ket{\Psi} = \sumal{\alpha}{}{\ket{u^\alpha}_A \otimes \ket{v^\alpha}_{A^c}}, 
\end{equation}
where the supports of each part of the wavefunction are indicated in the subscript, and $\{\ket{v^\alpha}_{A^c}\}$ are linearly independent.
Then, if $\ket{\Psi}$ is annihilated by some operators acting entirely within the region $A$, it follows that each $\ket{u^\alpha}_A$ is annihilated by the same operators, and we can write
\begin{equation}
    \ket{\Psi} = \sumal{\gamma}{}{\ket{e^\gamma}_A \otimes \ket{w^\gamma}_{A^c}},
\end{equation}
where $\{ \ket{e^\gamma}_A \}$ is a complete basis in the common kernel of the annihilators acting within $A$, and $\{ \ket{w^\gamma}_{A^c} \}$ are some new states on $A^c$ that are not required to be linearly independent.
The above are precise statements about the constraints on local ``parts" of such a $\ket{\Psi}$, and are used repeatedly (often implicitly) below.
Before proceeding with proofs, we introduce some shorthand notation.
On a segment $[j,j+r]$, we denote the common kernel of all $\{R_{[k,k+1];\ell}\}$ with support completely within the segment as $K_{[j,j+r]}$.
It is easy to see that it is spanned by $\{ \ket{\Phi_m}_{[j,j+r];t} \otimes \ket{\Phi_n}_{[j,j+r];b} \}$, where $\{ \ket{\Phi_m}_{[j,j+r]} \}$ is the full QMBS set in the original problem on the segment $[j,j+r]$.
We further denote a complete basis of $K_{[j,j+r]}$ as $\{ \ket{\Gamma^a}_{[j,j+r]} \}$. 
On the same segment, the common kernel of all $\{[1 - \ketbra{\iota}]_{k}\}$ acting within the segment is spanned by a single state for which we introduce a shorthand $\ket{\iota}_{[j,j+r]} \defn \otimes_{k=j}^{j+r} \ket{\iota}_{k}$.
\begin{lemma}
The common kernel of $R_{[j,j+1];\ell}$ and $[1 - \ketbra{\iota}]_k$ is trivial for any $\ell \in \{t, b\}$ if $k \in \{j,j+1\}$.
As corrollaries, $R_{[j,j+1]} \ket{\iota}_{[j,j+1]} \neq 0$ and $\ket{\iota}_{[j,j+r]} \notin K_{[j,j+r]}$ for any $r$.
\label{lem:comker}
\end{lemma}
\begin{proof}
We illustrate the proof for $k = j$, the proof for $k = j+1$ follows in the same way. 
Any state $\ket{v}_{[j,j+1]}$ on rungs $j$ and $j+1$ in the kernel of $[1 - \ketbra{\iota}]_j$ can without loss of generality be expressed as
\begin{equation}
    \ket{v}_{[j,j+1]} = \ket{\iota}_j \otimes \ket{\up}_{j+1; \ell} \otimes \sket{w^\up}_{j+1;\overline{\ell}} 
    + \ket{\iota}_j \otimes \ket{\dn}_{j+1; \ell} \otimes \sket{w^\dn}_{j+1;\overline{\ell}} ~,
\end{equation}
where $\sket{w^\up}_{j+1;\overline{\ell}}$ and $\sket{w^\dn}_{j+1;\overline{\ell}}$ are some states on the site $(j+1;\overline{\ell})$, where $\overline{\ell}$ is the compliment of $\ell$ in $\{t,b\}$.
The action of $R_{[j,j+1]; \ell}$ on $\ket{v}_{[j,j+1]}$ of this form reads
\begin{equation}
    R_{[j,j+1];\ell}\ket{v}_{[j,j+1]} = \frac{1}{2\sqrt{2}}(\ket{\up\dn} - \ket{\dn\up})_{[j,j+1];\ell}\otimes (-\ket{\dn}_{j;\overline{\ell}} \otimes \ket{w^\up}_{j+1, \overline{\ell}} + \ket{\up}_{j;\overline{\ell}} \otimes \ket{w^\dn}_{j+1, \overline{\ell}}).
\end{equation}
It is then easy to see that $R_{[j,j+1]; \ell}$ vanishes on $\ket{v}_{[j,j+1]}$ only if 
$\sket{w^\up}_{j+1;\overline{\ell}} = \sket{w^\dn}_{j+1;\overline{\ell}} = 0$, which in turn means that $\ket{v}_{[j,j+1]} = 0$.
It then directly follows that $R_{[j,j+1]}\ket{\iota}_{[j,j+1]} \neq 0$, since $\ket{\iota}_{[j,j+1]}$ clearly is in the kernel of $[1 - \ketbra{\iota}_j]$.  
This can also be directly verified, a simple calculation gives $R_{[j,j+1];t} \ket{\iota}_{[j,j+1]} = \frac{1}{4} \ket{S}_{[j,j+1];t} \otimes \ket{S}_{[j,j+1];b} \neq 0$, where $\ket{S}$ is a singlet state for two spins involved.
This also means $\ket{\iota}_{[j,j+r]} \notin K_{[j,j+r]}$, and $\ket{\iota}_{[j,j+r]}$ is linearly independent from $\{ \ket{\Gamma^a}_{[j,j+r]} \}$.
\end{proof}
\begin{lemma}
All the states satisfying Eqs.~(\ref{eq:exchcondition}) and (\ref{eq:psiconsttower}) in region $[j,j+r]$ are in $\text{span}\{\sket{\Gamma^a}_{[j,j+r]}, \sket{\iota}_{[j,j+r]}\}$. 
\end{lemma}
\begin{proof}
Clearly, $\{\ket{\Gamma^a_{[j,j+r]}}\}$ and $\ket{\iota}_{[j,j+r]}$ are annihilated by all $\{1 - P_{k,k+1}^{\text{rung}}\}$ and $\{R_{[k,k+1];\ell}(1-\ketbra{\iota})_{k+2}\}$ acting within the segment, hence they satisfy Eqs.~(\ref{eq:exchcondition}) and (\ref{eq:psiconsttower}).
We now show that any other state $\ket{\Psi}$ satisfying these conditions is spanned by the above states.
Starting with $r=1$, i.e., two rungs of the ladder at $j,j+1$, we show that just requiring annihilation by $[1 - P_{[j,j+1]}^{\text{rung}}]$ enforces $\ket{\Psi}$ to lie in the span of $K_{[j,j+1]}$ and $\ket{\iota}_{[j,j+1]}$. 
The kernel of $[1 - P_{[j,j+1]}^{\text{rung}}]$ is 10-dimensional consisting of all states that are symmetric under the rung exchanges, and its complete basis is $\{ \ket{T_m}_{[j,j+1];t} \otimes \ket{T_n}_{[j,j+1];b}, m,n \in \{0,\pm 1\} \} \cup \ket{S}_{[j,j+1];t} \otimes \ket{S}_{[j,j+1];b}$, where $\ket{T_m}$ and $\ket{S}$ are respectively triplet and singlet states for two spins involved.
For two spins, the triplet space $\{ \ket{T_m}_{[j,j+1]} \}$ is the same as $\{ \ket{\Phi_m}_{[j,j+1]}\}$, so the 9 states $\{ \ket{T_m}_{[j,j+1];t} \otimes \ket{T_n}_{[j,j+1];b} \}$ can be replaced by $\{ \ket{\Gamma^a}_{[j,j+1]} \}$.
Since $\ket{\iota}_{[j,j+1]}$ is also annihilated by $[1 - P_{[j,j+1]}^{\text{rung}}]$, and, being linearly independent of $\{ \ket{\Gamma_a}_{[j,j+1]} \}$, it can replace $\ket{S}_{j,j+1;t} \otimes \ket{S}_{j,j+1;b}$ in the constructed complete basis of the kernel of $1 - P_{j,j+1}^{\text{rung}}$.
Moving to $r=2$, for the appropriate parts of $\ket{\Psi}$ we can write
\begin{equation}
\ket{\Psi}_{[j,j+2]} = \sum_a \ket{\Gamma^a}_{[j,j+1]} \otimes \ket{v^a}_{j+2} + \ket{\iota}_{[j,j+1]} \otimes \ket{w}_{j+2} ~,
\end{equation}
where $\{ \ket{v^a}_{j+2} \}$ and $\{ \ket{w}_{j+2} \}$ are some states on the rung at $j+2$.
Now consider annihilation by $R_{[j,j+1];\ell} (1-\ketbra{\iota})_{j+2}$.
Since $R_{[j,j+1];\ell}$ annihilates all $\{ \ket{\Gamma^a}_{[j,j+1]} \}$, we deduce that
\begin{equation}
R_{[j,j+1];\ell} (1 - \ketbra{\iota})_{j+2} \; \ket{\iota}_{[j,j+1]} \otimes \ket{w}_{j+2} = 0 ~.    
\end{equation}
However, since $R_{[j,j+1];\ell} \ket{\iota}_{[j,j+1]} \neq 0$ following Lem.~\ref{lem:comker}, we have $\ket{w}_{j+2} = c \ket{\iota}_{j+2}$, hence
\begin{equation}
\sum_a \ket{\Gamma^a}_{[j,j+1]} \otimes \ket{v^a}_{j+2} = 
\ket{\psi}_{[j,j+2]} - c \ket{\iota}_{[j,j+2]} ~.
\end{equation}
Since $\ket{\psi}_{[j,j+2]}$ must be symmetric under the exchange of rungs $j+1$ and $j+2$ according to Eq.~(\ref{eq:exchcondition}), and $\ket{\iota}_{[j,j+2]}$ clearly is, the L.H.S.\ should be symmetric under this and all other rung exchanges in $[j,j+2]$.
Since the L.H.S. is annihilated by $R_{[j,j+1];\ell}$ it must then be also annihilated by $R_{j+1,j+2;\ell}$, i.e., it belongs to the span of $\{ \ket{\Gamma_{a'}}_{[j,j+2]} \}$.
This completes the proof for the segment $[j,j+2]$.
Essentially the same steps can then be used for an inductive proof going from $[j,j+r]$ to $[j,j+r+1]$, considering annihilation by $R_{[j+r-1,j+r];\ell} (1 - \ketbra{\iota})_{j+r+1}$ to peel off the $\ket{\iota}_{[j,j+r+1]}$ contribution, and then for the remainder deducing the symmetry under the rung exchange at $j+r$ and $j+r+1$ and proving that it belongs to $K_{[j,j+r+1]}$.
\end{proof}
Hence the only states on the full chain satisfying Eqs.~(\ref{eq:exchcondition}) and (\ref{eq:psiconsttower}) are
\begin{equation}
    \ket{\Psi} = \ket{\Phi_m}_t \otimes \ket{\Phi_n}_b = \oket{\ketbra{\Phi_m}{\Phi_n}} \;\;\text{or}\;\;\ket{\Psi} = \otimes_{j = 1}^L{\ket{\iota}}_j = \frac{1}{2^{\frac{L}{2}}}\oket{\mathds{1}}.  
\label{eq:schmidttower}
\end{equation}
In the original language, this means that the commutant of Eq.~(\ref{eq:Aiso}) is given by
\begin{equation}
    \mC_{\text{tower}} = \lgen \{\ketbra{\Phi_m}{\Phi_n}\} \rgen,
\end{equation}
where the $\mathds{1}$ is implicit in the notation $\lgen \cdots \rgen$.
\section{Eigenstates of the \texorpdfstring{$t-J_z$}{} Super-Hamiltonian}\label{app:tracerdiffusion}
In this Appendix, we provide details on the eigenstates of the super-Hamiltonian $\hmP_{t-J_z}$, and on the computation of the weights of the operator $\oket{S^z_j}$ on these eigenstates.
As discussed in Sec.~\ref{subsubsec:fragmentation}, understanding the cumulative weight function  $\Omega_{S^z_j}(E)$ is key to understanding the late-time behavior of the ensemble-averaged autocorrelation function $\overline{C_{S^z_j}(t)}$.
However, we are only able to analytically determine the weights corresponding to the ground states and the ``spin-wave" excited states, both of which ultimately vanish in the thermodynamic limit.
It is then crucial to include contributions from the higher excited states to analytically reproduce the results of Fig.~\ref{fig:tJzweight},  hence we also discuss the general setup for this weight calculation although we are only able to implement this numerically in general.
\subsection{General Setup}
To determine the low-energy excitations of $\hmP_{t-J_z}$ that have a non-zero overlap on $\oket{S^z_j}$, it is sufficient to work in the composite spin sector on the ladder (see Sec.~\ref{subsubsec:fragmentation} for discussions on this).
A convenient computational basis in terms of the composite spins $\{\tup, \tdn, \tzero\}$ defined in Eq.~(\ref{eq:tJzsuperspins}) is of the form
\begin{equation}
    \{\sket{\ttau^{(1)}_{j_1} \ttau^{(2)}_{j_2} \cdots \ttau^{(m)}_{j_m}}\},\;\;\; \ttau^{(q)} \in \{\trt, \tlt\},\;\;\sket{\trt} \defn \frac{\sket{\tup} + \sket{\tdn}}{\sqrt{2}},\;\;\;\sket{\tlt} \defn \frac{\sket{\tup} - \sket{\tdn}}{\sqrt{2}},
\label{eq:worddefn}
\end{equation}
where the subscripts indicate the positions of the $\ttau^{(q)}$'s and the rest of the sites are assumed to be occupied by the $\tzero$'s.
Note that we are working in the basis with the spins $\trt$ and $\tlt$ instead of $\tup$ and $\tdn$ since we are ultimately interested in the overlap with the operator $\oket{S^z_j}$, which maps onto the composite spin $\ket{\tlt}_j$ on the ladder.
The ground state space of $\hmP_{t-J_z|\text{comp}}$, shown in Eq.~(\ref{eq:PtJzGS}), in this basis is spanned by the equal weight superpositions of the states with fixed pattern $\tau^{(1)},\tau^{(2)},\dots,\tau^{(m)}$, i.e.,
\begin{equation}
    \sket{G^{\tau^{(1)} \cdots \tau^{(m)}}} = \sumal{j_1 < \cdots < j_m}{}{\sket{\ttau^{(1)}_{j_1} \cdots \ttau^{(m)}_{j_m}}}.
\label{eq:tJzcompGS}
\end{equation}
In the operator language, these ground states are the ``words" defined in App.~B in \cite{moudgalya2021hilbert}, which were shown to form an orthogonal basis for the commutant algebra $\mC_{t-J_z}$ for OBC, and were used to compute exact Mazur bounds.
We then study the excited states of $\hmP_{t-J_z|\text{comp}}$ in the computational basis of Eq.~(\ref{eq:worddefn}).  
We start with an eigenstate of the Heisenberg model $H_{\text{Heis}}$ of App.~\ref{app:Heiscanon} with $m$ $\dn$'s that has the form
\begin{equation}
    \ket{\lambda; m} = \sumal{j_1 < j_2 < \cdots < j_m}{}{C^\lambda_{j_1,  j_2,  \cdots, j_m}\ket{\dn_{j_1}\ \dn_{j_2}\  \cdots \dn_{j_m}}}, 
\label{eq:Heisexc}
\end{equation}
where $\lambda$ denotes the energy and $m$ denotes the number of $\dn$'s in the eigenstate, the subscripts of the $\dn$'s denote their positions, and the rest of the sites are assumed to be $\up$'s.
Utilizing the mapping between $H_{\text{Heis}}$ and $\hmP_{t-J_z|\text{comp}}$ and summarized in Eq.~(\ref{eq:tJzHeismapsummary}), we can write down $2^m$ degenerate eigenstates of the latter corresponding to each eigenstate $\ket{\lambda; m}$ of the former; these are of the form
\begin{equation}
    \sket{\lambda; \tau^{(1)} \cdots \tau^{(m)}} = \sumal{j_1 < \cdots < j_m}{}{C^\lambda_{j_1, \cdots, j_m}\sket{\tilde{\tau}^{(1)}_{j_1} \cdots \tilde{\tau}^{(m)}_{j_m}}},\;\;\;\tau^{(\ell)} \in \{\rt, \lt\}.
\label{eq:tJzexc}
\end{equation}
labelled by the fixed pattern $\tau^{(1)},\tau^{(2)},\dots,\tau^{(m)}$.
In all, given that there are $\binom{L}{m}$ Heisenberg eigenstates with $m$ $\dn$'s, we get a total of $\sum_{m = 0}^{L}{2^m \binom{L}{m}} = 3^L$ eigenstates of $\hmP_{t-J_z|\text{comp}}$, which as expected covers the entire Hilbert space of the composite spins.
To compute the overlap between these eigenstates and the operator $\oket{S^z_j}$, we note that this operator maps onto the following state on the ladder:
\begin{equation}
    \oket{S^z_j} \defn \ket{\zeta_j} = \prod_{k = 1}^{j-1}(\sqrt{2}\ket{\trt}_k + \ket{\tzero}_k) \otimes \sqrt{2}\ket{\tlt}_j \otimes \prod_{k = j+1}^L{(\sqrt{2}\ket{\trt}_k + \ket{\tzero}_k)}.
\label{eq:Zjexpression}
\end{equation}
The only configurations $\sket{\tilde{\tau}^{(1)}_{j_1} \cdots \tilde{\tau}^{(m)}_{j_m}}$ that $\ket{\zeta_j}$ has non-zero overlap with are the ones where $j_{\ell} = j$ for some $1 \leq \ell \leq m$, with $\tau^{(\ell)} = \lt$ and $\tau^{(k)} = \rt$ for $k \neq \ell$. 
Hence the only eigenstates of the form of Eq.~(\ref{eq:tJzexc}) that have a non-zero overlap with $\ket{\zeta_j}$ are those of the form
\begin{equation}
    \sket{\lambda; \rt^\alpha \lt \rt^\beta} = \sumal{j_1 < \cdots < j_{\alpha+\beta+1}}{}{C^\lambda_{j_1, \cdots, j_{\alpha+\beta+1}}\sket{\trt_{j_1} \cdots \trt_{j_\alpha} \tlt_{j_{\alpha+1}} \trt_{j_{\alpha+2}}\cdots \trt_{j_{\alpha+\beta+1}}}}, 
\label{eq:nonzerocont}
\end{equation}
where the sum is over $\{j_{\ell}\}$ for $1 \leq \ell \leq \alpha+\beta+1$, and ``$\tau^k$" in the pattern label for $\tau \in \{\rt, \lt\}$ denotes that $\tau$ is repeated $k$ times.
That is, in an $m$-spin pattern there is precisely one $\leftarrow$ whose position is parametrized by $\alpha,\beta$ with $\alpha+1+\beta=m$.
This is a total of $\sum_{m = 1}^{L}{m \binom{L}{m}} = L \times 2^{L-1}$ eigenstates that can have a non-zero overlap with the $\ket{\zeta_j}$.
We find that the overlap is generically non-zero for all such states, and such eigenstates are not necessarily only in the very low-energy part of the spectrum, e.g., in the one-magnon band.
Using Eqs.~(\ref{eq:Zjexpression}) and (\ref{eq:nonzerocont}), the weight of $\ket{\zeta_j}$ on the eigenstate $\ket{\lambda; {\rt^\alpha \lt \rt^\beta}}$ is given by [see Eq.~(\ref{eq:correlationavg}) for the definition]
\begin{equation}
    \frac{1}{3^L}|\sbraket{\zeta_j}{\lambda; {\rt^\alpha \lt \rt^\beta}}|^2 = \frac{2^{\alpha+\beta+1}}{3^L}\left|\sumal{j_1 < \cdots < j_{\alpha+\beta+1}}{}{C^\lambda_{j_1, \cdots, j_{\alpha+\beta+1}} \delta_{j_{\alpha+1}, j}}\right|^2.
\label{eq:overlapexpression}
\end{equation}
We are not able to use Eq.~(\ref{eq:overlapexpression}) to proceed analytically without any approximations. 
However, given the eigenstates of the Heisenberg model numerically, it is easy to use this expression to numerically compute the weights.
We employed this method to compute the cumulative weight $\Omega_{S^z_j}(E)$ shown in Fig.~\ref{fig:tJzweight} in the main text. 
\subsection{Spin-Wave Contribution}
We apply the results of the previous section to compute the overlap of $\oket{S^z_j}$ on the spin-wave excited states of $\hmP_{t-J_z|\text{comp}}$.
In the notation of Eq.~(\ref{eq:Heisexc}),  the spin-wave excited states for the OBC Heisenberg model, shown in Eq.~(\ref{eq:Heisenbergspinwavestates}), read
\begin{equation}
\ket{\lambda_k; m} \defn \frac{1}{\sqrt{\mM_{m,k}}}\sumal{j_1 < \cdots < j_m}{}{\sumal{\ell = 1}{m}{c_{j_{\ell},k}}}\ket{\dn_{j_1}\  \dn_{j_2}\ \cdots\  \dn_{j_m}}\;\;\implies\;\;
C^{\lambda_k}_{j_1, \cdots, j_m} = \frac{1}{\sqrt{\mM_{m,k}}}\times \sumal{\ell = 1}{m}{c_{j_{\ell},k}},
\label{eq:spinwavetranslation}
\end{equation}
where $k \in \frac{\pi n}{L}$ for $1 \leq n \leq L-1$, and $c_{j_\ell, k}$ and $\mM_{m,k}$ are shown in Eqs.~(\ref{eq:HeisenbergOBC}) and (\ref{eq:spinwavenorms}) respectively.
The corresponding eigenstates of $\hmP_{t-J_z|\text{comp}}$ are of the form $\ket{\lambda_k; \tau^{(1)} \cdots \tau^{(m)}}$, which can be explicitly written down using Eq.~(\ref{eq:tJzexc}).
To compute the weight in Eq.~(\ref{eq:overlapexpression}), we first compute
\begin{align}
	G^{\alpha, \beta}_{j, k} &\defn \sumal{j_1 < \cdots < j_{\alpha+\beta+1}}{}{c^{\lambda_k}_{j_1, \cdots, j_{\alpha+\beta+1}} \delta_{j_{\alpha+1},  j}} =\frac{1}{\sqrt{\mM_{\alpha+\beta+1,k}}} \sumal{1\leq j_1 < \cdots < j_{\alpha}\leq j-1}{}{\left[\sumal{j+1\leq j_{\alpha+2} < \cdots < j_{\alpha+\beta+1}\leq L}{}{ \left(\sumal{\ell = 1}{\alpha}{c_{j_\ell, k}} + c_{j,k} + \sumal{\ell = \alpha + 2}{\alpha + \beta + 1}{c_{j_\ell,  k}}\right)}\right]}\nn \\
	&= \frac{1}{\sqrt{\mM_{\alpha+\beta+1,k}}}\left[F^{1, j-1}_{\alpha, k} \binom{L -j}{\beta} + c_{j,k} \binom{j - 1}{\alpha} \binom{L - j}{\beta} + \binom{j-1}{\alpha} F^{j+1,L}_{\beta, k}\right], 
\label{eq:Cexpr}
\end{align}
where we have defined
\begin{equation}
F^{l, r}_{m, k} \defn \sumal{l \leq j_1 < \cdots < j_m \leq r}{}{\left(\sumal{\ell = 1}{m}{c_{j_{\ell}, k}}\right)} = \binom{r - l}{m-1} \sumal{j_\ell = l}{r}{c_{j_{\ell}, k}}.
\label{eq:Fexpr}
\end{equation}
Using Eqs.~(\ref{eq:Cexpr}), (\ref{eq:Fexpr}), and the OBC expression for $c_{j,k}$ in Eq.~(\ref{eq:HeisenbergOBC}), we obtain for $k \neq 0$
\begin{align}
    F^{l, r}_{m,k} &= \sqrt{\frac{2}{L}} \times \frac{\sin(k r) - \sin(k (l-1))}{2 \sin(\frac{k}{2})} \times \binom{r - l}{m-1},\nn \\
	G^{\alpha, \beta}_{j, k} &= \frac{\sqrt{2} \binom{j-1}{\alpha}\binom{L-j}{\beta}}{\sqrt{L \times \mM_{\alpha+\beta+1,k}}}\times \left[\frac{\alpha}{2(j-1)}\frac{\sin[k(j-1)]}{\sin(\frac{k}{2})} + \cos\left[k(j-\frac{1}{2})\right] +\frac{\beta}{2(L-j)}\frac{\sin(kL) - \sin(kj)}{\sin(\frac{k}{2})}\right], 
\label{eq:FGexpr}
\end{align}
while for $k=0$ we have
\begin{equation}
    F^{l, r}_{m,k=0} = \frac{\sqrt{2}}{\sqrt{L}} (r-l+1) \times \binom{r - l}{m-1},\;\;\;G^{\alpha, \beta}_{j, k=0} = \frac{\sqrt{2} \binom{j-1}{\alpha} \binom{L-j}{\beta}}{\sqrt{L \times \mM_{\alpha+\beta+1,k=0}}} \times \left[\alpha + 1 + \beta \right]. 
\label{eq:FGexprk0}
\end{equation}
For simplicity, we henceforth assume $L$ is odd and $j = \frac{L+1}{2}$.
We then set $k = \frac{n\pi}{L}$, and compute the total weight of the state $\sket{\zeta_{\frac{L+1}{2}}}$ on all eigenstates of the form $\ket{\lambda_k; \rt^\alpha \lt \rt^\beta}$ for all values of $\alpha$ and $\beta$.
Using Eq.~(\ref{eq:overlapexpression}), this weight is
\begin{align}
    W_{S^z_{\frac{L+1}{2}}}(k = \frac{n \pi}{L}) &\defn \frac{1}{3^L}\sumal{\alpha = 0}{\frac{L-1}{2}}{\sumal{\beta = 0}{\frac{L-1}{2}}{2^{\alpha+\beta+1}|G^{\alpha, \beta}_{\frac{L+1}{2}, k}|^2}} \nn \\
    &= \sumal{\alpha = 0}{\frac{L-1}{2}}{\sumal{\beta = 0}{\frac{L-1}{2}}{\frac{2^{\alpha + \beta + 1} \binom{\frac{L-1}{2}}{\alpha}^2\binom{\frac{L-1}{2}}{\beta}^2}{3^L \times \binom{L}{\alpha+\beta+1}}\times   \threepartdef{1}{n = 0}{\frac{2(L - 1 - \alpha-\beta)}{(L-1)(\alpha+\beta+ 1)}}{n\neq 0\ \text{even}}{\frac{2(\alpha-\beta)^2}{(L-1)(L - \alpha-\beta - 1)(\alpha+\beta +1)}\cot^2(\frac{k}{2})}{n\ \text{odd}}}}, 
\label{eq:MZjtJz}
\end{align}
where we have used Eq.~(\ref{eq:FGexpr}) and the normalization factors of Eq.~(\ref{eq:spinwavenorms}). 
Note that $W_{S^z_{\frac{L+1}{2}}}(k = 0)$ is the same as the Mazur bound computed in \cite{moudgalya2021hilbert}, done here in the composite spin language.
The expression Eq.~(\ref{eq:MZjtJz}) for general $k$ can be analyzed in detail using a saddle point analysis for large $L$, similar to the calculation for the Mazur bound demonstrated in \cite{moudgalya2021hilbert}, but for our purposes it is sufficient to schematically extract the $L$-dependence.
To obtain this, we substitute $\alpha = L p$ and $\beta = L q$ to convert the sums into integrals over $p$ and $q$. 
In App.~G of Ref.~\cite{moudgalya2021hilbert}, the expression for the Mazur bound $W_{S^z_{\frac{L+1}{2}}}(k = 0)$ was shown to be of the form (the computation was done there for a general $x = j/L$, whereas here we will set $x = 1/2$ and remove it from the arguments of the functions involved) 
\begin{equation}
    W_{S^z_{\frac{L+1}{2}}}(k = 0) = \sqrt{L}\int_0^{\frac{1}{2}}\int_0^{\frac{1}{2}}{dp\ dq\ C(p,q)\exp\left(L F(p,q)\right)} \approx \frac{2 \pi\ C(p_0, q_0)}{\sqrt{L\ \det H(p_0, q_0)}}\exp\left(L F(p_0,q_0)\right),
\label{eq:mazurfullintegral}
\end{equation}
where $C(p, q)$ and $F(p, q)$ are some $L$-independent functions, and a saddle point approximation has been performed in the second step, which we unpack below.
$H(p,q)$ is the Hessian of $F(p, q)$, and the ``saddle" is given by the point where $(\frac{\partial F}{\partial p}, \frac{\partial F}{\partial q})|_{(p, q) = (p_0, q_0)} = (0, 0)$, which turns out to be at $(p_0, q_0) = (\frac{1}{3}, \frac{1}{3})$.
It also turns out that $F(p_0, q_0) = 0$, hence we get the Mazur bound scaling of $\sim L^{-\frac{1}{2}}$~\cite{moudgalya2021hilbert}.
Note that both $H(p, q)$ and $(p_0, q_0)$ are completely determined by $F(p, q)$.   
Since the $n \neq 0$ expressions in Eq.~(\ref{eq:MZjtJz}) differ from the $n = 0$ case only by factors polynomial in $L$, we can express them in terms of $p$ and $q$, and analogously write down the form of the leading order terms in the saddle point approximation:
\begin{equation}
    W_{S^z_{\frac{L+1}{2}}}(k = \frac{n\pi}{L}) \approx \frac{2 \pi\ C(p_0, q_0)}{\sqrt{L\ \det H(p_0, q_0)}}\exp\left(L F(p_0,q_0)\right) \times \twopartdef{\frac{2 (1 - p_0 - q_0)}{L(p_0+q_0)}}{n \neq 0\ \text{even}}{\frac{2(p_0 - q_0)^2}{L (1 - p_0 - q_0)(p_0+q_0)}\cot^2(\frac{k}{2})}{n\ \text{odd}},
\label{eq:saddleapprox}
\end{equation}
where the saddle $(p_0, q_0)$ is unchanged since the function $F(p, q)$ in the exponent is unchanged from the $n = 0$ case. 
We can then use Eq.~(\ref{eq:saddleapprox}) with the fact that $F(p_0, q_0) = 0$ to determine the scaling of $W_{S^z_{\frac{L+1}{2}}}(k)$. 
For $n\neq 0$ even, we get a scaling of $\sim L^{-\frac{3}{2}}$, and adding the contributions over all even $n$, we get a total scaling of $\sim L^{-\frac{1}{2}}$.
For $n$ odd, since $p_0 = q_0$, the leading order term shown in Eq.~(\ref{eq:saddleapprox}) vanishes.
Since the subleading terms in the saddle point approximation are suppressed by a factor of $L$, we obtain a scaling of $\sim L^{-\frac{5}{2}} \cot^2(\frac{k}{2})$.
Adding these contributions over all odd $n = 2l+1$, we get $\sim L^{-\frac{5}{2}}\sum_{l = 0}^{\frac{L-1}{2}}{\cot^2[(l + \frac{1}{2})\frac{\pi}{L}]} \sim L^{-\frac{1}{2}}$, since it is dominated by a few values of $k$ close to $0$, where $\cot^2({\frac{k}{2}}) \sim k^{-2}$.
In all, the total weight of the operator $S^z_{\frac{L+1}{2}}$ on the ground states and the ``single-magnon'' spin-wave states scales as $\sim L^{-\frac{1}{2}}$, which is vanishing in the thermodynamic limit. 
This completes the demonstration that the contribution from these excitations to $\Omega_{S_j^z}(E)$ of Eq.~(\ref{eq:CAweight}) vanishes when $L \to \infty$, and more complicated excitations need to be considered to understand its form, which at present we have only done numerically.
\section{Asymptotic QMBS in Brownian Circuits}\label{app:QMBSbrownian}
In this appendix, we discuss asymptotic QMBS in Brownian circuits with exact QMBS.
We start with the assumption that there is a set of local projectors $\{R_{[j]}\}$ such that the common kernel of these projectors is spanned by the exact QMBS $\{\ket{\Phi_n}\}$; this is sometimes referred to as the Shiraishi-Mori condition~\cite{shiraishi2017systematic, moudgalya2022exhaustive}. 
In other words, the subspace spanned by these QMBS states can be expressed as the exhaustive ground state space of a frustration-free Hamiltonian, i.e., 
\begin{equation}
     \sum_j{R_{[j]}}\ket{\psi} = 0 \;\;\iff\;\;\ket{\psi} \in \mS = \text{span}\{\ket{\Phi_n}\}. 
\label{eq:exhaustive}
\end{equation}
Several examples of QMBS, including those in the spin-1 XY model~\cite{schecter2019weak}, Hubbard model and its deformations~\cite{mark2020eta, moudgalya2020eta}, and also those in the spin-1 AKLT model~\cite{moudgalya2018exact} can be understood this way~\cite{moudgalya2022exhaustive, moudgalya2021review, chandran2022review}. 
With this, we can generically write down a bond algebra $\mA_{\scar}$ of the form of Eq.~(\ref{eq:ACtower}), which has generators of the form $\{R_{[j]}\sigma^\alpha_k\}$, such that its centralizer is $\mC_{\scar}$ spanned by $\{\ketbra{\Phi_m}{\Phi_n}\}$ (see footnote \ref{ftn:scarbond} for precise statement).
This structure guarantees that the QMBS are degenerate eigenstates of all operators constructed out of the generators of $\mA_{\scar}$.
Here we consider Brownian circuits built out of the generators of $\mA_{\scar}$, and directly work with the evolution of states under this circuit, as opposed to operators discussed in Sec.~\ref{subsec:algebrabrownian}. 
Exact QMBS are stationary states under such circuits, since they are by definition eigenstates of each ``gate" of the circuit.
Here we show that working with states directly also shows the existence of asymptotic QMBS that relax slowly in such circuits. 
Denoting the generators of the bond algebra to be $\{\hH_\alpha\}$, the time-evolution of a state $\ket{\psi(t)}$ by a time-step $\dt$ can be written in direct analogy to Eq.~(\ref{eq:opevolutionliouv}) as
\begin{equation}
\ket{\psi(t + \dt)} = e^{-i \sum_{\alpha}{J_\alpha^{(t)} \hH_\alpha \dt}} \ket{\psi(t)} = \ket{\psi(t)} - i \dt \sumal{\alpha}{}{J_\alpha^{(t)} \hH_\alpha \ket{\psi(t)}} - \frac{(\dt)^2}{2} \sumal{\alpha,\beta}{}{J_\alpha^{(t)} J_\beta^{(t)} \hH_\alpha \hH_\beta \ket{\psi(t)}} + \mO((\dt)^3).
\label{eq:stateevolutionbrown}
\end{equation}
After ensemble-averaging with the distributions of Eq.~(\ref{eq:Jalphapdf}), the expression for the state in the continuum time limit reads, in analogy to Eq.~(\ref{eq:Oavgcontinuum}),  
\begin{equation}
    \overline{\ket{\psi(t)}} = e^{-\kappa \hP t}\ket{\psi(0)},\;\;\;\implies\;\;\; \overline{\braket{\psi(0)}{\psi(t)}} = \bra{\psi(0)}e^{-\kappa\hP t}\ket{\psi(0)},\;\;\;\hP \defn \sum_\alpha{\hH_\alpha^2}.
\label{eq:overlapdecay}
\end{equation}
As a consequence, the decay of the ensemble-averaged overlap is governed by the spectrum of $\hP$.
While we are not able to directly compute the ensemble-averaged fidelity using this approach, it can be lower-bounded in terms of the overlap as
\begin{equation}
\overline{\mF(t)} = \overline{|\braket{\psi(0)}{\psi(t)}|^2} \geq \left|\overline{\braket{\psi(0)}{\psi(t)}}\right|^2.
\label{eq:fidelitybound}
\end{equation}
We now restrict to specific examples of QMBS, where bond generators are of the form of Eq.~(\ref{eq:QMBSalgebra}). 
Then we have 
\begin{equation}
    \hP = \sum_{j, k,\alpha}{(R_{[j]}\sigma^\alpha_k)^2} = C \sum_j{R_{[j]}},
\label{eq:hPexpression}
\end{equation}
where $C$ is an overall constant that depends on the number of $\alpha$'s and $k$'s in the generators (assumed to be the same for each $j$ for simplicity), and we have used the fact that $R_{[j]}$ is a projector.
Strikingly, one can see that this is precisely the frustration-free Hamiltonian that appeared in Eq.~(\ref{eq:exhaustive}).
This already shows that exact QMBS never decay, since they are ground states of $\hP$. 
The slowly-relaxing states, or the asymptotic QMBS, are then the low-energy excitations of $\hP$, provided it is gapless.
In the case of the ferromagnetic tower of QMBS discussed in Sec.~\ref{subsubsec:towerofQMBS}, we have $R_{[j]} = \frac{1}{4} - \vec{S}_j\cdot\vec{S}_{j+1}$, hence $\hP$ is just the ferromagnetic Heisenberg model of Eq.~(\ref{eq:heisenbergcanon}) up to an overall factor.
The asymptotic QMBS $\{\ket{\Phi_{n,k}}\}$ are then simply spin-waves on top of the ferromagnet shown in Eq.~(\ref{eq:Heisenbergspinwavestates}); this explains their form in Eq.~(\ref{eq:aQMBS}). 
Since these states with small $k$ have energy $\sim p_k \sim k^2$ under $\hP$, their ensemble-averaged overlap $\overline{\braket{\Phi_{n,k}(0)}{\Phi_{n,k}(t)}}$ decays on timescales $\sim L^2$, which due to Eq.~(\ref{eq:fidelitybound}) is also a lower-bound for the timescale for the fidelity, consistent with Eq.~(\ref{eq:overlapavg}).
This method hence more directly reproduces the asymptotic QMBS found from the super-Hamiltonian perspective in Sec.~\ref{subsec:asymptoticQMBS}, and explains the significance of the corresponding super-Hamiltonian eigenstates, i.e., those associated with spin-waves on only one leg of the ladder that appeared there.
The appearance of the frustration-free Hamiltonian of Eq.~(\ref{eq:exhaustive}) in Eq.~(\ref{eq:hPexpression}), and the physical interpretation of its eigenstates as the decay modes of the overlap in Eq.~(\ref{eq:overlapdecay}), leads us to the conjecture \ref{con:aqmbs} on the conditions for the existence of asymptotic QMBS, restated here for clarity.
\aqmbs*
One might also notice that the form of the decay of the overlap of the asymptotic QMBS obtained using  Eq.~(\ref{eq:overlapdecay}) is necessarily a simple exponential of the form $\exp(-c t/L^2)$, where the asymptotic QMBS is an eigenstate of $\hP$ with eigenvalue $\sim k^2 \sim c/L^2$. 
Since the ensemble-averaged fidelity is lower-bounded by this [see Eq.~(\ref{eq:fidelityauto})], this predicts a fidelity decay timescale that scales as $\sim L^2$. 
This is different from the fidelity decay timescale of asymptotic QMBS predicted and observed in Hamiltonian systems in~\cite{gotta2023asymptotic}.
The fidelity of an initial state under Hamiltonian evolution at short times is of form $\exp(-\Delta H^2 t^2)$~\cite{camposvenuti2010unitary}, where $\Delta H^2$ is the variance of the energy in the initial state, $\Delta H^2 \equiv \expval{H^2}{\psi_0} - (\expval{H}{\psi_0})^2$.
Given that the variance of asymptotic QMBS scales as $\Delta H^2\sim 1/L^2$~\cite{gotta2023asymptotic}, the fidelity decay is of the form $\exp(-c' t^2/L^2)$, which predicts a decay timescale $\sim L$. 
The fidelity in the Brownian circuits hence decays parametrically slower than in the Hamiltonian evolution. 
This is reminiscent of the quantum Zeno effect, where unitary evolution is suppressed by external factors such as repeated measurements or fast-fluctuating stochasticity. 
It would be interesting to make this connection precise in future work, while here we can give a rough argument showing reconcilability of these results, which also sheds some light on quantitative relations between the Brownian circuit and Hamiltonian dynamics.
Note that the derivation of the ensemble-averaged state dynamics from Eq.~(\ref{eq:stateevolutionbrown}) to Eq.~(\ref{eq:overlapdecay}) formally required taking $\Delta t \to 0$ limit while taking the variance of the couplings $J_\alpha^{(t)}$ to diverge as $\sigma_J^2 = 2\kappa/\Delta t$ at fixed $\kappa$.
We can in fact use the obtained results also in the circuit setups where $\sigma_J^2$ is kept fixed while we take $\Delta t$ sufficiently small, which, however, then enters the characteristic rate $\kappa$ in all results: $\kappa = \sigma_J^2 \Delta t/2$.
This already shows that if the applied Hamiltonians have typical couplings of a given strength $\sim \sigma_J$ and hence typical dynamics rates $\sim \sigma_J$, changing the Hamiltonian randomly after every small time interval $\dt$ suppresses the dynamics rates to $\sim \kappa \sim \sigma_J \cdot \sigma_J \Delta t$, assuming $\sigma_J \dt \ll 1$.
For an initial state $\ket{\psi(0)}$ that is an eigenstate of $\hP$ of Eq.~(\ref{eq:overlapdecay}) with a bounded eigenvalue $p$: $\hP \ket{\psi(0)} = p \ket{\psi(0)}$, the ensemble-averaged Eq.~(\ref{eq:stateevolutionbrown}) gives
\begin{equation}
\overline{\ket{\psi(\Delta t)}} \approx \left(1 - \frac{(\Delta t)^2}{2} \sigma_J^2 p \right) \ket{\psi(0)} \quad \implies \quad \overline{\ket{\psi(t)}} \approx \left(1 - \frac{(\Delta t)^2}{2} \sigma_J^2 p \right)^{\frac{t}{\Delta t}} \ket{\psi(0)} \approx e^{-\frac{1}{2} \sigma_J^2 (\Delta t) p t} \ket{\psi(0)}, 
\label{eq:argument4finiteDt}
\end{equation}
which matches the result Eq.~(\ref{eq:overlapdecay}) with $\kappa = \frac{1}{2} \sigma_J^2 \Delta t$ as claimed, and the approximations used are controlled as long as $\frac{(\Delta t)^2}{2} \sigma_J^2 p \ll 1$.
This is also roughly consistent with the expected fidelity decay under a fixed Hamiltonian, controlled by its variance $\Delta H^2$ in the state~\cite{camposvenuti2010unitary}, which we expect to apply for individual evolution steps over time $\Delta t$:
\begin{equation}
|\braket{\psi(0)}{\psi(\Delta t)}|^2 = e^{-\Delta H^2 (\Delta t)^2} ~, \quad \Delta H^2 \defn \bra{\psi(0)} H^2 \ket{\psi(0)} - \bra{\psi(0)} H \ket{\psi(0)}^2 ~. 
\end{equation}
With $H = \sum_\alpha J_\alpha^{(t)} H_\alpha$, an elementary calculation averaging over the independent Gaussian $J_\alpha^{(t)}$ with variance $\sigma_J^2$ gives $\overline{\expval{H^2}{\psi_0}} = \sigma_J^2 \expval{\hP}{\psi_0} = \sigma_J^2 p$ if $\hP\ket{\psi_0} = p\ket{\psi_0}$.
This upper-bounds the ensemble-averaged variance of the energy, $\overline{\Delta H^2}$, and can be viewed as a reasonable estimate of a typical value of the variance of the energy in $\ket{\psi_0}$ for Hamiltonians drawn from this distribution.
We can then recover the qualitative form of Eq.~(\ref{eq:argument4finiteDt}) if we conjecture that the result of the multiple time steps of the Brownian circuit is to have roughly the same fidelity suppression factor for each step $\Delta t$: in this case the total suppression factor is $e^{-\Delta H_\text{typ}^2 (\Delta t)^2 \cdot \frac{t}{\Delta t}}$ qualitatively matching Eq.~(\ref{eq:argument4finiteDt}).
While this is not a precise quantitative argument, we can justify using the same suppression factor after each time step $\Delta t$ as follows.
At any given instance of the Brownian circuit, we have
\begin{equation}
\ket{\psi(\Delta t)} = \ket{\psi(0)} \braket{\psi(0)}{\psi(\Delta t)} + \sket{\delta\psi^\perp} ~,
\end{equation}
where the first term is the projection onto $\ket{\psi(0)}$ while $\sket{\delta\psi^\perp}$ is the deviation.
%
While the asymptotic QMBS property of $\ket{\psi(0)}$ is common for all $H^{(t)}$ (i.e., is essentially non-random across them), the deviation $\sket{\delta\psi^\perp}$ is particular to the Hamiltonian applied at that step, which is hence random across different steps, and it is plausible that the chances of the $\ket{\delta\psi^\perp}$ part to return close to $\ket{\psi(0)}$ under the subsequent steps are small.
For the purposes of calculating the fidelity we are only interested in keeping track of the $\ket{\psi(0)}$ projection, and we get a similar suppression factor at each step.
Finally, note that the presentation here focused on the asymptotic QMBS appearing due to exact towers of QMBS, where the $p \sim 1/L^2$ scaling of the super-Hamiltonian energy of an initial state implies the divergence of its fidelity decay time.
Mathematically the same arguments also go through in the case of exact isolated QMBS such as those in  Sec.~\ref{subsubsec:gappedisolatedQMBS}, where the super-Hamiltonian energy any initial state $\ket{\psi(0)}$ orthogonal to the exact QMBS is at least a constant, implying a constant fidelity decay time.
Nevertheless, this still signifies some hidden ``non-thermalness'' in the ``thermal" sector that occurs even due to a single exact scar state~\cite{lin2019exact,lin2020slow}, although it is less dramatic than in the asymptotic QMBS.
It would be interesting to understand whether this framework can be used to quantify such non-thermalness in more detail.
\section{Super-Hamiltonians for Equally Spaced Tower of QMBS}\label{app:QMBStower}
We now discuss different choices of super-Hamiltonians suitable for analyzing cases with a tower of QMBS,  and we show that the physics of asymptotic QMBS is captured in this setup. 
As discussed in \cite{moudgalya2022exhaustive}, we can describe such a tower of exact ferromagnetic scar states split in energy by the Zeeman field term using the algebras
\begin{equation}
    \mA_{\text{tower-lift}} \defn \lgen \{R_{j,j+1}\sigma^\alpha_{j+2}, \sigma^\alpha_{j-1} R_{j,j+1}\}, \sum_j{\sigma^z_j}\rgen,\;\;\;\mC_{\text{tower-lift}} = \lgen \{\ketbra{\Phi_{n,0}}\}\rgen.
\label{eq:ACtowerlift}
\end{equation}
This seems to necessarily require the addition of an extensive-local term $Z_{\tot} \defn \sum_j{\sigma^z_j}$ to the generators of $\mA_{\text{tower}}$.
In order to capture the commutant $\mC_{\text{tower-lift}}$ as the ground states of a super-Hamiltonian, we need to sacrifice either locality or Hermiticity of the super-Hamiltonian, and we discuss both options below.
\subsection{Non-Local Super-Hamiltonian}
We can naively follow the discussion in Sec.~\ref{subsec:liouvsuper} and construct a super-Hamiltonian using Eq.~(\ref{eq:psdLiouv}).
While the resulting super-Hamiltonian $\hmP_{\text{tower-lift}}$ is Hermitian, it becomes non-local, with the addition of the term $[\sum_j{(\sigma^z_{j;t} - \sigma^z_{j;b})}]^2$.
This term changes the energy of the ground states $\ket{\Phi_{m,0}} \otimes \ket{\Phi_{n,0}}$ of $\hmP_{\text{tower}}$ from $0$ to $(m - n)^{2}$.
Hence the ground states of the new super-Hamiltonian $\hmP_{\text{tower-lift}}$ with this non-local term now require $m = n$ and are hence $\{\ket{\Phi_{n,0}}_t \otimes \ket{\Phi_{n,0}}_b\}$, which shows that the above algebras are indeed centralizers of each other.
The operators $\ketbra{\Phi_{n,k}}{\Phi_{n,0}}$ are still exact low-energy eigenstates of $\hmP_{\text{tower-lift}}$ with eigenvalue $p_k = 8[1 - \cos(k)]$, which allows us to understand the asymptotic QMBS.
Equation~(\ref{eq:overlapavg}) holds for $\hA = \ketbra{\Phi_{n,k}}{\Phi_{n,0}}$ since $\ket{\Phi_{n,0}}$ is still a singlet of $\mA_{\text{tower-lift}}$,
which rigorously lower-bounds the fidelity decay time in the case of non-degenerate towers as well. 
Note that this quantity is not straightforward to bound from the direct consideration of the dynamics of states discussed in App.~\ref{app:QMBSbrownian}.
The effective Hamiltonian $\hP$ shown in Eq.~(\ref{eq:overlapdecay}) acquires an additional term $(\sum_j{\sigma^z_j})^2$, which then shows that
\begin{equation}
    \overline{\braket{\Phi_{n,k}(0)}{\Phi_{n,k}(t)}} = e^{-\kappa [p_k + (L - 2 n)^2] t},  
\end{equation}
which decays fast when $n \neq \frac{L}{2}$, i.e., when the eigenvalue $Z_\text{tot} = L - 2n \neq 0$. 
This is also the case for the overlap $\overline{\braket{\Phi_{n,0}(0)}{\Phi_{n,0}(t)}}$, even though $\ket{\Phi_{n,0}}$ is an exact QMBS, and it is an effect of averaging over the random phases acquired by the action of $Z_{\tot}$.
While these are mathematically correct properties of the Brownian circuit with random fluctuating Zeeman field, they are not useful for understanding the fidelity properties of the asymptotic QMBS.
\subsection{Non-Hermitian Super-Hamiltonian}
We now discuss an alternate super-Hamiltonian for the algebra of Eq.~(\ref{eq:ACtowerlift}) that preserves locality but sacrifices Hermiticity.
This naturally appears in a Brownian circuit where the coefficient of the magnetic field is constant and \textit{not} random.
We can then redo the analysis of Eqs.~(\ref{eq:opevolutionliouv})-(\ref{eq:Oavgcontinuum}) to derive an effective Hamiltonian that describes the ensemble-averaged operator dynamics.
Given a Brownian circuit evolving under a set of operators $\{\hH_\alpha\}$ with random coefficients $\{J^{(t)}_\alpha\}$ chosen from the distribution of Eq.~(\ref{eq:Jalphapdf}), and an operator $\hG$ with a constant $\mathcal{O}(1)$ coefficient $K$, analogous to Eq.~(\ref{eq:Oavgfinal}) we obtain
\begin{equation}
	 \oket{\overline{\hO(t + \dt)}} =  \oket{\overline{\hO(t)}} + i \dt K \hmL_{\hG}\oket{\overline{\hO(t)}} - \dt \sumal{\alpha}{}{\kappa_\alpha \hmL^\dagger_\alpha \hmL_\alpha \oket{\overline{\hO(t)}}} + \mO((\dt)^2),
\label{eq:Oavgconst}
\end{equation}
where $\hmL_{\hG} \defn \hG_{t}\otimes \mathbb{1}_b - \mathbb{1}_t \otimes \hG_b$ is the Liouvillian corresponding to $\hG$ (assumed to be Hermitian and real-valued in the working basis for simplicity).
In the continuous time limit we then obtain
\begin{equation}
	\frac{d}{dt} \oket{\overline{\hO(t)}} = -\Big[\sumal{\alpha}{}{\kappa_\alpha \hmL^\dagger_\alpha\hmL_\alpha} - i K \hmL_{\hG} \Big] \oket{\overline{\hO(t)}},\nn \\
 \implies\;\oket{\overline{\hO(t)}} = e^{- (\kappa \hmP - i K \hmL_{\hG}) t}\oket{\hO(0)}.
\label{eq:Oavgcontinuumconst}
\end{equation}
Hence the physics of this system can be understood using the non-Hermitian super-Hamiltonian $\hmP_{\text{n-h}} = \kappa \hmP - i K \hmL_{\hG}$.
Operators in the commutant $\mC_{\text{ext-loc}}$ of the algebra $\mA_{\text{ext-loc}} = \lgen \{\hH_\alpha\}, \hG \rgen$ are guaranteed to have zero eigenvalue under $\hmP_{\text{n-h}}$. 
Moreover, any zero eigenvalued operator of $\hmP_{\text{n-h}}$ is in the commutant $\mC_{\text{ext-loc}}$, a simple proof is as follows.
Starting from $\hmP_{\text{n-h}}\oket{\Psi} = 0$, we have $\obra{\Psi} \hmP_{\text{n-h}} \oket{\Psi} = \kappa \obra{\Psi} \hmP \oket{\Psi} - i K \obra{\Psi} \hmL_{\hG} \oket{\Psi} = 0$, which, since $\hmP$ and $\hmL_G$ are Hermitian, means $\obra{\Psi} \hmP \oket{\Psi} = 0$ and $\obra{\Psi} \hmL_G \oket{\Psi} = 0$.
Since $\hmP$ is positive semi-definite, we can conclude that $\hmP\oket{\Psi} = 0$, which also means that $\hmL_G\oket{\Psi}=0$, showing that $\oket{\Psi}$ is in the commutant $\mC_{\text{ext-loc}}$.
Applying this to asymptotic QMBS, where $\{\hH_\alpha\}$ are the generators of $\mA_{\text{tower}}$ of Eq.~(\ref{eq:ACtower}) and $\hG = \sum_j{\sigma^z_j}$, we have $\hmL_{\hG} = \sum_j{(\sigma^z_{j;t} - \sigma^z_{j;b})}$. 
Since $\hmP$ has the form of the dissipator of a Lindblad master equation [see Eq.~(\ref{eq:lindbladian})] with jump operators $\{\hH_\alpha\}$, the full $\hmP_{\nh}$ has the form of a full Lindbladian with the Hamiltonian $\hG$ and the jump operators $\{\hH_\alpha\}$.
Hence the eigenvalues of $\hmP_{\nh}$ are guaranteed to have non-negative real parts. 
It is easy to verify that the eigenstates $\ket{\Phi_{n,k}}_t \otimes \ket{\Phi_{m,0}}_b$ of $\hmP$, discussed in Eq.~(\ref{eq:aQMBSeigenstate}), continue to be eigenstates of $\hmP_{\nh}$, and we have
\begin{equation}
    \hmP_{\nh}\oket{\ketbra{\Phi_{n,k}}{\Phi_{m,0}}} = [\kappa p_k + 2 i  K (n-m) ]\oket{\ketbra{\Phi_{n,k}}{\Phi_{m,0}}}, 
\end{equation}
where $p_k = 8[1 - \cos(k)]$ and we have used that $\ket{\Phi_{n,0}}$ and $\ket{\Phi_{n,k}}$ are eigenstates of
$Z_\text{tot}$ with eigenvalue $L-2n$. 
We can then follow the same arguments as in Sec.~\ref{subsec:asymptoticQMBS} to obtain exact results -- lower bounds -- for the ensemble-averaged fidelity, by relating it to the autocorrelation function of an operator $\hA = \ketbra{\Phi_{n,k}}{\Phi_{m,0}}$, which leads to Eq.~(\ref{eq:overlapavg}), consistent with the expected slow decay of asymptotic QMBS.
\end{document}